\documentclass{statsoc}

\usepackage[final]{graphicx}
\usepackage{amsmath,amsfonts,amssymb}
\usepackage{verbatim}
\usepackage{color}
\usepackage{epsfig}
\usepackage[draft=false,letterpaper,breaklinks,colorlinks,linktocpage,citecolor=blue,linkcolor=blue,urlcolor=blue]{hyperref}

\newcommand{\PSigma} {\Pi_{\Sigma}}
\newcommand{\Pxi} {\Pi_{\xi}}
\newcommand{\PTheta} {\Pi_{\Theta}}
\newcommand{\PSigmaO} {\Pi_{\Sigma_0}}

\newcommand{\figdir}{figs} 

\def\postcap{\vspace{-0.275in}}

\newtheorem{assumption}{Assumption}[section]

\newtheorem{theorem}{Theorem}[section]

\newtheorem{lemma}{Lemma}[section]

\title{Bayesian Nonparametric Covariance Regression}
\author{Emily Fox}
\address{Duke University,
Durham, NC
USA.}
\email{fox@stat.duke.edu}
\author[Fox and Dunson]{David Dunson}
\address{Duke University,
Durham, NC
USA.}
\email{dunson@stat.duke.edu}

\begin{document}

\ifpdf
	\DeclareGraphicsExtensions{.pdf, .jpg, .tif}
	\else
	\DeclareGraphicsExtensions{.eps, .jpg}
	\fi

\begin{abstract}
Although there is a rich literature on methods for allowing the variance in a univariate regression model to vary with predictors, time and other factors, relatively little has been done in the multivariate case. Our focus is on developing a class of nonparametric covariance regression models, which allow an unknown $p \times p$ covariance matrix to change flexibly with predictors. The proposed modeling framework induces a prior on a collection of covariance matrices indexed by predictors through priors for predictor-dependent loadings matrices in a factor model. In particular, the predictor-dependent loadings are characterized as a sparse combination of a collection of unknown dictionary functions (e.g, Gaussian process random functions). The induced covariance is then a regularized quadratic function of these dictionary elements. Our proposed framework leads to a highly-flexible, but computationally tractable formulation with simple conjugate posterior updates that can readily handle missing data. Theoretical properties are discussed and the methods are illustrated through simulations studies and an application to the Google Flu Trends data.
\end{abstract}


{\small \textbf{\emph{Keywords: } covariance estimation; Gaussian process; heteroscedastic regression; nonparametric Bayes; stochastic process.}}

\section{Introduction}
\label{sec:intro}

Spurred by the increasing prevalence of high-dimensional datasets and the computational capacity to analyze them, capturing heteroscedasticity in multivariate processes has become a growing focus in many applied domains.  For example, within the field of financial time series modeling, capturing the time-varying \emph{volatility} and \emph{co-volatility} of a collection of risky assets is key in devising a portfolio management scheme.  Likewise, the spatial statistics community is often faced with multivariate measurements (e.g., temperature, precipitation, etc.) recorded at a large collection of locations, necessitating methodology to model the strong spatial (and spatio-temporal) variations in correlations.  More generally, imagine that one has some arbitrary, potentially multivariate predictor space $\mathcal{X}$ and a collection of multivariate response vectors $y$.  The problem of mean regression (i.e., $\mu(x) = E(y \mid x)$) has been well studied in both the univariate and multivariate settings. Although there is a rich literature on methods for allowing the variance in a univariate regression model to vary with predictors, there is a dearth of methodologies for the general case of multivariate covariance regression (i.e., $\Sigma(x) = \mbox{cov}(y \mid x)$).  The covariance matrix captures key correlations between the elements of the response vector, and the typical assumption of a homoscedastic model can have significant impact on inferences.  

Historically, the problem of multivariate covariance regression has typically been addressed by standard regression operations on the unconstrained elements of the log or Cholesky decomposition of the covariance (or precision) matrix.  For example, \cite{Pourahmadi:99} proposes to model elements of $\mbox{chol}(\Sigma(x)^{-1})$ as a linear function of the predictors.  The weights associated with the $i$th row have a nice interpretation in terms of the conditional distribution of $y_i$ given $y_1,y_2,\dots,y_{i-1}$; however, the model is not invariant to permutations of the elements of $y$ which is problematic in applications where there does not exist a natural ordering.  Alternatively,~\cite{Chiu:96} consider modeling each element of $\log(\Sigma(x))$ as a linear function of the predictor.  An issue with this formulation is in the interpretability of the model: a submatrix of $\Sigma(x)$ does not necessarily coincide with a submatrix of the matrix logarithm.  Additionally, both the models of \cite{Pourahmadi:99} and \cite{Chiu:96} involve a large number of parameters (specifically, $d\times p(p+1)/2$ assuming $x\in \Re^d$.)  More recently,~\cite{Hoff:10} propose a covariance regression model in which $\Sigma(x) = A + Bxx'B'$ with $A$ positive definite and $B$ real. This model has interpretable parameters and may be equated with a latent factor model, leading to computational advantages.  However, the model still has key limitations in (i) scaling to large $p$ domains, and (ii) flexibility based on the parametric approach.  Specifically, the model restricts the difference between $\Sigma(x)$ and the baseline matrix $A$ to be rank 1.  Higher rank models can be considered via extensions such as $\Sigma(x) = A + Bxx'B' + Cxx'C'$, but this dramatically increases the parameterization and requires definition of the maximal rank difference.  

For volatility modeling where the covariate space is typically taken to be discrete time, heteroscedasticity has classically been captured via either variants of ARCH~\citep{Engle:02} or stochastic volatility models~\citep{Harvey:94}.  The former directly specifies the volatility matrix as a linear combination of lagged volatilities and squared returns, which suffers from curse of dimensionality, and is typically limited to datasets with 5 or fewer dimensions. 
Alternatively, multivariate stochastic volatility models assume $\Sigma(t) = A\Gamma(t)A'$, with $A$ real, $\Gamma(t) = \mbox{diag}(\exp h_{1t},\ldots,\exp h_{pt})$, and $h_{it}$ independent autoregressive processes. See~\cite{Chib:09} for a survey of such approaches.  
More recently, a number of papers have examined inducing covariance processes through variants of a Wishart process. \cite{PhilipovGlickman:06} take $\Sigma(t)^{-1}\mid \Sigma(t-1) \sim W(n,S_{t-1})$ with $S_{t-1} = 1/n (A^{1/2})(\Sigma(t-1)^{-1})^\nu(A^{1/2})'$\footnote{Extending to higher dimensions,~\cite{PhilipovGlickman:06b} apply this model to the covariance of a lower-dimensional latent factor in a standard latent factor model.}.  Alternatively, a conditionally \emph{non-central} Wishart distribution is induced on the precision matrix in~\cite{Gourieroux:09} by taking $\Sigma(t) = \sum_{k=1}^p x_{kt}x_{kt}'$, with each $x_{k}$ independently a first order Gaussian autoregressive process.  Key limitations of these types of Wishart processes are that posterior computations are extremely challenging, theory is lacking (e.g., simple descriptions of marginal distributions), and single parameters (e.g., $n$ and $\nu$) control the inter- and intra-temporal covariance relationships.  \cite{PradoWest} review alternative models of time-varying covariance matrices for dynamic linear models via discounting methods that maintain conjugacy.  Central to all of the cited volatility models is the assumption of Markov dynamics, limiting the ability to capture long-range dependencies and often leading to spiky trajectories.  Additionally, these methods assume a regular grid of observations that cannot easily accommodate missing values.

Within the spatial statistics community, the term \emph{Wishart process} is typically used to specify a different formulation than those described herein for volatility modeling.  Specifically, letting $\Sigma(s)$ denote the covariance of a $p$-dimensional observation at geographic location $s \in \mathcal{R}^2$,~\cite{Gelfand:04} assume that $\Sigma(s) = A(s)A(s)' + \Sigma_0$ with $\Sigma_0$ diagonal and $T(s)=A(s)A(s)'$ following a \emph{matric-variate Wishart process}. This Wishart process is such that $T(s)^{-1} = \Theta\xi(s)\xi(s)'\Theta'$ with $\xi_{\ell j} \sim \mbox{GP}(0,c_j)$ independently for each $\ell,j$ and $\Theta$ typically taken to be diagonal.  The induced distribution on $T(s)$ is then marginally inverse Wishart\footnote{More generally,~\cite{Gelfand:04} develop a \emph{spatial} coregionalization model such that $\mbox{cov}(y(s_1),y(s_2)) = \rho(s_1-s_2)A(s_1)A(s_2)' + \Sigma_0$ (i.e., a model with spatial dependencies arising in both the covariance and \emph{cross} covariance).}.  Posterior computations in this model rely on Metropolis-Hastings proposals that do not scale well to dimensions $p$ larger than 2-3 and cannot naturally accommodate missing data.  In terms of spatio-temporal processes,~\cite{Lopes:08} build upon a standard dynamic factor model to develop nonseparable and nonstationary space-time models. Specifically, the vector $y_t$ of univariate observations $y_{st}$ at spatial locations $s$ is modeled as $y_t = \mu_t + \beta f_t + \epsilon_t$ with the components of the latent factors $f_t$ independently evolving according to a first-order autoregressive process and columns of the factor loadings matrix $\beta$ independently drawn from Markov random fields.  Key assumptions of this formulation are that the observations evolve in discrete time on a regular grid, and that the dynamics of the spatio-temporal process are captured by independent random walks on the components of the latent factors.

In this paper, we present a Bayesian nonparametric approach to multivariate covariance regression that allows the covariance matrix to change flexibly with predictors and readily scales to high-dimensional datasets. The proposed modeling framework induces a prior on a collection of covariance matrices $\Sigma_{\mathcal{X}} = \{ \Sigma(x), x \in \mathcal{X} \}$ through specification of a prior on a predictor-dependent latent factor model. In particular, the predictor-dependent loadings are characterized as a sparse combination of a collection of unknown dictionary functions (e.g, Gaussian process random functions). The induced covariance is then a regularized quadratic function of these dictionary elements.  The proposed methodology has numerous advantages over previous approaches.  By employing collections of continuous random functions, we allow for an irregular grid of observations.  Similarly, we can easily cope with missing data within our framework without relying on imputing the missing values.  Another fundamental property of the proposed methodology is the fact that our combined use of a shrinkage prior with a latent factor model enables us (in theory) to handle high-dimensional datasets (e.g., on the order of hundreds of dimensions) in the presence of a limited number of observations.  Essential in being able to cope with such large datasets in practice is the fact that our computations are tractable, based on simple conjugate posterior updates.  Finally, we are able to state theoretical properties of our proposed prior, such as large support.

The paper is organized as follows.  In Section~\ref{sec:covReg}, we describe our proposed Bayesian nonparametric covariance regression model in addition to analyzing the theoretical properties of the model.  Section~\ref{sec:comp} details the Gibbs sampling steps involved in our posterior computations.  Finally, a number of simulation studies are examined in Section~\ref{sec:sim}, with an application to the Google Trends flu dataset presented in Section~\ref{sec:app}.

\section{Covariance Regression Priors}
\label{sec:covReg}

\subsection{Notation and Motivation}

Let $\Sigma(x)$ denote the $p \times p$ covariance matrix at ``location'' $x \in \mathcal{X}$.  In general, $x$ is an arbitrary, possibly multivariate predictor value.  In dynamical modeling, $x$ may simply represent a discrete time index (i.e., $\mathcal{X} = \{1,\ldots, T\}$) or, in spatial modeling, a geographical location (i.e., $\mathcal{X} = \Re^2$).  Another simple, tractable case is when $x$ represents an ordered categorical predictor (i.e.,  $\mathcal{X} = \{1,\ldots, N\}$). We seek a prior for $\Sigma_{\mathcal{X}} = \{ \Sigma(x), x \in \mathcal{X} \}$, the collection of covariance matrices over the space of predictor values.

Letting $\Sigma_{\mathcal{X}} \sim \PSigma$ our goal is to choose a prior $\PSigma$ for the collection of covariance matrices that has large support and leads to good performance in large $p$ settings.  By ``good'' we mean accurate estimation in small samples, taking advantage of shrinkage priors and efficient computation that scales well as $p$ increases.  We initially focus on the relatively simple setting in which 
\begin{align}
	y_i \sim \mathcal{N}_p( \mu(x_i), \Sigma(x_i))
	\label{eqn:obsModel}
\end{align}	
independently for each $i$.  Such a formulation could be extended to settings in which data are collected at repeated times for different subjects, as in multivariate longitudinal data analysis, by embedding the proposed model within a hierarchical framework.  See Section~\ref{sec:discussion} for a discussion.

\subsection{Proposed Latent Factor Model}

In large $p$ settings, modeling a $p \times p$ covariance matrix $\Sigma(x)$ over an arbitrary predictor space $\mathcal{X}$ represents an enormous dimensional regression problem; we aim to reduce dimensionality for tractability in building a flexible nonparametric model for the predictor-dependent covariances.  A popular approach for coping with such high dimensional (non-predictor-dependent) covariance matrices $\Sigma$ in the presence of limited data is to assume that the covariance has a decomposition as $\Lambda\Lambda' + \Sigma_0$ where $\Lambda$ is a $p \times k$ factor loadings matrix with $k << p$ and $\Sigma_0$ is a $p \times p$ diagonal matrix with non-negative entries.  To build in predictor dependence, we assume a decomposition
\begin{align}
	\Sigma(x) = \Lambda(x)\Lambda(x)' + \Sigma_0,
\end{align}
where $\Lambda(x)$ is a $p \times k$ factor loadings matrix that is indexed by predictors $x$ and where $\Sigma_0 = \mbox{diag}(\sigma_1^2,\ldots, \sigma_p^2)$. Assuming initially for simplicity that $\mu(x) = 0$, such a decomposition is induced by marginalizing out a set of latent factors $\eta_i$ from the following latent factor model:
\begin{align}
	\begin{aligned}
		y_i = \Lambda(x_i)&\eta_i + \epsilon_i\\
		\eta_i \sim \mathcal{N}_k(0,I_k), &\quad \epsilon_i \sim \mathcal{N}_p(0,\Sigma_0). \label{eq:base}
	\end{aligned}
\end{align}
Here, $x_i = (x_{i1},\ldots, x_{iq})'$ is the predictor associated with the $i$th observation $y_i$.

Despite the dimensionality reduction introduced by the latent factor model of Eq.~\eqref{eq:base}, modeling a $p \times k$ dimensional predictor-dependent factor loadings matrix $\Lambda(x)$ still represents a significant challenge for large $p$ domains.  To further reduce dimensionality, and following the strategy of building a flexible high-dimensional model from simple low-dimensional pieces, we assume that each element of $\Lambda(x)$ is a linear combination of a much smaller number of unknown dictionary functions $\xi_{\ell k}: \mathcal{X} \to \Re$. That is, we propose to let
\begin{align}
\Lambda(x_i) = \Theta \xi(x_i), \label{eq:Lamx}
\end{align}
where $\Theta \in \Re^{p \times L}$ is the matrix of coefficients relating the predictor-dependent factor loadings matrix to the set of dictionary functions comprising the $L \times k$ dimensional matrix $\xi(x)$. Typically, $k << p$ and $L << p$.  Since we can write
\begin{align}
[\Lambda(\cdot)]_{rs} = \sum_{\ell = 1}^L \theta_{r\ell}\xi_{\ell s}(\cdot),
\end{align}
we see that the weighted sum of the $s$th column of dictionary functions $\xi_{\cdot s}(\cdot)$, with weights specified by the $r$th row of $\Theta$, characterizes the impact of the $s$th latent factor on $y_{ir}$, the $r$th component of the response at predictor location $x_i$.  By characterizing the elements of $\Lambda(x_i)$ as a linear combination of these flexible dictionary functions, we obtain a highly-flexible but computationally tractable 
formulation. 

In marginalizing out the latent factors, we now obtain the following induced covariance structure
\begin{align}
\mbox{cov}(y_i\, |\, x_i=x) = \Sigma(x) = \Theta \xi(x) \xi(x)' \Theta' + \Sigma_0. \label{eq:covyx}
\end{align}  
Note that the above decomposition of $\Sigma(x)$ is not unique and there are actually infinitely many such equivalent decompositions.  For example, take $\Theta_1 = c\Theta$ and $\xi_1(\cdot) = (1/c)\xi(\cdot)$. Alternatively, consider $\xi_1(\cdot) = \xi(\cdot)P$ for any orthogonal matrix $P$ or $\Theta_1 = [\Theta \,\, 0_{p \times d}]$ and $\xi_1 = [\xi;\, \xi_0]$ for any $d \times k$ matrix of dictionary functions $\xi_0$.  One can also increase the dimension of the latent factors and take $\xi_1 = [\xi \,\, 0_{L \times d}]$. In standard (non-predictor-dependent) latent factor modeling, a common approach to ensure identifiability is to constrain the factor loadings matrix to be block lower triangular with strictly positive diagonal elements~\citep{GewekeZhou:96}, though such a constraint induces order dependence among the responses~\citep{AguilarWest:00,West:03,LopesWest:04,Carvalho:08}. However, for tasks such as inference on the covariance matrix and prediction, identifiability of a unique decomposition is not necessary. Thus, we do not restrict ourselves to a unique decomposition of $\Sigma(x)$, allowing us to define priors with better computational properties. 

Although we are not interested in identifying a unique decomposition of $\Sigma(x)$, we are interested in characterizing the class of covariance regressions $\Sigma(x)$ that can be decomposed as in Eq.~\eqref{eq:covyx}.  Lemma~\ref{lemma:factorization} states that for $L$ and $k$ sufficiently large, any covariance regression has such a decomposition. For $L,k \geq p$, let $\mathcal{X}_{\xi}$ denote the space of all $L \times k$ dimensional matrices of arbitrarily complex dictionary functions mapping from $\mathcal{X}\rightarrow \Re$, $\mathcal{X}_{\Sigma_0}$ be the space of all $p \times p$ diagonal matrices with non-negative entries, and $\mathcal{X}_{\Theta}$ be the space of all $p \times L$ dimensional matrices $\Theta$ such that $\Theta\Theta'$ has finite elements.

\begin{lemma}
	Given a symmetric positive semidefinite matrix $\Sigma(x) \succ 0,\, \forall x \in \mathcal{X}$, there exists $\{\xi(\cdot),\Theta,\Sigma_0\} \in \mathcal{X}_{\xi} \otimes \mathcal{X}_{\Theta} \otimes \mathcal{X}_{\Sigma_0}$ such that
	\begin{align}
		\Sigma(x) = \Theta \xi(x)\xi(x)'\Theta' + \Sigma_0, \quad \forall x \in \mathcal{X}.
	\end{align}
	\begin{proof}
	Assume without loss of generality that $\Sigma_0 = 0_{p\times p}$ and take $k,L\geq p$.  Consider
	\begin{align}
		\begin{aligned}
			\Theta = [I_p \,\, 0_{p\times L-p}] \hspace{0.2in}
			\xi(x) = \begin{bmatrix} \mbox{chol}(\Sigma(x)) & 0_{p \times k-p}\\ 0_{L-p \times p} & 0_{L-p \times k-p} \end{bmatrix}.
		\end{aligned}
	\end{align}
	Then, $\Sigma(x) = \Theta\xi(x)\xi(x)'\Theta', \, \forall x \in \mathcal{X}$.
	\end{proof}
	\label{lemma:factorization}
\end{lemma}
Now that we have established that there exist decompositions of $\Sigma(x)$ into the form specified by Equation~\eqref{eq:covyx}, the question is whether we can specify a prior on the elements $\xi(\cdot)$, $\Theta$, and $\Sigma_0$ that provides large support on such decompositions.  This is explored in Section~\ref{sec:priors}.

In order to generalize the model to also allow the mean $\mu(x)$ to vary flexibly with predictors, we can follow a nonparametric latent factor regression approach and let 
\begin{align}
\eta_i = \psi( x_i ) + \nu_i,\quad \nu_i \sim \mathcal{N}_k( 0, I_k ), \label{eq:latreg}
\end{align}
where $\psi(x_i) = [ \psi_1(x_i),\ldots, \psi_k(x_i) ]'$, and $\psi_j: \mathcal{X} \to \Re$ is an unknown function relating the predictors to the mean of the $j$th factor, for $j=1,\ldots, k$.  These $\psi_j(\cdot)$ functions can be modeled in a related manner to the $\xi_{\ell k}(\cdot)$ functions described above.  The induced mean of $y_i$ conditionally on $x_i=x$ and marginalizing out the latent factors is then $\mu(x) = \Theta \xi(x) \psi(x)$. For simplicity, however, we focus our discussions on the case where $\mu(x)=0$.

\subsection{Prior Specification}
\label{sec:priors}

Working within a Bayesian framework, we place independent priors on $\xi(\cdot)$, $\Theta$, and $\Sigma_0$ in Eq.~\eqref{eq:covyx} to induce a prior on $\Sigma_{\mathcal{X}}$. Let $\Pxi$, $\PTheta$, and $\PSigmaO$ denote each of these independent priors, respectively.  Recall that $\PSigma$  denotes the induced prior on $\Sigma_{\mathcal{X}}$. 

Aiming to capture covariances that vary continuously over $\mathcal{X}$ combined with the goal of maintaining simple computations for inference, we specify the dictionary functions as
\begin{align}
	\xi_{\ell k}(\cdot) \sim \mbox{GP}(0,c)
\end{align}
independently for all $\ell,k$, with $c$ a squared exponential correlation function having $c(\xi,\xi') = \exp(-\kappa ||\xi-\xi'||_2^2)$.

To cope with the fact that the number of latent dictionary functions is a model choice we are required to make, we seek a prior $\PTheta$ that favors many values of $\Theta$ being close to zero so that we may choose $L$ larger than the expected number of dictionary functions (also controlled by the latent factor dimension $k$). As proposed in~\citet{Bhattacharya:10}, we use the following shrinkage prior:
\begin{align}
	\begin{aligned}
		\theta_{j\ell} \mid \phi_{j \ell},\tau_{\ell} \sim \mathcal{N}(0,\phi_{j \ell}^{-1}\tau_{\ell}^{-1}) &\quad \phi_{j \ell} \sim \mbox{Ga}(3/2,3/2)\\
		\delta_1 \sim \mbox{Ga}(a_1,1), \quad \delta_h \sim \mbox{Ga}(a_2&,1),\, h \geq 2, \quad \tau_{\ell} = \prod_{h=1}^\ell \delta_h.
	\end{aligned} 
	\label{eqn:shrinkage}
\end{align}
Choosing $a_2>1$ implies that $\delta_h$ is greater than 1 in expectation so that $\tau_{\ell}$ tends stochastically towards infinity as $\ell$ goes to infinity, thus shrinking the elements $\theta_{j\ell}$ toward zero increasingly as $\ell$ grows.  The $\phi_{j\ell}$ precision parameters allow for flexibility in how the elements of $\Theta$ are shrunk towards zero by incorporating local shrinkage specific to each element of $\Theta$, while $\tau_{\ell}$ provides a global column-wise shrinkage factor.

Finally, we specify $\PSigmaO$ via the usual inverse gamma priors on the diagonal elements of $\Sigma_0$.  That is,
\begin{align}
	\sigma_j^{-2} \sim \mbox{Ga}(a_\sigma,b_\sigma)
\end{align}
independently for each $j=1,\ldots,p$.

\subsection{Theoretical Properties}

In this section, we explore the theoretical properties of the proposed Bayesian nonparametric covariance regression model.  In particular, we focus on the support of the induced prior $\PSigma$ based on the priors $\Pxi$, $\PTheta$, and $\PSigmaO$ defined in Section~\ref{sec:priors}.  Large support implies that the prior can generate covariance regressions that are arbitrarily close to any function $\{\Sigma(x),x \in \mathcal{X}\}$ in a large class.  Such a support property is the defining feature of a Bayesian nonparametric approach and cannot simply be assumed.  Often, seemingly flexible models can have quite restricted support due to hidden constraints in the model and not to real prior knowledge that certain values are implausible.  Although we have chosen a specific form for a shrinkage prior $\PTheta$, we aim to make our statement of prior support as general as possible and thus simply assume that $\PTheta$ satisfies a set of two conditions given by Assumption~\ref{assumption:absSum} and Assumption~\ref{assumption:rankTheta}.  In Lemma~\ref{lemma:absSum}, we show that the $\PTheta$ specified in Eq.~\eqref{eqn:shrinkage} satisfies these assumptions.  The proofs associated with the theoretical statements made in this section can be found in the Appendix.

\begin{assumption}
	$\PTheta$ is such that $\sum_{\ell} E[|\theta_{j\ell}|] < \infty$.  This property ensures that the prior on the rows of $\Theta$ shrinks the elements towards zero fast enough as $\ell \rightarrow \infty$.
	\label{assumption:absSum}
\end{assumption}

\begin{assumption}
	$\PTheta$ is such that $\PTheta\left(\mbox{rank}(\Theta)=p\right) > 0$.  That is, there is positive prior probability of $\Theta$ being full rank. \label{assumption:rankTheta}
\end{assumption}

The following theorem shows that, for $k \geq p$ and as $L \rightarrow \infty$, the induced prior $\PSigma$ places positive probability on the space of all covariance functions $\Sigma^*(x)$ that are continuous on $\mathcal{X}$.  
\begin{theorem}
	Let $\PSigma$ denote the induced prior on $\{\Sigma(x),x \in \mathcal{X}\}$ based on the specified prior $\Pxi \otimes \PTheta \otimes \PSigmaO$ on $\mathcal{X}_{\xi} \otimes \mathcal{X}_{\Theta} \otimes \mathcal{X}_{\Sigma_0}$.  Assuming $\mathcal{X}$ compact, for all continuous $\Sigma^*(x)$ and for all $\epsilon>0$,
	\begin{align}
		\PSigma \left( \sup_{x\in\mathcal{X}} ||\Sigma(x)-\Sigma^*(x)||_2 < \epsilon \right) > 0.
	\end{align}
	\label{thm:priorsupport}
\end{theorem}
Intuitively, the support on continuous covariance functions $\Sigma^*(x)$ arises from the continuity of the Gaussian process dictionary functions. However, since we are mixing over infinitely many such dictionary functions, we need the mixing weights specified by $\Theta$ to tend towards zero, and to do so ``fast enough''---this is where Assumption~\ref{assumption:absSum} becomes important.  See Theorem~\ref{thm:continuity}.  The proof of Theorem~\ref{thm:priorsupport} relies on the large support of $\PSigma$ at any point $x_0 \in \mathcal{X}$. Since each $\xi_{\ell k}(x_0)$ is independently Gaussian distributed (based on properties of the Gaussian process prior), $\xi(x_0)\xi(x_0)'$ is Wishart distributed. Conditioned on $\Theta$, $\Theta\xi(x_0)\xi(x_0)'\Theta'$ is also Wishart distributed.  More generally, for fixed $\Theta$, $\Theta\xi(x)\xi(x)'\Theta'$ follows the matric-variate Wishart process of~\cite{Gelfand:04}.  Combining the large support of the Wishart distribution with that of the gamma distribution on the inverse elements of $\Sigma_0$ provides the desired large support of the induced prior $\PSigma$ at each predictor location $x_0$.

\begin{theorem}
	For every finite $k$ and $L\rightarrow \infty$ (or $L$ finite), $\Lambda(\cdot)= \Theta\xi(\cdot)$ is almost surely continuous on $\mathcal{X}$.
	\label{thm:continuity}
\end{theorem}

Lemma~\ref{lemma:absSum} specifies the conditions under which the prior $\PTheta$ specified in Eq.~\eqref{eqn:shrinkage} satisfies Assumption~\ref{assumption:absSum}, which provides a sufficient condition used in the proof of prior support.
\begin{lemma}
	Based on the prior specified in Eq.~\eqref{eqn:shrinkage} and choosing $a_2>2$, Assumption~\ref{assumption:absSum} is satisfied.  That is,
	%
		$\sum_{\ell} E[|\theta_{j\ell}|] < \infty$.
	%
	\label{lemma:absSum}
\end{lemma}

It is also of interest to analyze the moments associated with the proposed prior.  As detailed in the Appendix, the first moment can be derived based on the implied inverse gamma prior on the $\sigma_j^2$ combined with the fact that $\Theta\xi(x)\xi(x)'\Theta'$ is marginally Wishart distributed at every location $x$, with the prior on $\Theta$ specified in Equation~\eqref{eqn:shrinkage}.
\begin{lemma}
	Let $\mu_{\sigma}$ denote the mean of $\sigma_j^2$, $j=1,\dots,p$. Then, 
\begin{align}
	E[\Sigma(x)] = \mbox{diag}\left(k\sum_{\ell}\phi_{1\ell}^{-1}\tau_{\ell}^{-1} + \mu_{\sigma},\dots,k\sum_{\ell}\phi_{p\ell}^{-1}\tau_{\ell}^{-1} + \mu_{\sigma}\right).
\end{align}
\label{lemma:firstMoment}
\end{lemma}

Since our goal is to develop a covariance regression model, it is natural to consider the correlation induced between an element of the covariance matrix at different predictor locations $x$ and $x'$.
\begin{lemma}
	Let $\sigma_{\sigma}^2$ denote the variance of $\sigma_j^2$, $j=1,\dots,p$. Then, 
\begin{align}
	\mbox{cov}(\Sigma_{ij}(x),\Sigma_{ij}(x')) = \left\{ 
		\begin{array}{ll}
			kc(x,x')\left(5\sum_{\ell}\phi_{i\ell}^{-2}\tau_{\ell}^{-2} + (\sum_{\ell}\phi_{i\ell}^{-1}\tau_{\ell}^{-1})^2\right) + \sigma_{\sigma}^2 & i=j,\\
			kc(x,x')\left(\sum_{\ell}\phi_{i\ell}^{-1}\phi_{j\ell}^{-1}\tau_{\ell}^{-2} + \sum_{\ell}\phi_{i\ell}^{-1}\tau_{\ell}^{-1}\sum_{\ell'}\phi_{j\ell'}^{-1}\tau_{\ell'}^{-1}\right) & i\neq j.
		\end{array}\right.
		\label{eqn:covSigma_ij}
\end{align}
For any two elements $\Sigma_{ij}(x)$ and $\Sigma_{uv}(x')$ with $i\neq u$ or $j \neq v$,
\begin{align}
	\mbox{cov}(\Sigma_{ij}(x),\Sigma_{uv}(x')) = 0.
	\label{eqn:covSigma_uv}
\end{align}
\label{lemma:secondMoment}
\end{lemma}

We can thus conclude that in the limit as the distance between the predictors tends towards infinity, the correlation decays at a rate defined by the Gaussian process kernel $c(x,x')$ with a limit:
\begin{align}
	\lim_{||x-x'||\rightarrow \infty}\mbox{cov}(\Sigma_{ij}(x),\Sigma_{uv}(x')) = \left\{ 
		\begin{array}{cc}
		\sigma_{\sigma}^2 & i=j=u=v,\\
			0 & \mbox{otherwise}.
		\end{array}\right.
\end{align}
It is perhaps counterintuitive that the correlation between $\Sigma_{ii}(x)$ and $\Sigma_{ii}(x')$ does not go to zero as the distance between the predictors $x$ and $x'$ tends to infinity.  However, although the correlation between $\xi(x)$ and $\xi(x')$ goes to zero, the diagonal matrix $\Sigma_0$ does not depend on $x$ or $x'$ and thus retains the correlation between the diagonal elements of $\Sigma(x)$ and $\Sigma(x')$.

Equation~\eqref{eqn:covSigma_ij} implies that the autocorrelation $ACF(x) = \mbox{corr}(\Sigma_{ij}(0),\Sigma_{ij}(x))$ is simply specified by $c(0,x)$.  When we choose a Gaussian process kernel $c(x,x') = \exp(-\kappa||x-x'||_2^2)$, we have
\begin{align}
ACF(x) = \exp(-\kappa ||x||_2^2).
\label{eqn:ACF}
\end{align}
Thus, we see that the length-scale parameter $\kappa$ directly determines the shape of the autocorrelation function.

Finally, one can analyze the stationarity properties of the proposed covariance regression prior.
\begin{lemma}
	Our proposed covariance regression model defines a first-order stationary process in that
	%
		$\PSigma(\Sigma(x)) = \PSigma(\Sigma(x')), \, \forall x,x' \in \mathcal{X}$.
	%
	Furthermore, the process is wide sense stationary: $\mbox{cov}(\Sigma_{ij}(x),\Sigma_{uv}(x'))$ solely depends upon $||x-x'||$.
	\begin{proof}
		The first-order stationarity follows immediately from the stationarity of the Gaussian process dictionary elements $\xi_{\ell k}(\cdot)$ and recalling that $\Sigma(x) = \Theta\xi(x)\xi(x)'\Theta' + \Sigma_0$.  Assuming a Gaussian process kernel $c(x,x')$ that solely depends upon the distance between $x$ and $x'$ (as in Section~\ref{sec:priors}), Equations~\eqref{eqn:covSigma_ij}-~\eqref{eqn:covSigma_uv} imply that the defined process is wide sense stationary.
	\end{proof}
	\label{lemma:stationarity}
\end{lemma}

\section{Posterior Computation}
\label{sec:comp}

\subsection{Gibbs Sampling with a Fixed Truncation Level}
\label{sec:Gibbs}

Based on a fixed truncation level $L^*$ and a latent factor dimension $k^*$, we propose a Gibbs sampler for posterior computation.  The derivation of Step 1 is provided in the Appendix. 

\paragraph{Step 1}

Update each dictionary function $\xi_{\ell m}(\cdot)$ from the conditional posterior given $\{y_i\}$, $\Theta$, $\{\eta_i\}$, $\Sigma_0$. We can rewrite the observation model for the $j$th component of the $i$th response as

\begin{align}
	y_{ij} = \sum_{m=1}^{k^*} \eta_{im} \sum_{\ell=1}^{L^*} \theta_{j\ell}\xi_{\ell m}(x_i) + \epsilon_{ij}.
\end{align}
Conditioning on $\xi(\cdot)^{-\ell m} = \{\xi_{rs}(\cdot), r\neq\ell, s\neq m\}$, our Gaussian process prior on the dictionary functions implies the following conditional posterior
\begin{align}
	\begin{bmatrix} \xi_{\ell m}(x_1) \\ \xi_{\ell m}(x_2) \\ \vdots \\ \xi_{\ell m}(x_n) \end{bmatrix} \mid \{y_i\},\Theta,\eta,\xi(\cdot)^{-\ell m},\Sigma_0 &\sim 
	\mathcal{N}_n\left(\tilde{\Sigma}_{\xi}\begin{bmatrix} \eta_{1m}\sum_{j=1}^p \theta_{j\ell}\sigma_j^{-2}\tilde{y}_{1j} \\ \vdots \\ \eta_{nm}\sum_{j=1}^p \theta_{j\ell}\sigma_j^{-2}\tilde{y}_{nj} \end{bmatrix},\tilde{\Sigma}_{\xi} \right),
\end{align} 
where $\tilde{y}_{ij} = y_{ij} - \sum_{(r,s)\neq(\ell,m)} \theta_{jr}\xi_{rs}(x_i)$ and, taking $K$ to be the Gaussian process covariance matrix with $K_{ij} = c(x_i,x_j)$,
\begin{align}
	\tilde{\Sigma}_{\xi}^{-1} &= K^{-1} + \mbox{diag}\left(\eta_{1m}^2\sum_{j=1}^p \theta_{j\ell}^2\sigma_j^{-2},\dots,\eta_{nm}^2\sum_{j=1}^p\theta_{j\ell}^2\sigma_j^{-2}\right).	
\end{align}

\paragraph{Step 2}

Next we sample each latent factor $\eta_i$ given $y_i$, $\xi(\cdot)$, $\Theta$, $\Sigma_0$. Recalling Eq.~\eqref{eq:base} and the fact that $\eta_i \sim \mathcal{N}_{k^*}(0,I_{k^*})$,
\begin{multline}
	\eta_i \mid y_i,\Theta,\xi(x_i),\Sigma_0\\ \sim \mathcal{N}_{k^*}\left(\left(I + \xi(x_i)'\Theta'\Sigma_0^{-1}\Theta\xi(x_i) \right)^{-1}\xi(x_i)'\Theta'\Sigma_0^{-1}y_i,\left(I + \xi(x_i)'\Theta'\Sigma_0^{-1}\Theta\xi(x_i) \right)^{-1} \right).
	\label{eqn:etaPost}
\end{multline}

\paragraph{Step 3}

Let $\theta_{j\cdot} = \begin{bmatrix} \theta_{j1} & \dots & \theta_{jL^*} \end{bmatrix}$. Recalling the $\mbox{Ga}(a_\sigma,b_\sigma)$ prior on each precision parameter $\sigma_j^{-2}$ associated with the diagonal noise covariance matrix $\Sigma_0$, standard conjugate posterior analysis yields the posterior
\begin{align}
	\sigma_j^{-2} \mid \{y_i\},\Theta,\eta, \xi(\cdot) \sim \mbox{Ga}\left(a_\sigma + \frac{n}{2},b_\sigma + \frac{1}{2}\sum_{i=1}^n (y_{ij} - \theta_{j\cdot}\xi(x_i)\eta_i)^2\right).
\end{align}

\paragraph{Step 4}

Conditioned on the hyperparameters $\phi$ and $\tau$, the Gaussian prior on the elements of $\Theta$ specified in Eq.~\eqref{eqn:shrinkage} combined with the likelihood defined by Eq.~\eqref{eq:base} imply the following posterior for each row of $\Theta$:
\begin{align}
	\theta_{j\cdot} \mid \{y_i\},\eta,\xi(\cdot),\phi,\tau \sim \mathcal{N}_{L^*}\left( \tilde{\Sigma}_{\theta}
	\tilde{\eta}'\sigma_j^{-2}\begin{bmatrix} y_{1j} \\ \vdots \\ y_{nj} \end{bmatrix},
	\tilde{\Sigma}_{\theta} \right),
\end{align}
where $\tilde{\eta}' = \begin{bmatrix} \xi(x_1)\eta_1 & \xi(x_2)\eta_2 & \dots & \xi(x_n)\eta_n \end{bmatrix}$ and
\begin{align}
	\tilde{\Sigma}_{\theta}^{-1} = \sigma_j^{-2}\tilde{\eta}'\tilde{\eta} + \mbox{diag}(\phi_{j1}\tau_1,\dots,\phi_{jL^*}\tau_{L^*}).
\end{align}

\paragraph{Step 5}

Examining Eq.~\eqref{eqn:shrinkage} and using standard conjugate analysis results in the following posterior for each local shrinkage hyperparameter $\phi_{j \ell}$ given $\theta_{j \ell}$ and $\tau_{\ell}$:
\begin{align}
	\phi_{j \ell} \mid \theta_{j \ell},\tau_{\ell} \sim \mbox{Ga}\left(2, \frac{3+\tau_\ell\theta_{j \ell}^2}{2}\right).
\end{align}

\paragraph{Step 6}

As in~\citet{Bhattacharya:10}, the global shrinkage hyperparameters are updated as
\begin{align}
	\begin{aligned}
		\delta_1 \mid \Theta,\tau^{(-1)} &\sim \mbox{Ga}\left(a_1 + \frac{pL^*}{2},1 + \frac{1}{2}\sum_{\ell=1}^{L^*} \tau_\ell^{(-1)}\sum_{j=1}^p\phi_{j\ell}\theta_{j\ell}^2\right)\\
		\delta_h \mid \Theta,\tau^{(-h)} &\sim \mbox{Ga}\left(a_2 + \frac{p(L^*-h+1)}{2},1 + \frac{1}{2}\sum_{\ell=1}^{L^*} \tau_\ell^{(-h)}\sum_{j=1}^p\phi_{j\ell}\theta_{j\ell}^2\right),
	\end{aligned} 
\end{align}
where $\tau_{\ell}^{(-h)} = \prod_{t=1,t\neq h}^\ell \delta_t$ for $h=1,\dots,p$.

\subsection{Incorporating nonparametric mean $\mu(x)$}

If one wishes to incorporate a latent factor regression model such as in Eq.~\eqref{eq:latreg} to induce a predictor-dependent mean $\mu(x)$, the MCMC sampling is modified as follows.  Steps 1, 3, 4, 5, and 6 are exactly as before.  Now, however, the sampling of $\eta_i$ of Step 2 is replaced by a block sampling of $\psi(x_i)$ and $\nu_i$.  Specifically, let $\Omega_i = \Theta\xi(x_i)$.  We can rewrite the observation model as $y_i = \Omega_i\psi(x_i) + \Omega_i\nu_i + \epsilon_i$.  Marginalizing out $\nu_i$, $y_i = \Omega_i\psi(x_i) + \omega_i$ with $\omega_i \sim \mathcal{N}(0,\tilde{\Sigma}_i \triangleq \Omega_i\Omega_i' + \Sigma_0)$. Assuming nonparametric mean vector components $\psi_{\ell}(\cdot) \sim \mbox{GP}(0,c)$, the posterior of $\psi_{\ell}(\cdot)$ follows analogously to that of $\xi(\cdot)$ resulting in
\begin{align}
	\begin{bmatrix} \psi_{\ell}(x_1) \\ \psi_{\ell}(x_2) \\ \vdots \\ \psi_{\ell}(x_n) \end{bmatrix} \mid \{y_i\},\psi(\cdot)^{-\ell},\Theta,\eta,\xi(\cdot),\Sigma_0 &\sim 
	\mathcal{N}_n\left(\tilde{\Sigma}_{\psi}\begin{bmatrix} [\Omega_1]_{\cdot \ell}'\tilde{\Sigma}_1^{-1}\tilde{y}_{1}^{-\ell} \\ \vdots \\ [\Omega_n]_{\cdot \ell}'\tilde{\Sigma}_n^{-1}\tilde{y}_{n}^{-\ell} \end{bmatrix},\tilde{\Sigma}_{\psi} \right),
\end{align} 
where $\tilde{y}_{i}^{-\ell} = y_{i} - \sum_{(r \neq \ell)} [\Omega_i]_{\cdot r} \psi_r(x_i)$.  Once again taking $K$ to be the Gaussian process covariance matrix,
\begin{align}
	\tilde{\Sigma}_{\xi}^{-1} &= K^{-1} + \mbox{diag}\left([\Omega_1]_{\cdot \ell}'\tilde{\Sigma}_1^{-1}[\Omega_1]_{\cdot \ell},\dots,[\Omega_n]_{\cdot \ell}'\tilde{\Sigma}_n^{-1}[\Omega_n]_{\cdot \ell}\right).	
\end{align}

Conditioned on $\psi(x_i)$, we consider $\tilde{y}_i = y_i - \Omega_i\psi(x_i) = \Omega_i\nu_i + \epsilon_i$.  Then, using the fact that $\nu_i\sim \mathcal{N}(0,I_{k*})$,
\begin{multline}
	\nu_i \mid \tilde{y}_i,\psi(x_i),\Theta,\xi(x_i),\Sigma_0\\ \sim \mathcal{N}_{k^*}\left(\left(I + \xi(x_i)'\Theta'\Sigma_0^{-1}\Theta\xi(x_i) \right)^{-1}\xi(x_i)'\Theta'\Sigma_0^{-1}\tilde{y}_i,\left(I + \xi(x_i)'\Theta'\Sigma_0^{-1}\Theta\xi(x_i) \right)^{-1} \right).
\end{multline}

\subsection{Hyperparameter Sampling}
\label{sec:hyp}

One can also consider sampling the Gaussian process length-scale hyperparameter $\kappa$.  Due to the linear-Gaussianity of the proposed covariance regression model, we can analytically marginalize the latent Gaussian process random functions in considering the posterior of $\kappa$.  Once again taking $\mu(x)=0$ for simplicity, our posterior is based on marginalizing the Gaussian process random vectors $\xi_{\ell m} = [\xi_{\ell m}(x_1) \,\, \dots \,\, \xi_{\ell m}(x_n)]'$. Noting that 
\begin{align}
	\begin{bmatrix} y_1' & y_2' & \dots & y_n' \end{bmatrix}' = 
	\sum_{\ell m} \left[ \mbox{diag}(\eta_{\cdot m}) \otimes \theta_{\cdot \ell} \right] \xi_{\ell m} + \begin{bmatrix} \epsilon_1' & \epsilon_2' & \dots & \epsilon_n' \end{bmatrix}',
\end{align} 
and letting $K_{\kappa}$ denote the Gaussian process covariance matrix based on a length-scale $\kappa$,
\begin{align}
	\begin{bmatrix} y_1 \\ y_2 \\ \vdots \\ y_n \end{bmatrix} \mid \kappa,\Theta,\eta,\Sigma_0 \sim \mathcal{N}_{np}\left(\sum_{\ell,m} \left[\mbox{diag}(\eta_{\cdot m}) \otimes \theta_{\cdot \ell}\right] K_{\kappa} \left[ \mbox{diag}(\eta_{\cdot m}) \otimes \theta_{\cdot \ell} \right]' + \begin{bmatrix} \Sigma_0 & & & \\ & \Sigma_0 & & \\ & & \ddots & \\ & & & \Sigma_0 \end{bmatrix} \right). \label{eqn:Kc_like}
\end{align} 
We can then Gibbs sample $\kappa$ based on a fixed grid and prior $p(\kappa)$ on this grid. Note, however, that computation of the likelihood specified in Eq.~\eqref{eqn:Kc_like} requires evaluation of an $np$-dimensional Gaussian for each value $\kappa$ specified in the grid.  For large $p$ scenarios, or when there are many observations $y_i$, this may be computationally infeasible.  In such cases, a naive alternative is to iterate between sampling $\xi(\cdot)$ given $K_{\kappa}$ and $K_{\kappa}$ given $\xi(\cdot)$.  However, this can lead to extremely slow mixing.  Alternatively, one can consider employing the recent Gaussian process hyperparameter slice sampler of~\cite{AdamsMurray:10}.

In general, because of the quadratic mixing over Gaussian process dictionary elements, our model is relatively robust to the choice of the length-scale parameter and the computational burden imposed by sampling $\kappa$ is typically unwarranted.  Instead, one can pre-select a value for $\kappa$ using a data-driven heuristic, which leads to a quasi-empirical Bayes approach.  Recalling Equation~\eqref{eqn:ACF}, we have
\begin{align}
	-\log(ACF(x)) = \kappa ||x||_2^2.
	\label{eqn:logACF}
\end{align}
Thus, if one can devise a procedure for estimating the autocorrelation function from the data, one can set $\kappa$ accordingly.  We propose the following.
\begin{itemize}
\item[1.]  For a set of evenly spaced knots $x_k \in \mathcal{X}$, compute the sample covariance $\hat{\Sigma}(x_k)$ from a local bin of data $y_{k-k_0:k+k_0}$ with $k_0 > p/2$.
\item[2.]  Compute the Cholesky decomposition $C(x_k) = chol(\hat{\Sigma}(x_k))$.
\item[3.]  Fit a spline through the elements of the computed $C(x_k)$.  Denote the spline fit of the Cholesky by $\tilde{C}(x)$ for each $x\in\mathcal{X}$
\item[4.]  For $i=1,\dots,n$, compute a point-by-point estimate of $\Sigma(x_i)$ from the splines: $\Sigma(x_i) = \tilde{C}(x_i)\tilde{C}(x_i)'$.
\item[5.]  Compute the autocorrelation function of each element $\Sigma_{ij}(x)$ of this kernel-estimated $\Sigma(x)$.
\item[6.]  According to Equation~\eqref{eqn:logACF}, choose $\kappa$ to best fit the most correlated $\Sigma_{ij}(x)$ (since less correlated components can be captured via weightings of dictionary elements with stronger correlation.)
\end{itemize}

\subsection{Computational Considerations}
\label{sec:compIssues}

In choosing a truncation level $L^*$ and latent factor dimension $k^*$, there are a number of computational considerations.  The Gibbs sampler outlined in~Section~\ref{sec:Gibbs} involves a large number of simulations from Gaussian distributions, each of which requires the inversion of an $m$-dimensional covariance matrix, with $m$ the dimension of the Gaussian.  For large $m$, this represents a large computational burden as the operation is, in general, $O(m^3)$.  The computations required at each stage of the Gibbs sampler are summarized in Table~\ref{table:computations}.  From this table we see that depending on the number of observations $n$ and the dimension of these observations $p$, various combinations of $L^*$ and $k^*$ lead to more or less efficient computations.

\begin{table}
\caption{\label{table:computations}Computations required at each Gibbs sampling step.}
\centering
\fbox{%
\begin{tabular}{c|c}
Gibbs Update & Computation\\
\hline
Step 1 & $L^* \times k^*$ draws from an $n$ dimensional Gaussian\\
Step 2 & $n$ draws from a $k^*$ dimensional Gaussian\\
Step 3 & $p$ draws from a gamma distribution\\
Step 4 & $p$ draws from an $L^*$ dimensional Gaussian\\
Step 5 & $p \times L^*$ draws from a gamma distribution\\
Step 6 & $L^*$ draws from a gamma distribution\\
\end{tabular}}
\end{table}

In~\citet{Bhattacharya:10}, a method for adaptively choosing the number of factors in a non-predictor dependent latent factor model was proposed.  One could directly apply such a methodology for adaptively selecting $L^*$. To handle the choice of $k^*$, one could consider an augmented formulation in which 
\begin{align}
	\Lambda(x) = \Theta\xi(x)\Gamma,
\end{align}
where $\Gamma = \mbox{diag}(\gamma_1,\dots,\gamma_k)$ is a diagonal matrix of parameters that shrink the columns of $\xi(x)$ towards zero.  One can take these shrinkage parameters to be distributed as
\begin{align}
	\begin{aligned}
	\gamma_i \sim \mathcal{N}(0,\omega_i^{-1}), &\quad \omega_i = \prod_{h=1}^i\zeta_h\\
	\zeta_1 \sim \mbox{Ga}(a_3,1), &\quad \zeta_h \sim \mbox{Ga}(a_4,1) \quad h=2,\dots,k.
	\end{aligned}
\end{align}
For $a_4>1$, such a model shrinks the $\gamma_i$ values towards zero for large indices $i$ just as in the shrinkage prior on $\Theta$. The $\gamma_i$ close to zero provide insight into redundant latent factor dimensions.  Computations in this augmented model are a straightforward extension of the Gibbs sampler presented in Section~\ref{sec:Gibbs}.  Based on the inferred values of the latent $\gamma_i$ parameters, one can design an adaptive strategy similar to that for $L^*$. 

Note that in Step 1, the $n$-dimensional inverse covariance matrix $\tilde{\Sigma}_{\xi}^{-1}$ which needs to be inverted in order to sample $\xi_{\ell m}$ is a composition of a diagonal matrix and an inverse covariance matrix $K^{-1}$ that has entries that tend towards zero as $||x_i-x_j||_2$ becomes large (i.e., for distant pairs of predictors.)  That is, $K^{-1}$ (and thus $\tilde{\Sigma}_{\xi}^{-1}$) is nearly band-limited, with a bandwidth dependent upon the Gaussian process parameter $\kappa$.  Inverting a given $n \times n$ band-limited matrix with bandwidth $d << n$ can be efficiently computed in $O(m^2d)$~\citep{KavcicMoura:00} (versus the naive $O(m^3)$).  Issues related to tapering the elements of $K^{-1}$ to zero while maintaining positive-semidefiniteness are discussed in~\cite{ZhangDu:08}. 

\section{Simulation Example}
\label{sec:sim}
In the following simulation examples, we aim to analyze the performance of the proposed Bayesian nonparametric covariance regression model relative to competing alternatives in terms of both covariance estimation and predictive performance.  We initially consider the case in which $\Sigma(x)$ is generated from the assumed nonparametric Bayes model in Section~\ref{sec:qualSim} and~\ref{sec:pred}, while in Section~\ref{sec:modelMismatch} we simulate from a parametric model and compare to a Wishart matrix discounting method~\citep{PradoWest} over a set of replicates. 
%
%
\subsection{Estimation Performance}
\label{sec:qualSim}
We simulated a dataset from the model as follows.  The set of predictors is a discrete set $\mathcal{X}=\{1,\dots,100\}$, with a 10-dimensional observation $y_i$ generated for each $x_i \in \mathcal{X}$.  The generating mechanism was based on weightings of a latent $5 \times 4$ dimensional matrix $\xi(\cdot)$ of Gaussian process dictionary functions (i.e, $L=5$, $k=4$), with length-scale $\kappa=10$ and an additional nugget effect adding $1e^{-5}I_n$ to $K$.  Here, we first scale the predictor space to $(0,1]$.  The additional latent mean dictionary elements $\psi(\cdot)$ were similarly distributed.  The weights $\Theta$ were simulated as specified in Eq.~\eqref{eqn:shrinkage} choosing $a_1=a_2=10$.  The precision parameters $\sigma_j^{-2}$ were each drawn independently from a $\mbox{Ga}(1,0.1)$ distribution with mean $10$.  Figure~\ref{fig:truecov} displays the resulting values of the elements of $\mu(x)$ and $\Sigma(x)$. 
\begin{figure}[t!] \centering 
	\begin{tabular}{cc}
		\includegraphics[width = 1.75in]{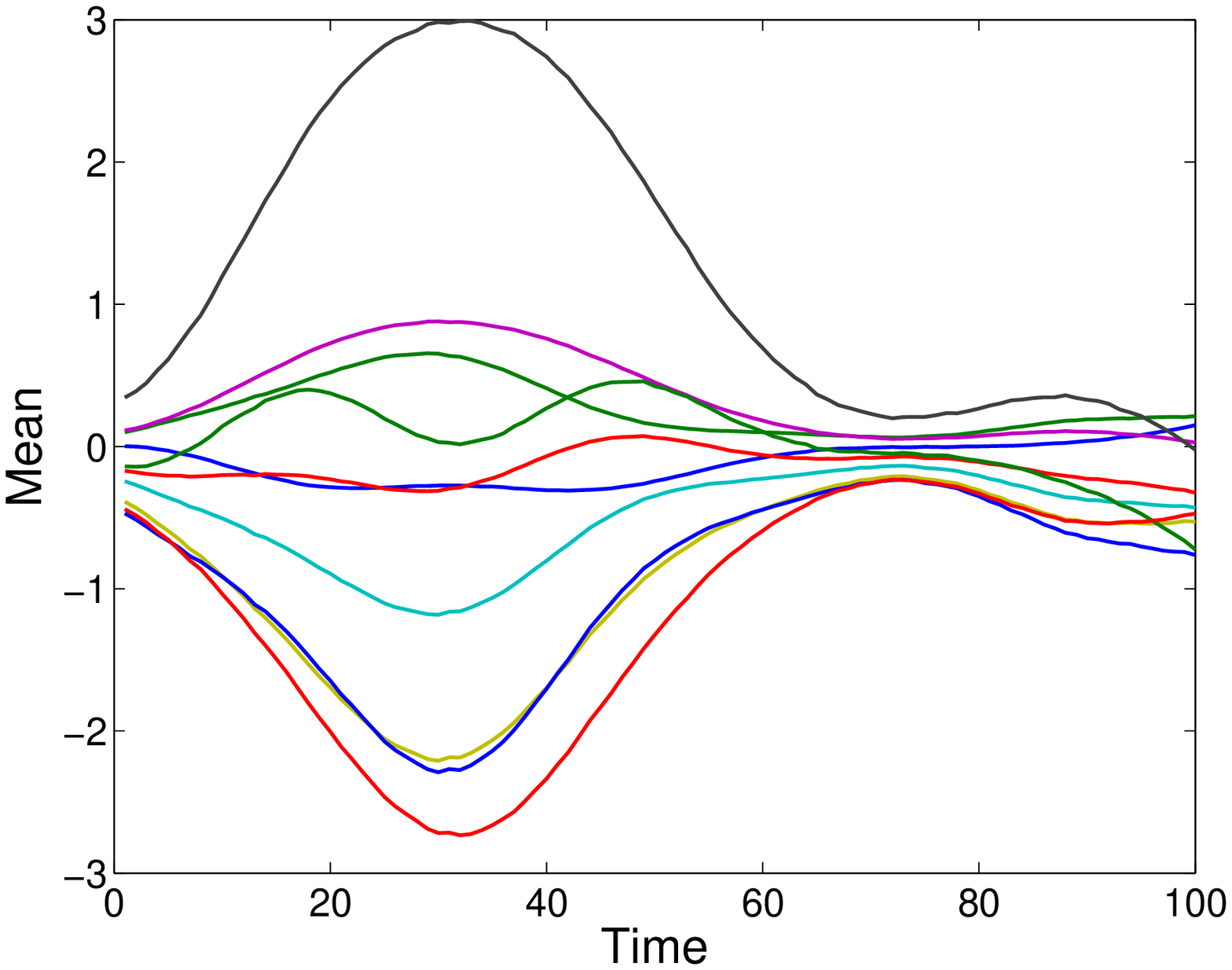} & \includegraphics[width = 1.75in]{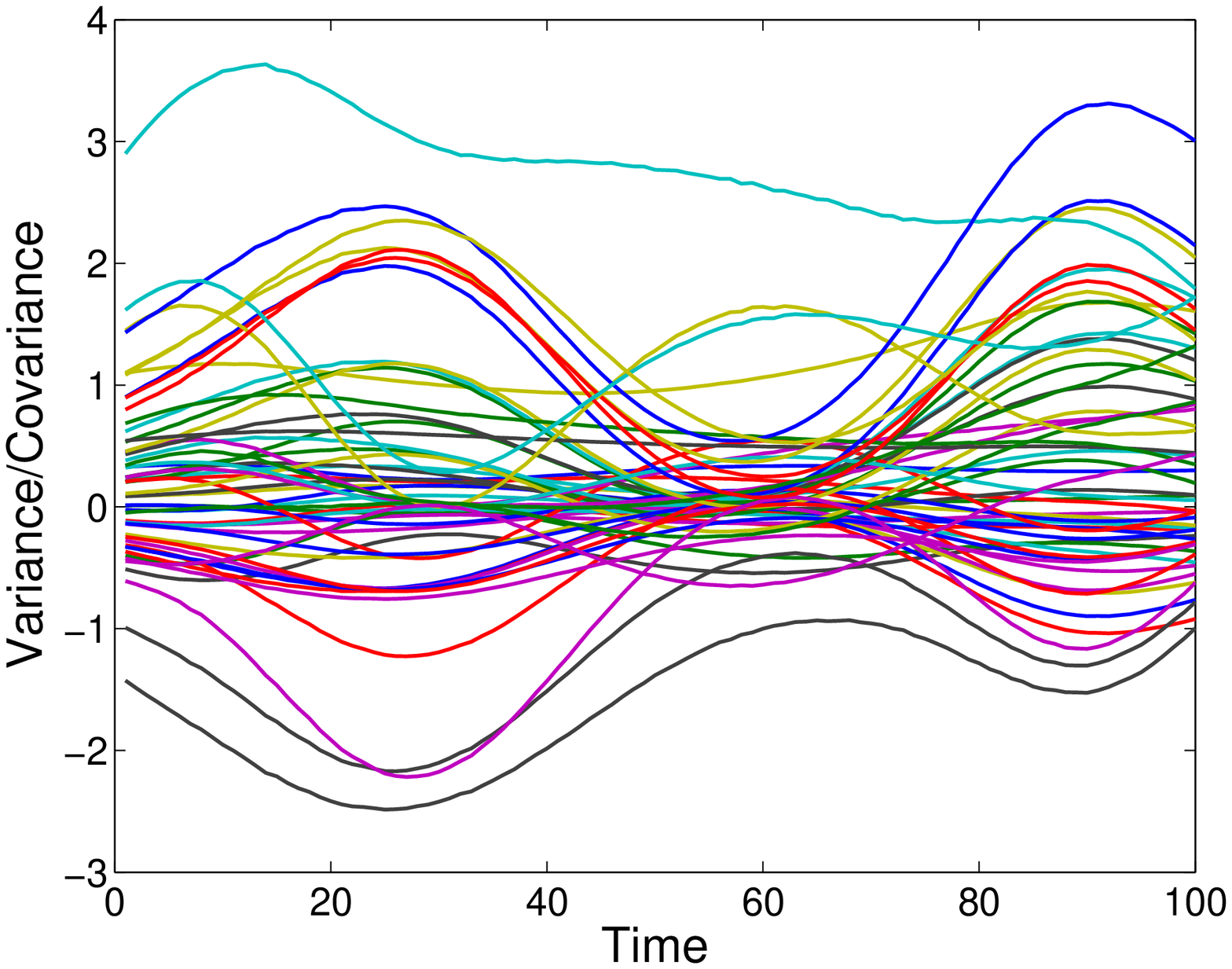}\\
		(a) & (b)
	\end{tabular}
	\caption{Plot of each component of the (a) true mean vector $\mu(x)$ and (b) true covariance matrix $\Sigma(x)$ over the predictor space $\mathcal{X}=\{1,\ldots,100\}$, taken here to represent a time index.} \label{fig:truecov} \postcap \vspace{0.1in}
\end{figure}
\begin{figure}[t!]
	\centering
	\begin{tabular}{ccc}
		\hspace{-0.2in}
		\includegraphics[width = 1.75in]{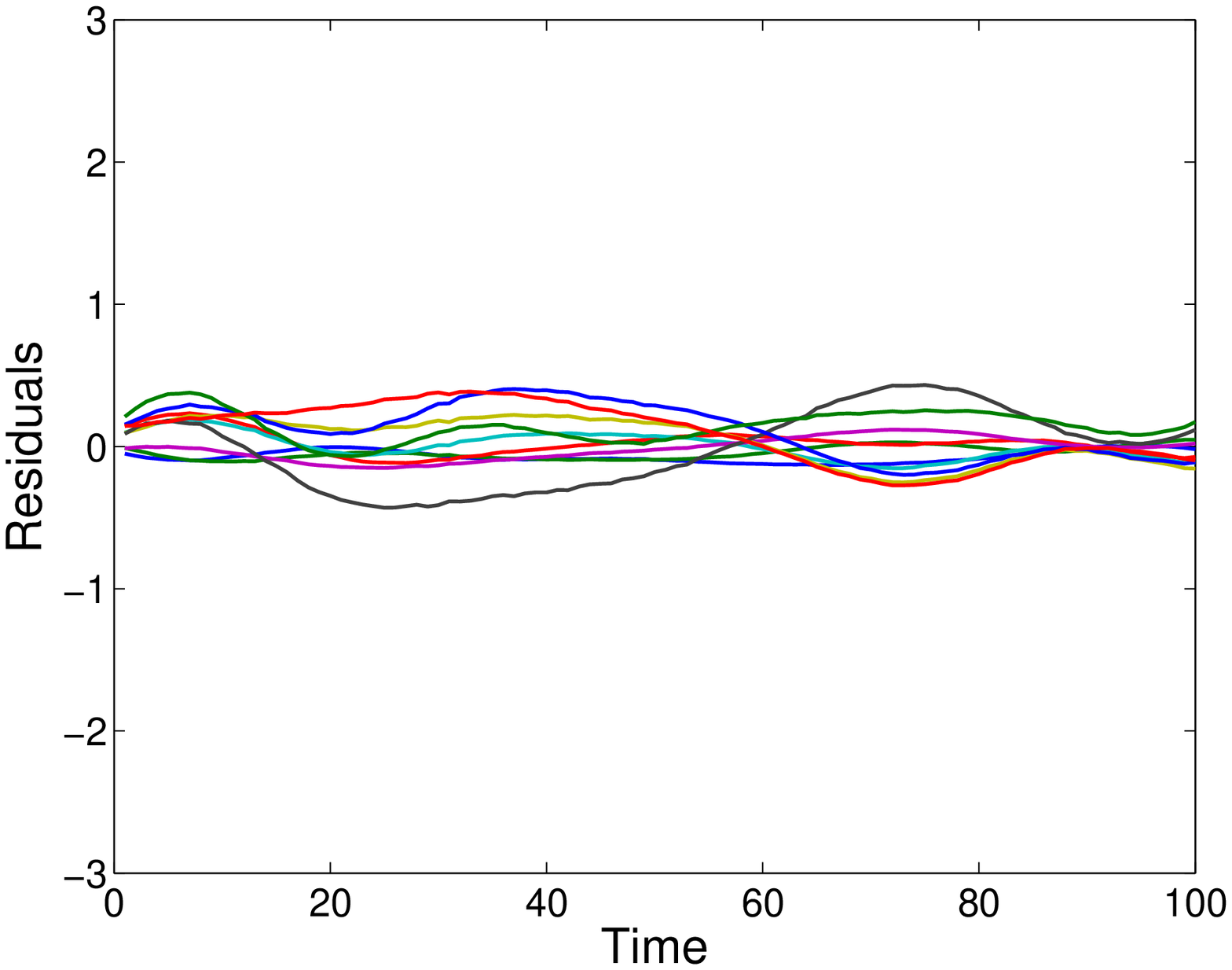} & \hspace{-0.2in}
		\includegraphics[width = 1.75in]{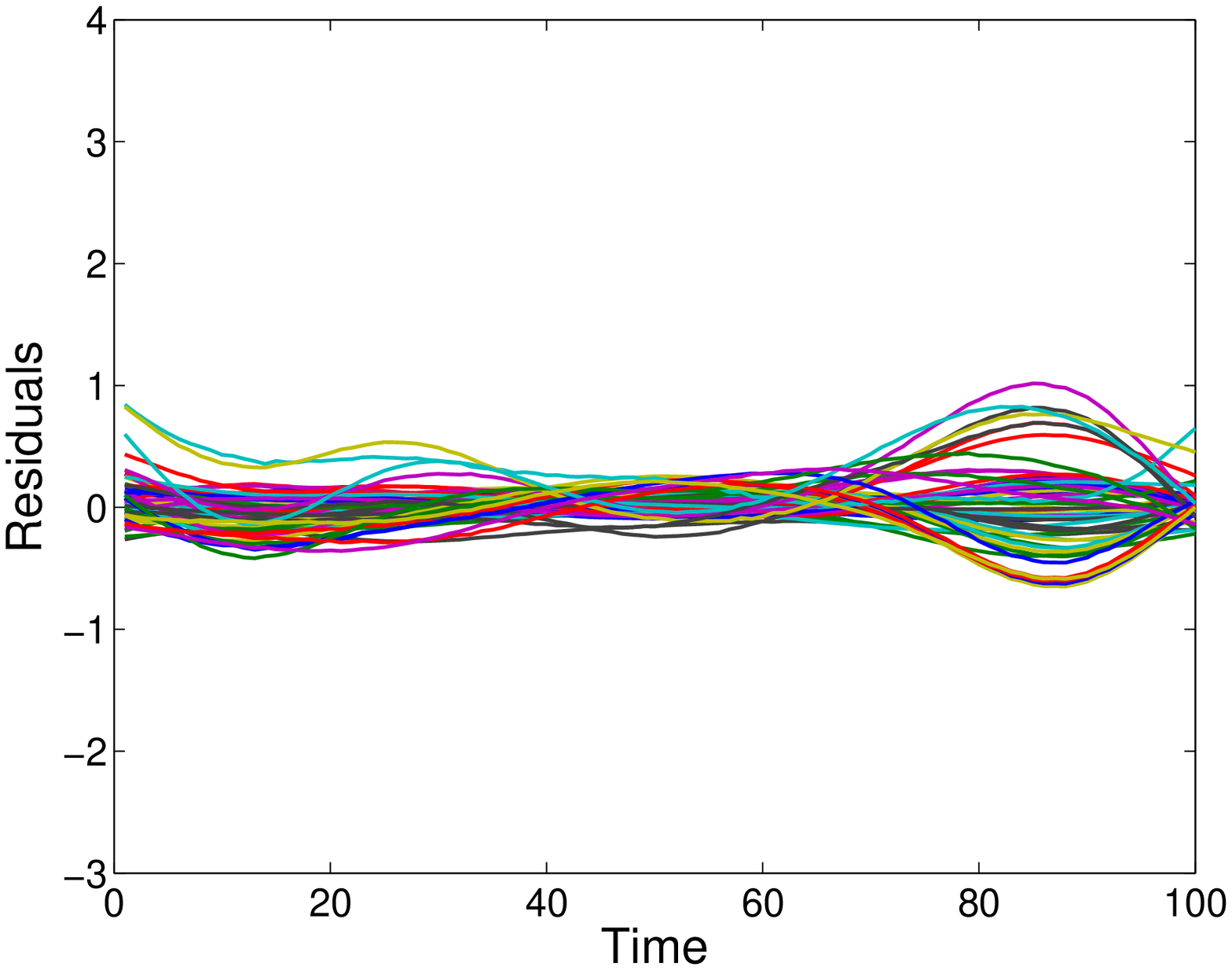} & \hspace{-0.2in}
		\includegraphics[width = 1.8in]{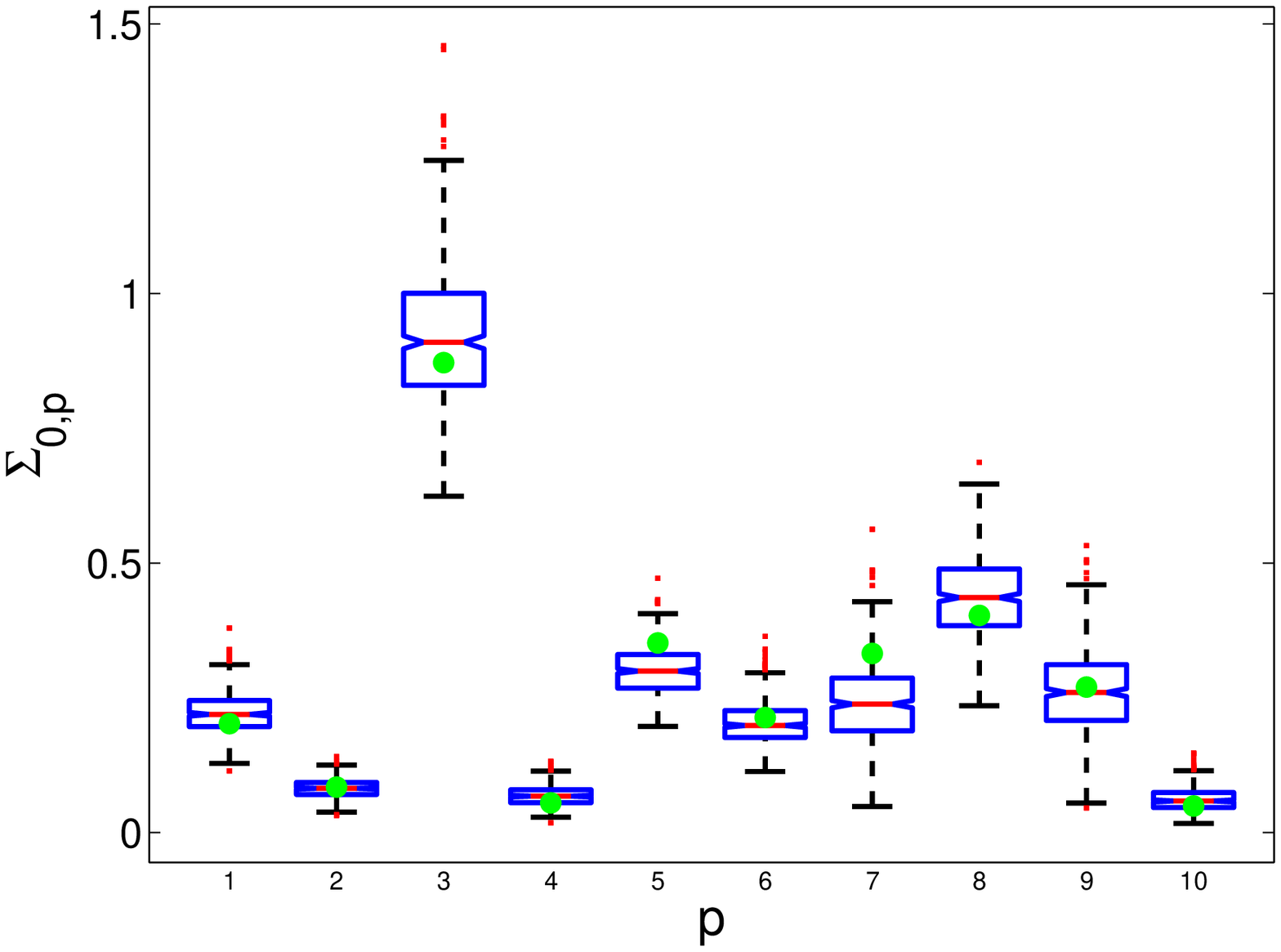} \\
		\hspace{-0.2in} (a) & \hspace{-0.2in}(b) & \hspace{-0.2in}(c)\\
	\end{tabular}
	\caption{Residuals between each component of the true and posterior mean of (a) the mean $\mu(x)$, and (b) covariance $\Sigma(x)$. The scaling of the axes matches that of Figure~\ref{fig:truecov}.  (c) Box plot of posterior samples of the noise covariance terms $\sigma_j^2$ for $j=1,\dots,p$ compared to the true value (green).}\label{fig:simStudy} \postcap \vspace{0.1in}
\end{figure}

For inference, we set the hyperparameters as follows.  We use truncation levels $k^*=L^*=10$, which we found to be sufficiently large from the fact that the last few columns of the posterior samples of $\Theta$ were consistently shrunk close to $0$.  We set $a_1=a_2=2$ and placed a $\mbox{Ga}(1,0.1)$ prior on the precision parameters $\sigma_j^{-2}$. The length-scale parameter $\kappa$ was set from the data according to the heuristic described in Section~\ref{sec:hyp} using 20 knots evenly spaced in $\mathcal{X}=\{1,\dots,100\}$, and was determined to be $10$ (after rounding).  

Although experimentally we found that our sampler was insensitive to initialization in lower-dimensional examples such as the one analyzed here, we use the following more intricate initialization for consistency with later experiments on larger datasets in which mixing becomes more problematic.  The predictor-independent parameters $\Theta$ and $\Sigma_0$ are sampled from their respective priors (first sampling the shrinkage parameters $\phi_{j\ell}$ and $\delta_h$ from their priors).  The variables $\eta_i$ and $\xi(x_i)$ are set via a data-driven initialization scheme in which an estimate of $\Sigma(x_i)$ for $i=1,\dots,n$ is formed using Steps 1-4 of Section~\ref{sec:hyp}. Then, $\Theta\xi(x_i)$ is taken to be a $k^*$-dimensional low-rank approximation to the Cholesky of the estimates of $\Sigma(x_i)$.  The latent factors $\eta_i$ are sampled from the posterior given in Equation~\eqref{eqn:etaPost} using this data-driven estimate of $\Theta\xi(x_i)$.  Similarly, the $\xi(x_i)$ are initially taken to be spline fits of the pseudo-inverse of the low-rank Cholesky at the knot locations and the sampled $\Theta$.  We then iterate a couple of times between sampling: (i) $\xi(\cdot)$ given $\{y_i\}$, $\Theta$, $\Sigma_0$, and the data-driven estimates of $\eta$, $\xi(\cdot)$; (ii) $\Theta$ given $\{y_i\}$, $\Sigma_0$, $\eta$, and the sampled $\xi(\cdot)$; (iii) $\Sigma_0$ given $\{y_i\}$, $\Theta$, $\eta$, and $\xi(\cdot)$; and (iv) determining a new data-driven approximation to $\xi(\cdot)$ based on the newly sampled $\Theta$.  Results indistinguishable from those presented here were achieved (after a short burn-in period) by simply initializing each of $\Theta$, $\xi(\cdot)$, $\Sigma_0$, $\eta_i$, and the shrinkage parameters $\phi_{j\ell}$ and $\delta_h$ from their respective priors.    

We ran 10,000 Gibbs iterations and discarded the first 5,000 iterations.  We then thinned the chain every 10 samples.  The residuals between the true and posterior mean over all components are displayed in Figure~\ref{fig:simStudy}(a) and (b).  Figure~\ref{fig:simStudy}(c) compares the posterior samples of the elements $\sigma_j^2$ of the noise covariance $\Sigma_0$ to the true values.  Finally, in Figure~\ref{fig:samplePaths} we display a select set of plots of the true and posterior mean of components of $\mu(x)$ and $\Sigma(x)$, along with the 95\% highest posterior density intervals computed at each predictor value $x=1,\dots,100$.  

\begin{figure}[t!]
	\centering
	\begin{tabular}{ccc}
		\hspace{-0.2in}
		\includegraphics[width = 1.75in]{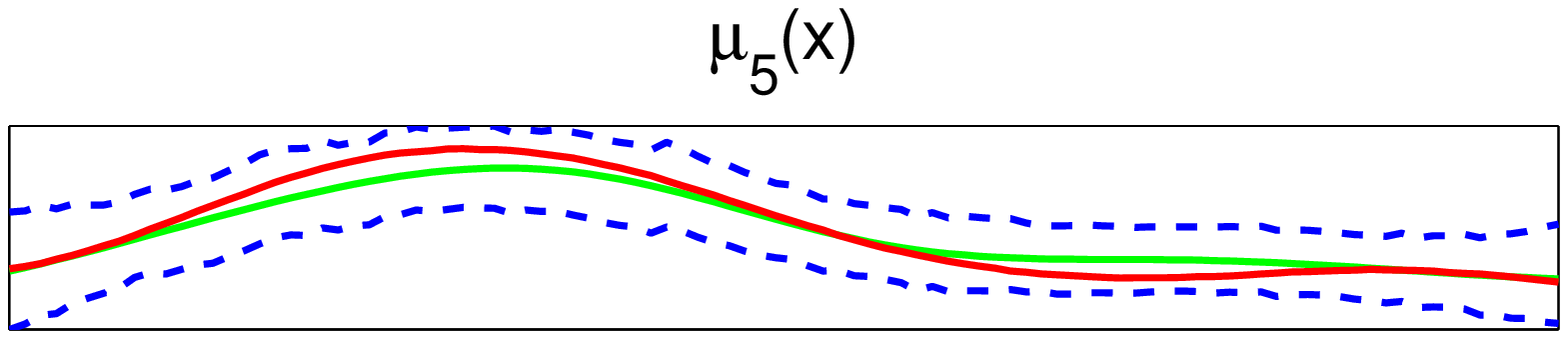} & \hspace{-0.2in}
		\includegraphics[width = 1.75in]{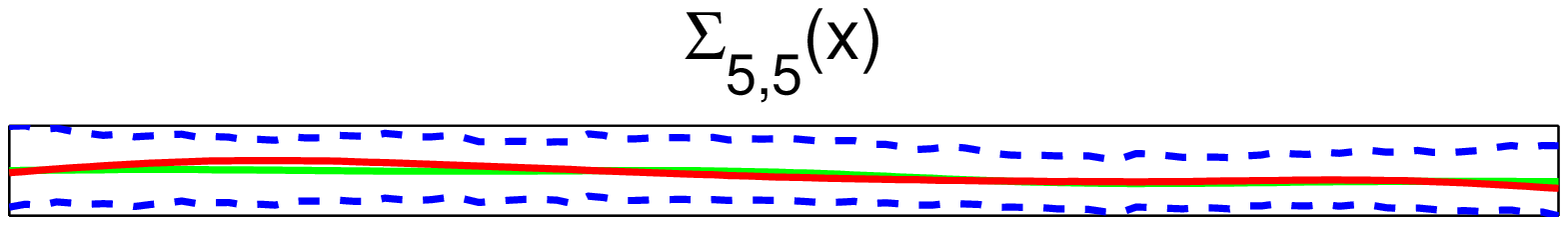} & \hspace{-0.2in}
		\includegraphics[width = 1.75in]{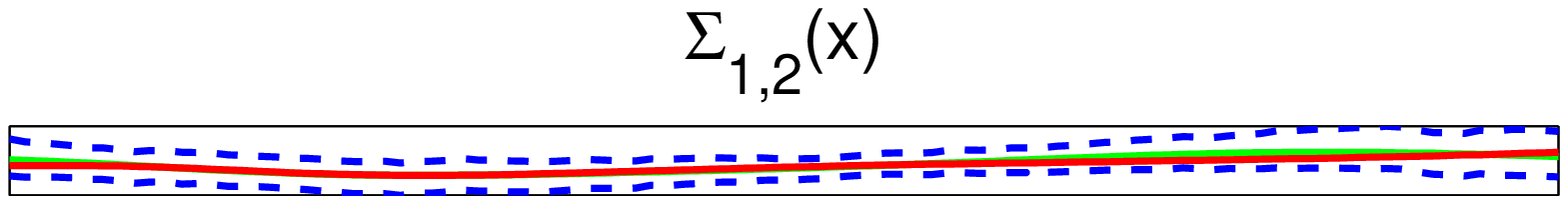} \\
		\hspace{-0.2in}
		\includegraphics[width = 1.75in]{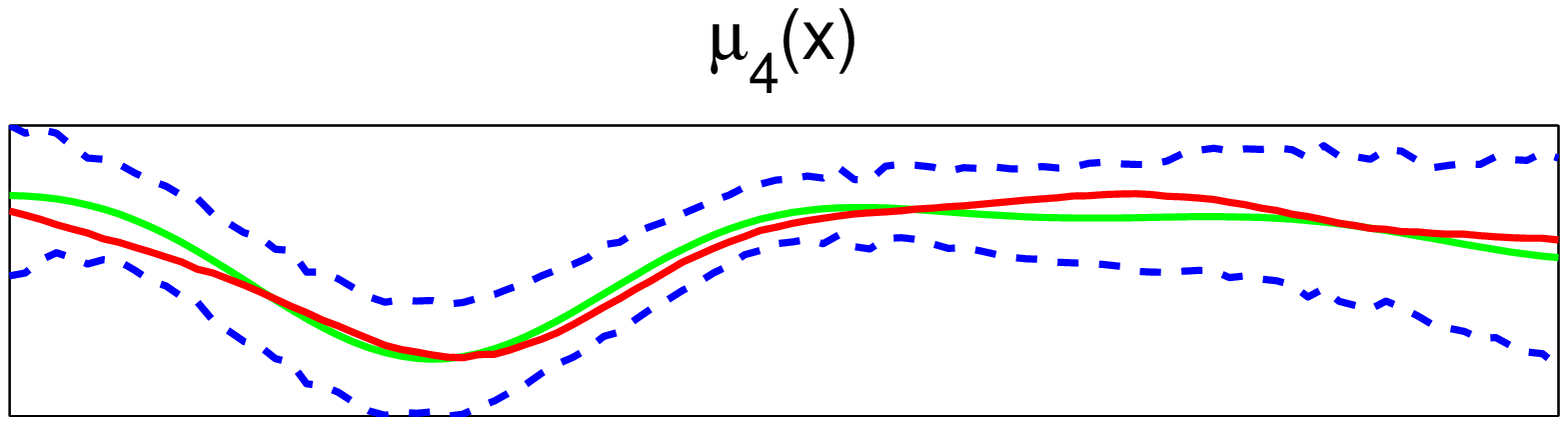} & \hspace{-0.2in}
		\includegraphics[width = 1.75in]{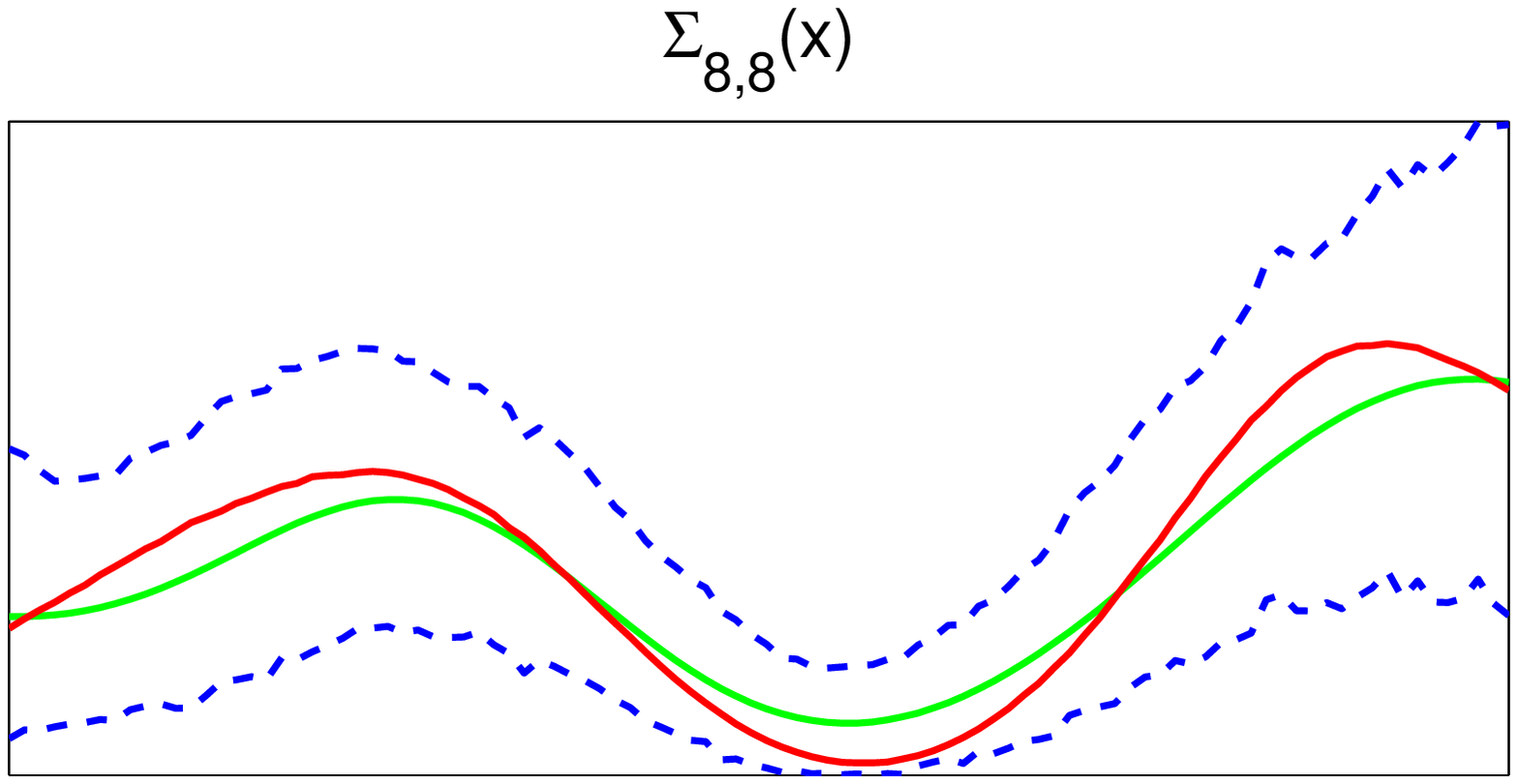} & \hspace{-0.2in}
		\includegraphics[width = 1.75in]{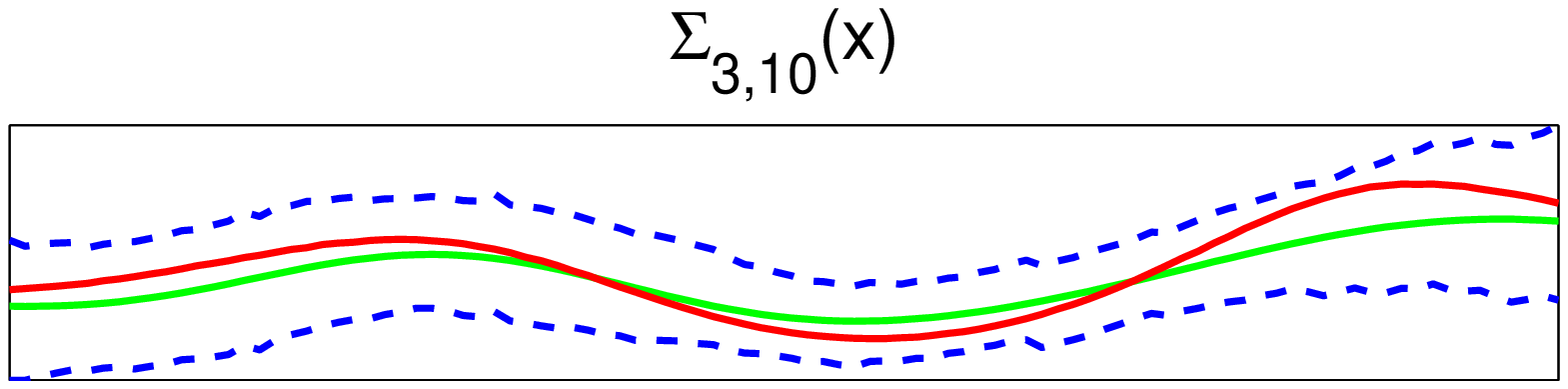} \\
		\hspace{-0.2in}
		\includegraphics[width = 1.75in]{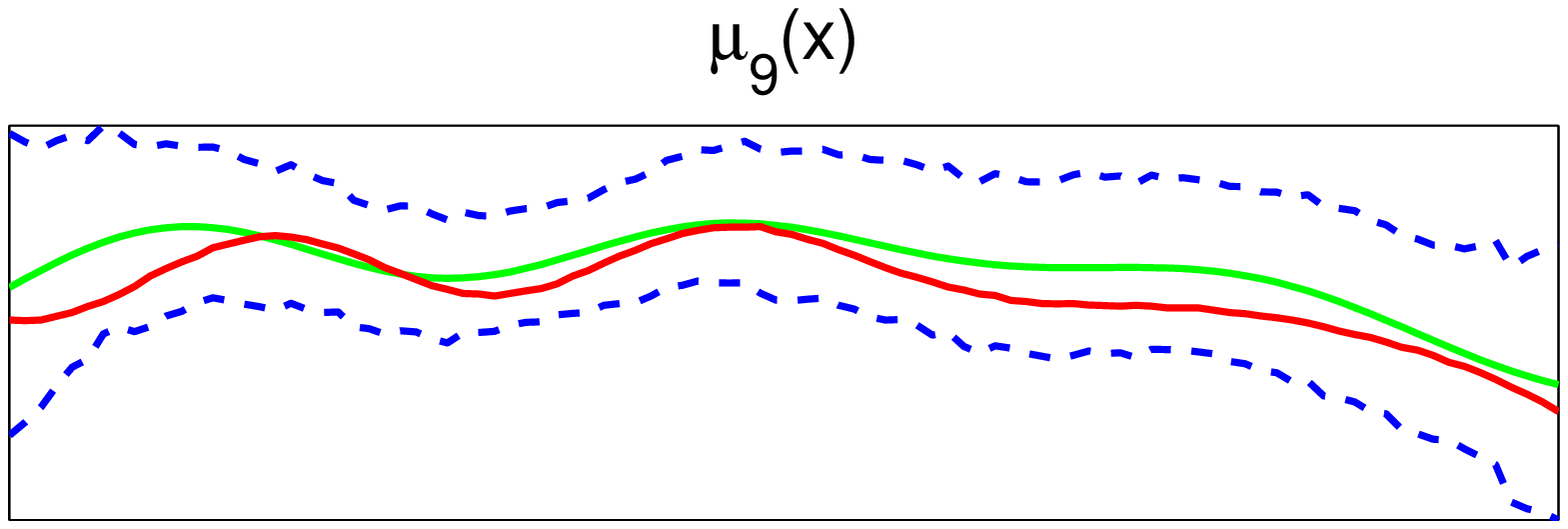} & \hspace{-0.2in}
		\includegraphics[width = 1.75in]{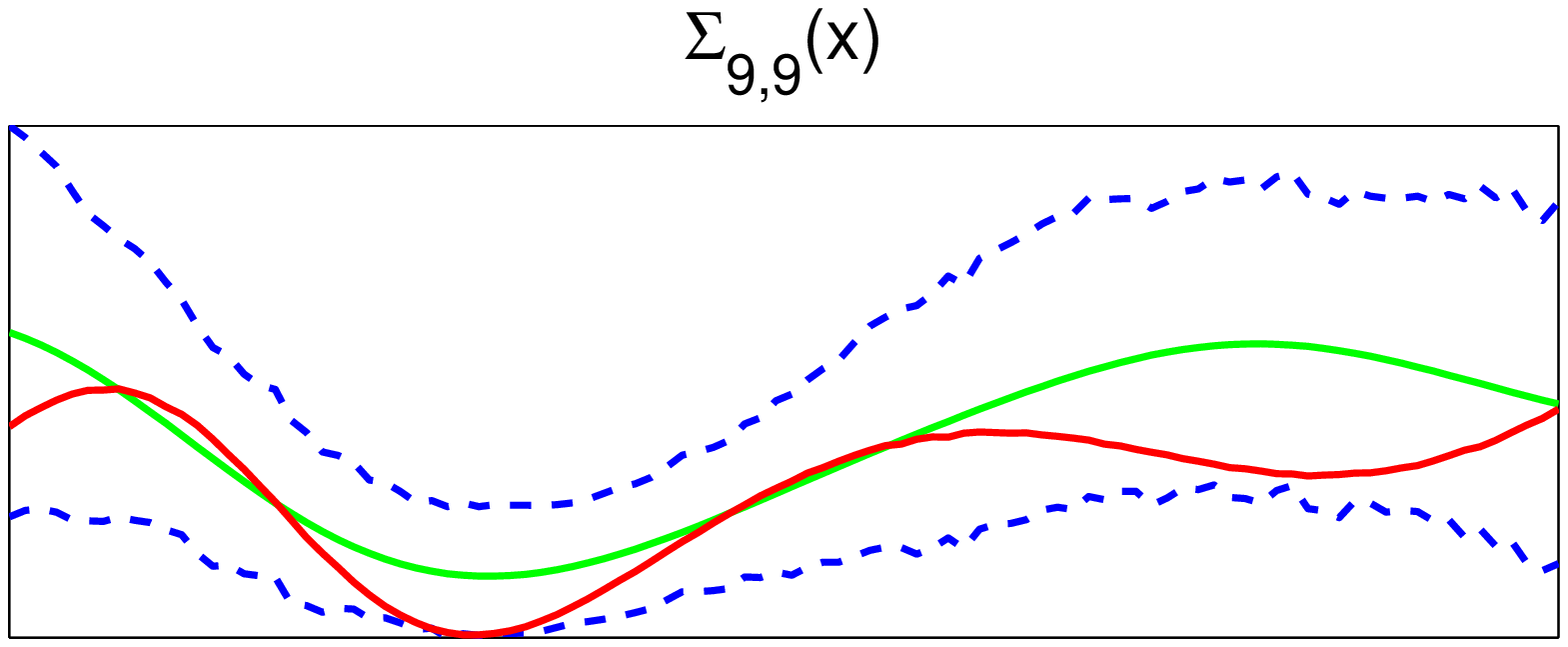} & \hspace{-0.2in}
		\includegraphics[width = 1.75in]{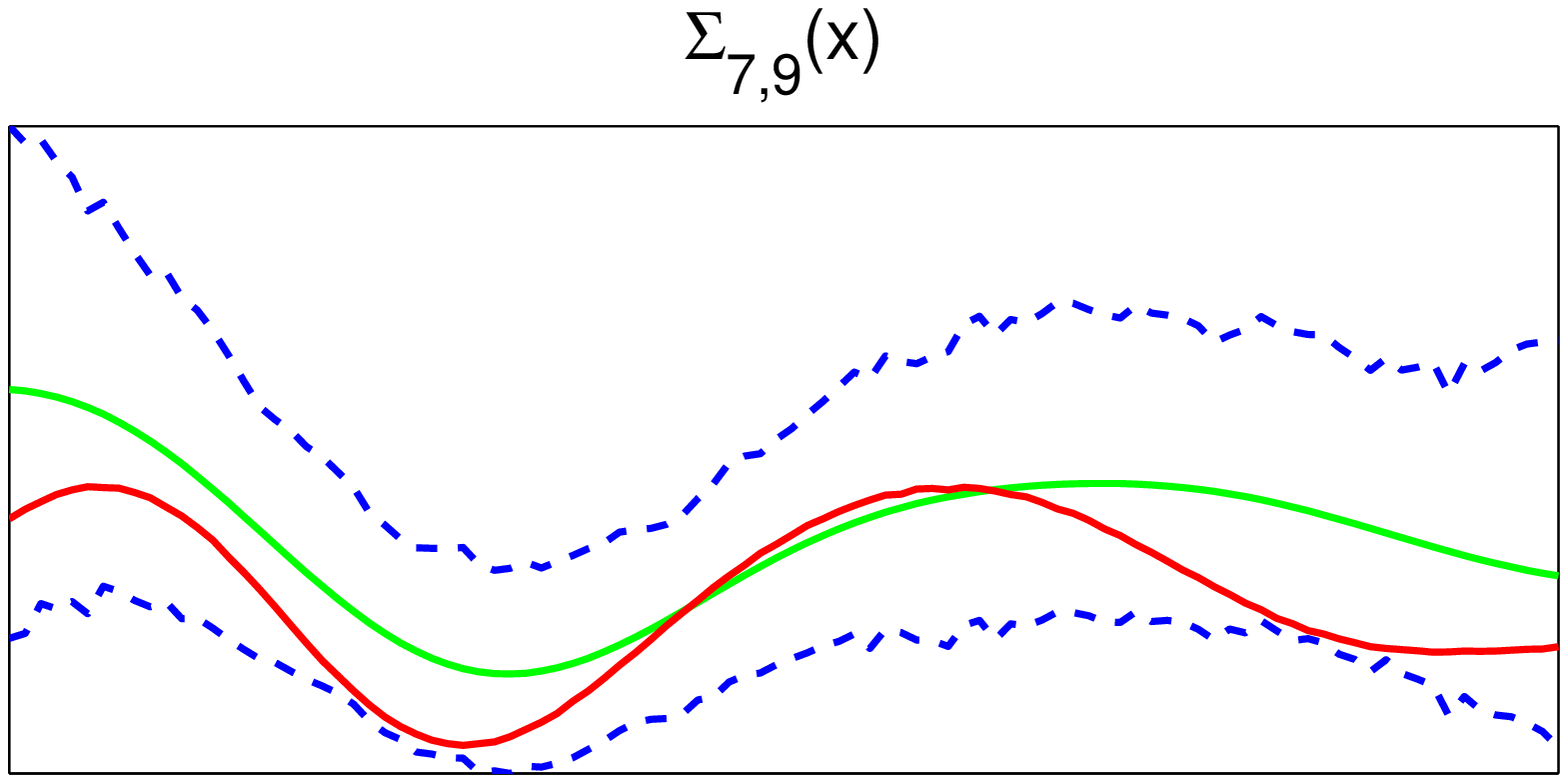} \\
	\end{tabular}
	\caption{Plots of truth (red) and posterior mean (green) for select components of the mean $\mu_p(x)$ (\emph{left}), variances $\Sigma_{pp}(x)$ (\emph{middle}), and covariances $\Sigma_{pq}(x)$ (\emph{right}).  The point-wise 95\% highest posterior density intervals are shown in blue.  The top row represents the component with the lowest $L2$ error between the truth and posterior mean. Likewise, the middle row represents median $L2$ error and the bottom row the worst $L2$ error. The size of the box indicates the relative magnitudes of each component.} \label{fig:samplePaths} \postcap \vspace{0.1in}
\end{figure}

From the plots of Figures~\ref{fig:simStudy} and~\ref{fig:samplePaths}, we see that we are clearly able to capture heteroscedasticity in combination with a nonparametric mean regression. The true values of the mean and covariance components are (even in the worst case) contained within the 95\% highest posterior density intervals, with these intervals typically small such that the overall interval bands are representative of the shape of the given component being modeled.

\subsection{Predictive Performance}
\label{sec:pred}
Capturing heteroscedasticity can significantly improve estimates of the predictive distribution of new observations or missing data.  To explore these gains within our proposed Bayesian nonparametric covariance regression framework, we compare against two possible homoscedastic formulations that each assume $y \sim \mathcal{N}(\mu(x),\Sigma)$.  The first is a standard Gaussian process mean regression model with each element of $\mu(x)$ an independent draw from $\mbox{GP}(0,c)$.  The second builds on our proposed regularized latent factor regression model and takes $\mu(x)=\Theta\xi(x)\psi(x)$, with $\{\Theta,\xi(x),\psi(x)\}$ as defined in the heteroscedastic case.  However, instead of having a predictor-dependent covariance $\Sigma(x) = \Theta\xi(x)\xi(x)'\Theta' + \Sigma_0$, the homoscedastic model assumes that $\Sigma$ is an arbitrary covariance matrix constant over predictors.  By comparing to this latter homoscedastic model, we can directly analyze the benefits of our heteroscedastic model since both share exactly the same mean regression formulation.  For each of the homoscedastic models, we place an inverse Wishart prior on the covariance $\Sigma$.

We analyze the same simulation dataset as described in Section~\ref{sec:qualSim}, but randomly remove approximately 5\% of the observations.  Specifically, independently for each element $y_{ij}$ (i.e., the $j$th response component at predictor $x_i$) we decide whether to remove the observation based on a $\mbox{Bernoulli}(p_i)$ draw.  We chose $p_i$ to be a function of the matrix norm of the true covariance at $x_i$ to slightly bias towards removing values from predictor regions with a tighter distribution.  This procedure resulted in removing 48 of the 1000 response elements.  

Table~\ref{table:KLdiv} compares the average Kullback-Leibler divergence $D_{KL}(P_{i,m}||Q_i)$, $i=1,\dots,100$, for the following definitions of $P_{i,m}$ and $Q_i$.  The distribution $Q_i$ is the predictive distribution of all missing elements $y_{ij}$ given the observed elements of $y_i$ under the true parameters $\mu(x_i)$ and $\Sigma(x_i)$.  Likewise, $P_{i,m}$ is taken to be the predictive distribution based on the $m$th posterior sample of $\mu(x_i)$ and $\Sigma(x_i)$.  In this scenario, the missing observations $y_{ij}$ are imputed as an additional step in the MCMC computations\footnote{Note that it is not necessary to impute the missing $y_{ij}$ within our proposed Bayesian covariance regression model because of the conditional independencies at each Gibbs step.  In Section~\ref{sec:app}, we simply sample based only on actual observations.  Here, however, we impute in order to directly compare our performance to the homoscedastic models.}.  The results, once again based on 10,000 Gibbs iterations and discarding the first 5,000 samples, clearly indicate that our Bayesian nonparametric covariance regression model provides more accurate predictive distributions.  We additionally note that using a regularized latent factor approach to mean regression improves on the naive homoscedastic model in high dimensional datasets in the presence of limited data.  Not depicted in this paper due to space constraints is the fact that the proposed covariance regression model also leads to improved estimates of the mean $\mu(x)$ in addition to capturing heteroscedasticity.
\begin{table}
\caption{\label{table:KLdiv}Average Kullback-Leibler divergence $D_{KL}(P_{i,m}||Q_i)$, $i=1,\dots,100$, where $P_{i,m}$ and $Q_i$ are the predictive distributions of all missing elements $y_{ij}$ given the observed elements of $y_i$ based on the $m$th posterior sample of and true parameters $\mu(x_i)$ and $\Sigma(x_i)$, respectively. We compare the predictive performance for two homoscedastic models to our covariance regression framework.}
\centering
\fbox{%
\begin{tabular}{*{2}{|c|}}
\em Model & \em Average Posterior Predictive KL Divergence\\
\hline
Homoscedastic Mean Regression & 0.3409\\
Homoscedastic Latent Factor Mean Regression & 0.2909\\
Heteroscedastic Mean Regression & \textbf{0.1216}\\
\end{tabular}}
\end{table}
\subsection{Model Mismatch}
\label{sec:modelMismatch}
We now examine our performance over a set of replicates from a 30-dimensional \emph{parametric} heteroscedastic model.  To generate the covariance $\Sigma(x)$ over $\mathcal{X} = \{1,\dots,500\}$, we chose a set of 5 evenly spaced knots $x_k = 1, 125, 250, 375, 500$ and generated
\begin{align}
	S(x_k) \sim \mathcal{N}(0,\Sigma_s)
\end{align}
with $\Sigma_s = \sum_{j=1}^{30}s_j s_j'$ and $s_j \sim \mathcal{N}([-29 \,\, -28 \,\, \dots \,\, 28 \,\, 29]',I_{30})$.  This construction implies that $S(x_k)$ and $S(x_k')$ are correlated.  We then fit a spline $\tilde{S}_{ij}(\cdot)$ independently through each element $S_{ij}(x_k)$ and evaluate this spline fit at $x=1,\dots,500$. The covariance is constructed as
\begin{align}
	\Sigma(x) = \alpha\tilde{S}(x)\tilde{S}(x)' + \Sigma_0,
\end{align}
where $\Sigma_0$ is a diagonal matrix with a truncated-normal prior, $\mathcal{TN}(0,1)$, on its diagonal elements. The constant $\alpha$ is chosen to scale the maximum value of $\alpha\tilde{S}(x)\tilde{S}(x)'$ to 1.  The resulting covariance is shown in Figure~\ref{fig:splineCov}(a).
\begin{figure}[t!] \centering 
	\begin{tabular}{ccc}
		\includegraphics[width = 1.75in]{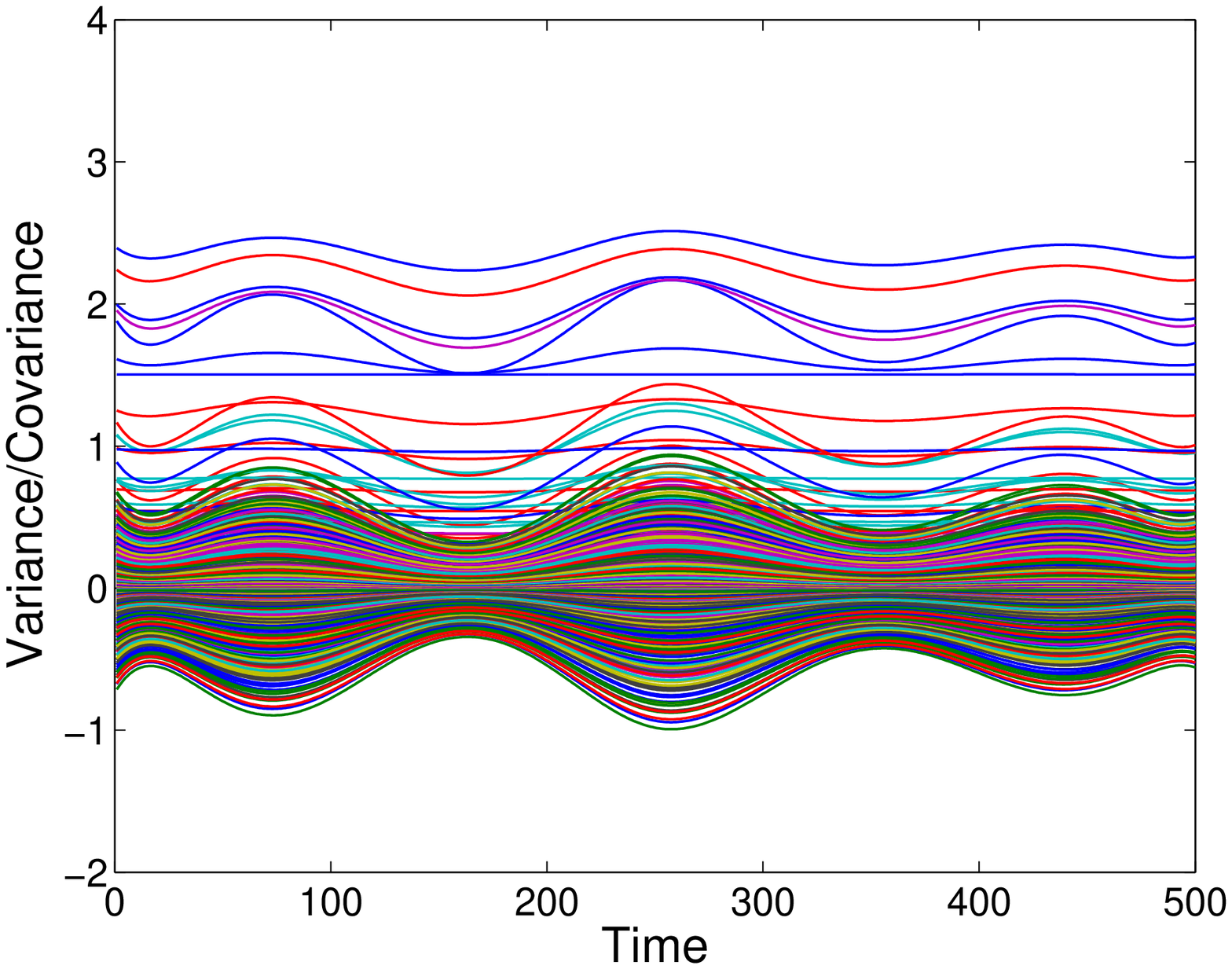} & \includegraphics[width = 1.75in]{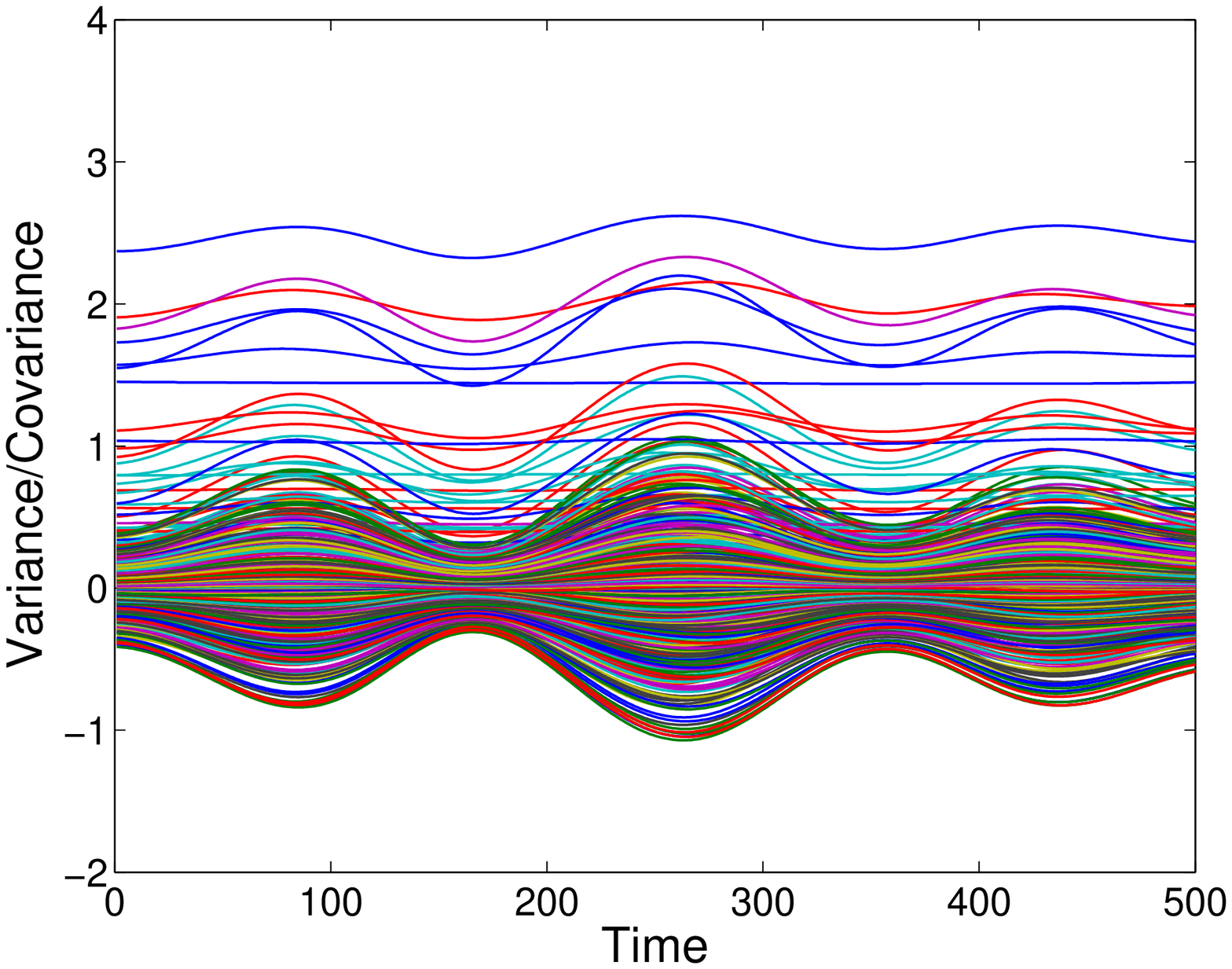} & \includegraphics[width = 1.75in]{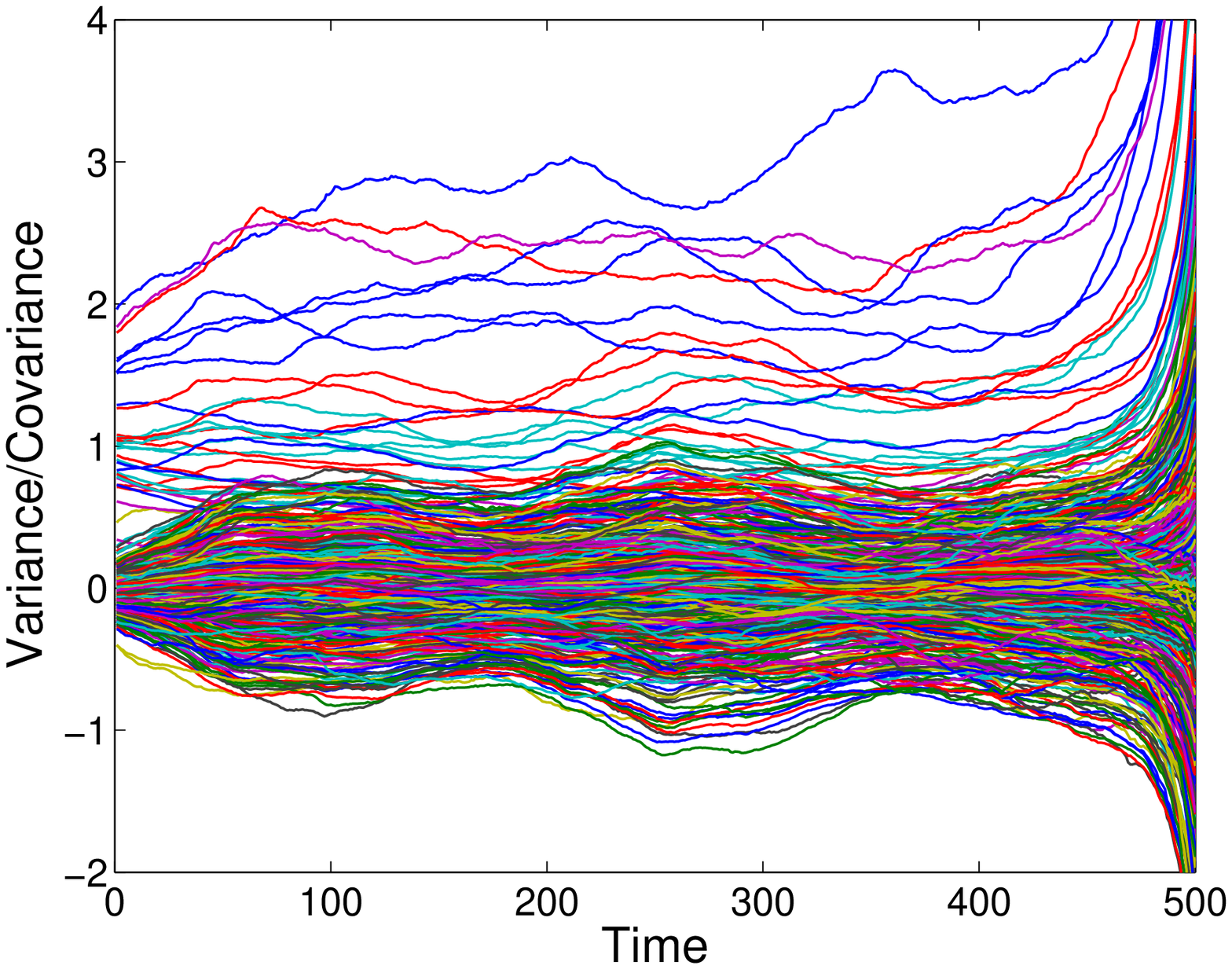}\\
		(a) & (b) & (c)
	\end{tabular}
	\caption{(a) Plot of each component of the true covariance matrix $\Sigma(x)$ over the predictor space $\mathcal{X}=\{1,\ldots,500\}$, taken to represent a time index. Analogous plot for the mean estimate of $\Sigma(x)$ are shown for (b) our proposed Bayesian nonparametric covariance regression model based on Gibbs iterations 5000 to 10000, and (c) a Wishart matrix discounting model over 100 independent FFBS samples.  Both mean estimates of $\Sigma(x)$ are for a single replicate $\{y^{(1)}_i\}$.  Note that the scaling of (c) crops the large estimates of $\Sigma(x)$ for $x$ near 500.} \label{fig:splineCov} \postcap \vspace{0.1in}
\end{figure}

Each replicate $m=1,\dots,30$ of this parametric heteroscedastic model is generated as
\begin{align}
	y_i^{(m)} \sim \mathcal{N}(0,\Sigma(x_i)).
\end{align}

Our hyperparameters and initialization scheme are exactly as in Section~\ref{sec:sim}.  The only difference is that we use truncation levels $k^*=L^* = 5$ based on an initial analysis with $k^*=L^*=17$.  For each replicate, we once again run 10,000 Gibbs iterations and thin the chain by examining every $10$th sample.  A mean estimate of $\Sigma(x)$ is displayed in Figure~\ref{fig:splineCov}(b). In Figure~\ref{fig:splineNorm}, we plot the mean and 95\% highest posterior density intervals of the Frobenius norm $||\Sigma^{(\tau,m)}(x)-\Sigma(x)||_2$ aggregated over iterations $\tau=9,000, \dots, 10,000$ and replicates $m=1,\dots,30$.  The average norm error over $\mathcal{X}$ is around 3, which is equivalent to each element of the inferred $\Sigma^{(\tau,m)}(x)$ deviating from the true $\Sigma(x)$ by 0.1.  Since the covariance elements are approximately in the range of $[-1,1]$ and the variances in $[0,3]$, these norm error values indicate very good estimation performance.

We compare our performance to that of the Wishart matrix discounting model (see Section 10.4.2 of~\cite{PradoWest}), which is commonly used in stochastic volatility modeling of financial time series.  Let $\Phi_t = \Sigma_t^{-1}$.  The Wishart matrix discounting model is a discrete-time covariance evolution model that accounts for the slowly changing covariance by discounting the cumulated information.  Specifically, assume $\Phi_{t-1}\mid y_{1:t-1},\beta \sim W(h_{t-1},D_{t-1}^{-1})$, with $D_t = \beta D_{t-1} + y_ty_t'$ and $h_t = \beta h_{t-1} + 1$.  The discounting model then specifies
\begin{align}
	\Phi_t \mid y_{1:t-1},\beta \sim W(\beta h_{t-1}, (\beta D_{t-1})^{-1})
\end{align}
such that $E[\Phi_t\mid y_{1:t-1}] = E[\Phi_{t-1}\mid y_{1:t-1}] = h_{t-1}D_{t-1}^{-1}$, but with certainty discounted by a factor determined by $\beta$.  The update with observation $y_t$ is conjugate, maintaining a Wishart posterior on $\Phi_t$. A limitation of this construction is that it constrains $h_t > p-1$ (or $h_t$ integral) implying that $\beta > (p-2)/(p-1)$.  We set $h_0 = 40$ and $\beta = 1-1/h_0$ such that $h_t = 40$ for all $t$ and ran the forward filtering backward sampling (FFBS) algorithm outlined in~\cite{PradoWest}, generating 100 independent samples.  A mean estimate of $\Sigma(x)$ is displayed in Figure~\ref{fig:splineCov}(c) and the Frobenius norm error results are depicted in Figure~\ref{fig:splineNorm}.  Within the region $x=1,\dots,400$, we see that the error of the Wishart matrix discounting method is approximately twice that of our proposed methodology.  Furthermore, towards the end of the time series (interpreting $\mathcal{X}$ as representing a batch of time), the estimation error is especially poor due to errors accumulated in forward filtering.  Increasing $h_t$ mitigates this problem, but shrinks the model towards homoscedasticity.  In general, the formulation is sensitive to the choice of $h_t$, and in high-dimensional problems this degree of freedom is forced to take large (or integral) values.
\begin{figure}[t!] \centering 
	\begin{tabular}{cc}
		\includegraphics[width = 1.75in]{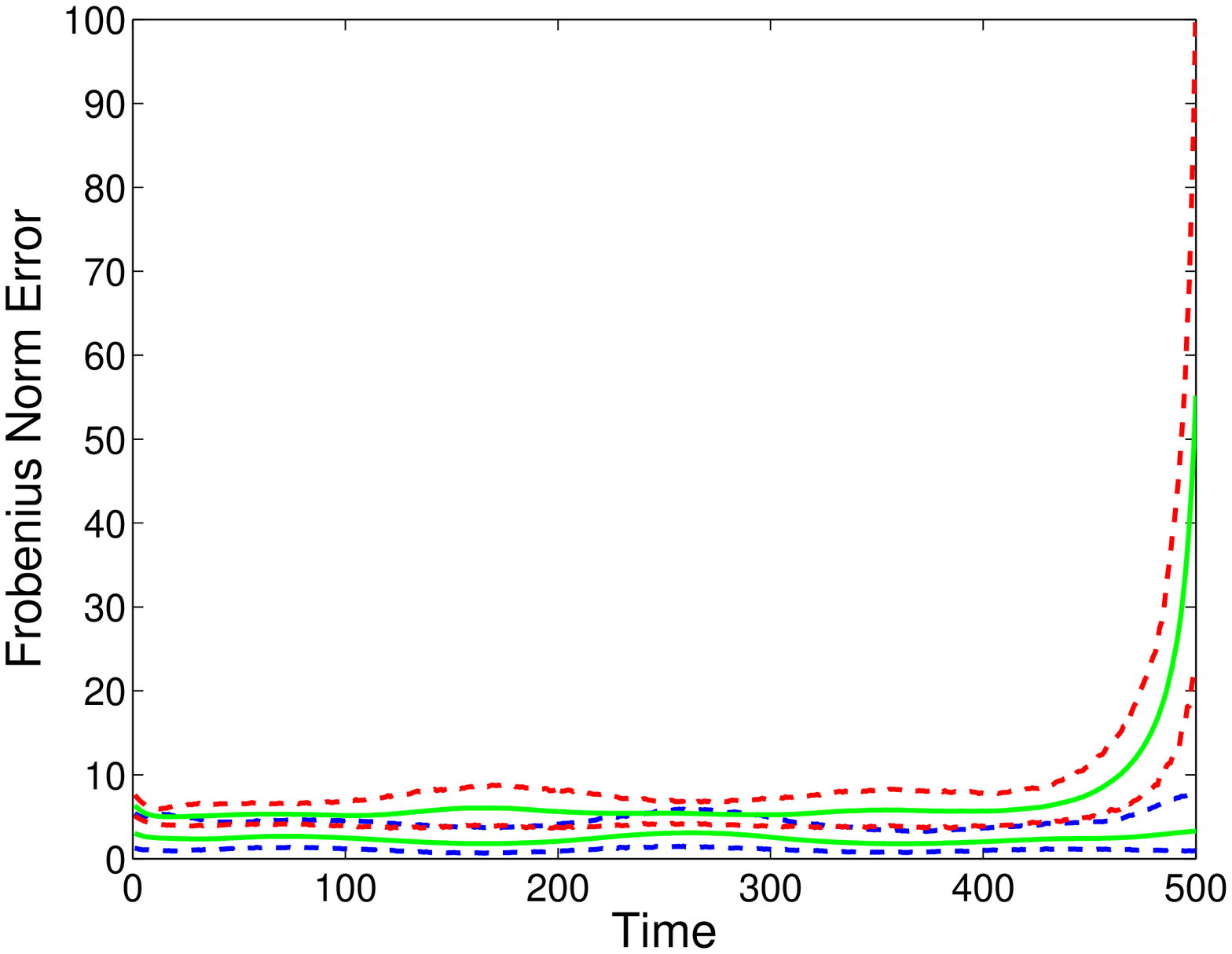} & \includegraphics[width = 1.75in]{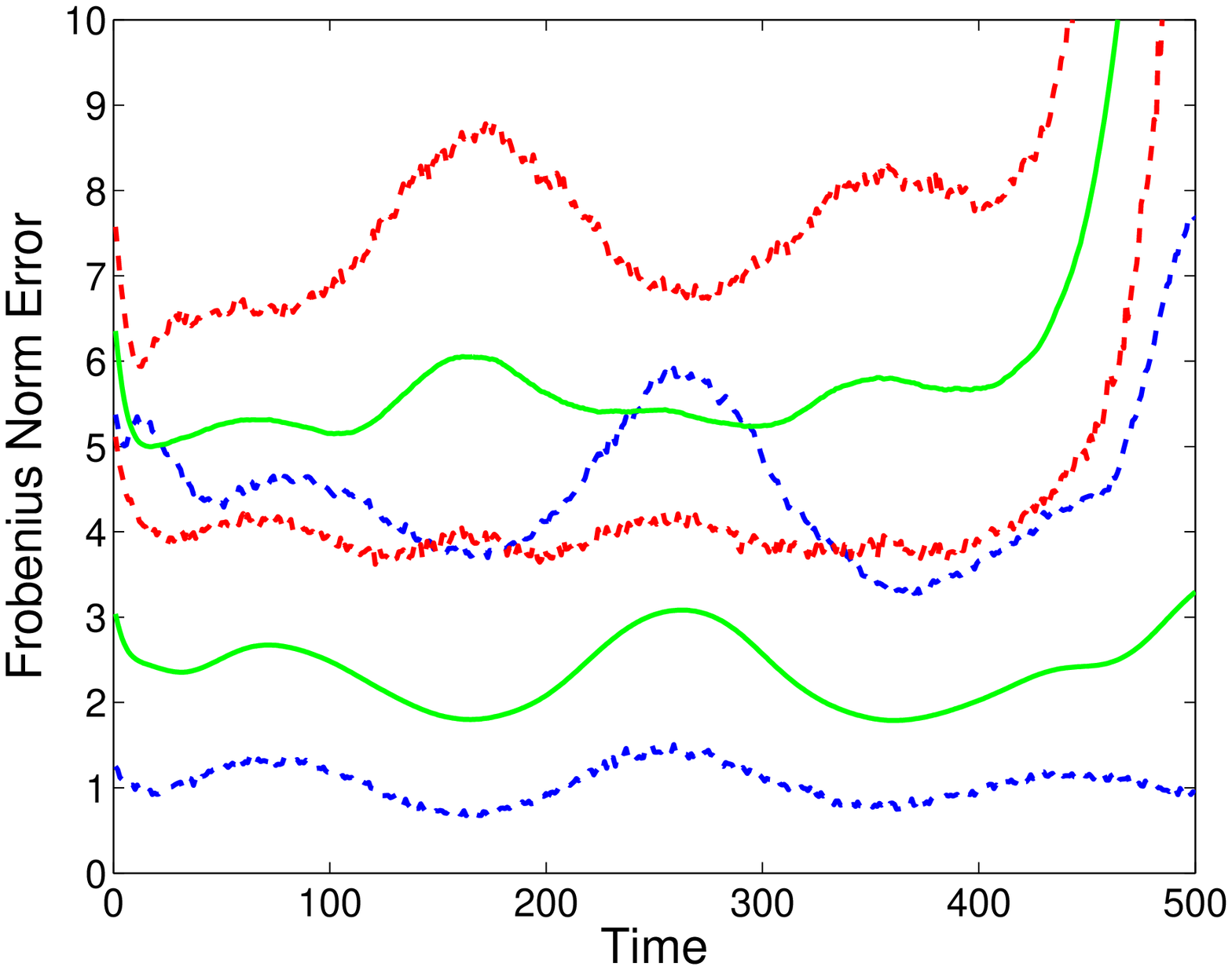}\\
		(a) & (b)
	\end{tabular}
	\caption{(a) Plot of the mean and 95\% highest posterior density intervals of the Frobenius norm $||\Sigma^{(\tau,m)}(x)-\Sigma(x)||_2$ for the proposed Bayesian nonparametric covariance regression model (blue and green) and the Wishart matrix discounting model (red and green).  The results are aggregated over 100 posterior samples and replicates $m=1,\dots,30$. For the Bayesian nonparametric covariance regression model, these samples are taken at iterations $\tau=[9000:10:10000]$. (b) Analogous plot, but zoomed in to more clearly see the differences over the range of $x=1,\dots,400$.} \label{fig:splineNorm} \postcap \vspace{0.1in}
\end{figure}
\section{Applications}
\label{sec:app}
We applied our Bayesian nonparametric covariance regression model to the problem of capturing spatio-temporal structure in influenza rates in the United States (US).  Surveillance of influenza has been of growing interest following a series of pandemic 
scares (e.g., SARS and avian flu) and the 2009 H1N1 pandemic, previously known as ``swine flu''.  Although influenza pandemics have a long history, such as the 1918-1919 ``Spanish flu'', 1957-1958 ``Asian flu'', and 1968-1969 ``Hong Kong flu'', a convergence of factors are increasing the current public interest in influenza surveillance.  These include both practical reasons such as the rapid rate by which geographically distant cases of influenza can spread worldwide, along with other driving factors such as an increased media coverage.
\subsection{CDC Influenza Monitoring}
The surveillance of influenza within the US is coordinated by the Centers for Disease Control and Prevention (CDC), which collects data from a large network of diagnostic laboratories, hospitals, clinics, individual healthcare providers, and state health departments (see http://www.cdc.gov/flu/weekly/).  The approximately 3,000 participating outpatient sites, collectively referred to as the US Outpatient Influenza-Like Illness Surveillance Network (ILINet), provide the CDC with key information about rates of influenza-like illness (ILI)\footnote{An influenza-like illness (ILI) is defined as any case of a person having over 100 degrees Fahrenheit fever along with a cough and/or sore throat in absence of any other known cause.}.  The CDC consolidates the ILINet observed cases and produces reports for 10 geographic regions in addition to a US aggregate rate based on a population-based weighted average of state-level rates.  The CDC weekly flu reports are typically released after a 1-2 week delay and are subject to retroactive adjustments based on corrected ILINet reports.  

A plot of the number of isolates tested positive by the WHO and NREVSS from 2003-2010 is shown in Figure~\ref{fig:flu_traces}(a).  From these data and the CDC weekly flu reports, we defined a set of six events (Events A-F) corresponding to the 2003-2004, 2004-2005, 2005-2006, 2006-2007, 2007-2008, and 2009-2010 flu seasons, respectively.  The 2003-2004 flu season began earlier than normal, and coincided with a flu vaccination shortage in many states.  For the vaccination that was available, the CDC found that it was ``not effective or had very low effectiveness'' (http://www.cdc.gov/media/pressrel/fs040115.htm).  The 2004-2005 and 2007-2008 flu seasons were more severe than the 2005-2006 and 2006-2007 seasons.  However, the 2005-2006 season coincided with an avian flu (H5N1) scare in which Dr. David Narbarro, Senior United Nations System Coordinator for Avian and Human Influenza, was famously quoted as predicting that an avian flu pandemic would lead to 5 million to 150 million deaths.  Finally, the 2009-2010 flu season coincides with the emergence of the 2009 H1N1 (``swine flu'') subtype\footnote{According to the CDC, ``Antigenic characterization of 2009 influenza A (H1N1) viruses indicates that these viruses are only distantly related antigenically and genetically to seasonal influenza A (H1N1) viruses''.  See http://www.cdc.gov/flu/weekly/weeklyarchives2009-2010/weekly20.htm.} in the United States.
\begin{figure}[t!]
	\centering
	\begin{tabular}{ccc}
		\hspace{-0.2in}
		\includegraphics[width = 1.75in]{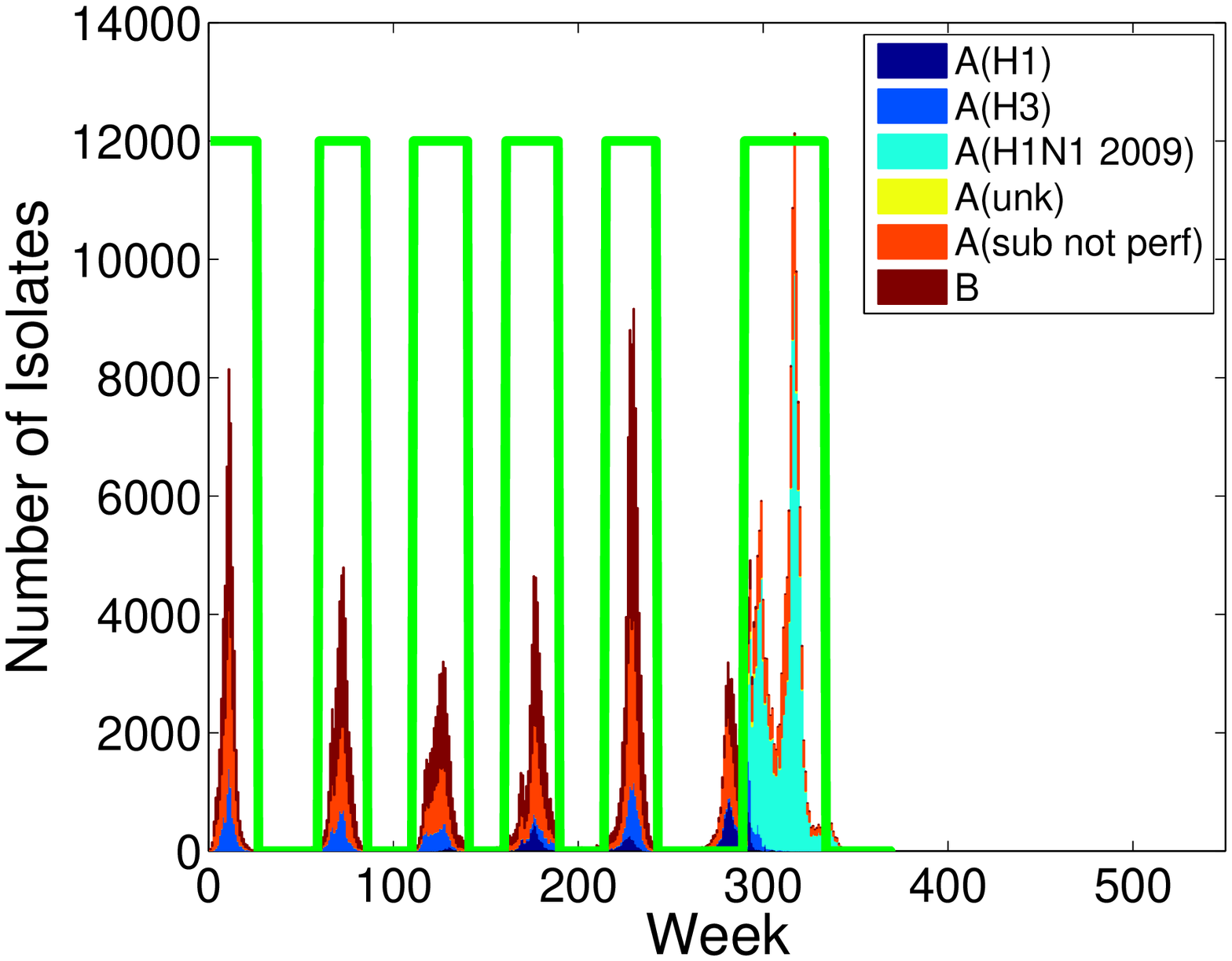} & \hspace{-0.2in}
		\includegraphics[width = 1.75in]{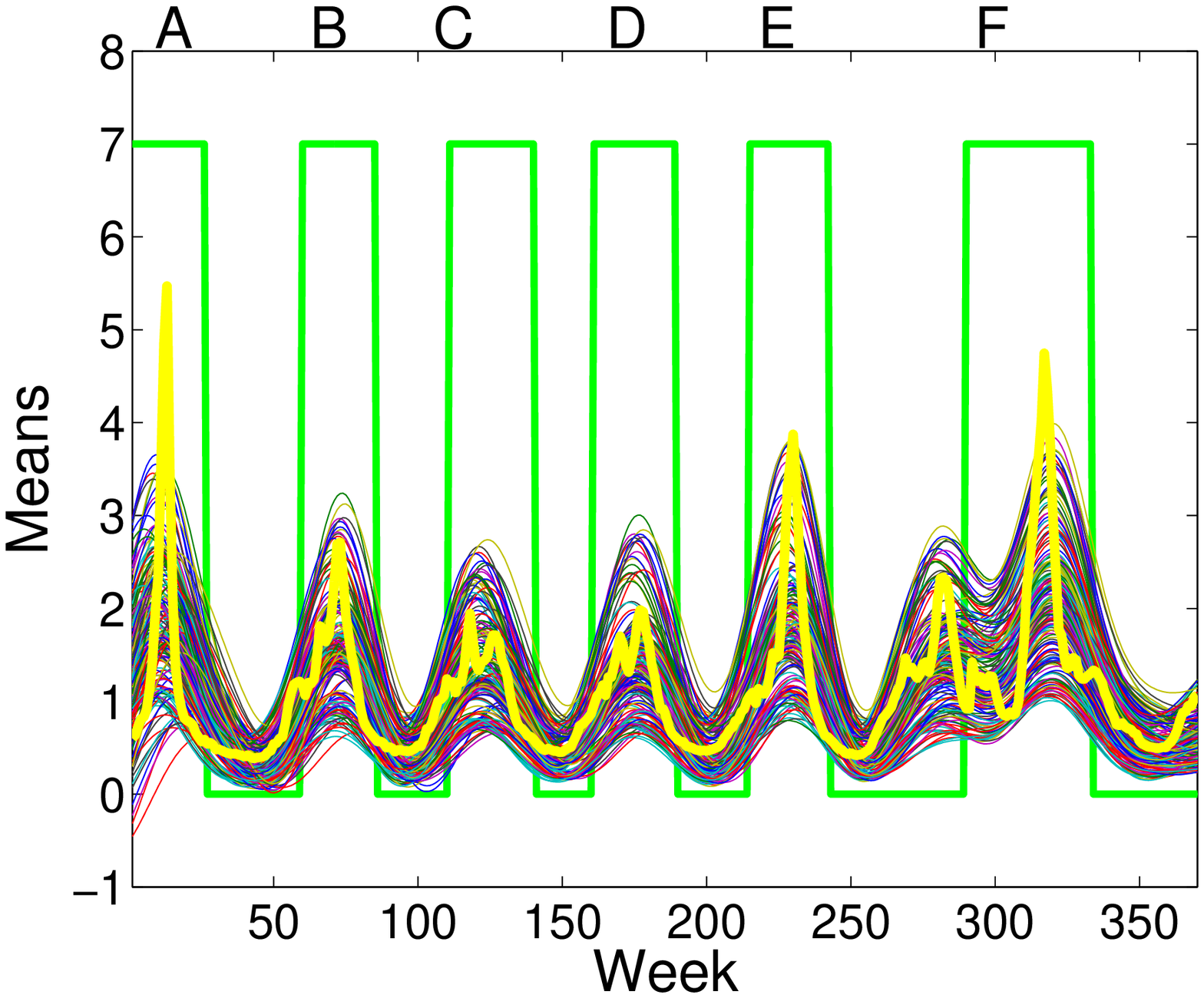} & \hspace{-0.2in}
		\includegraphics[width = 1.75in]{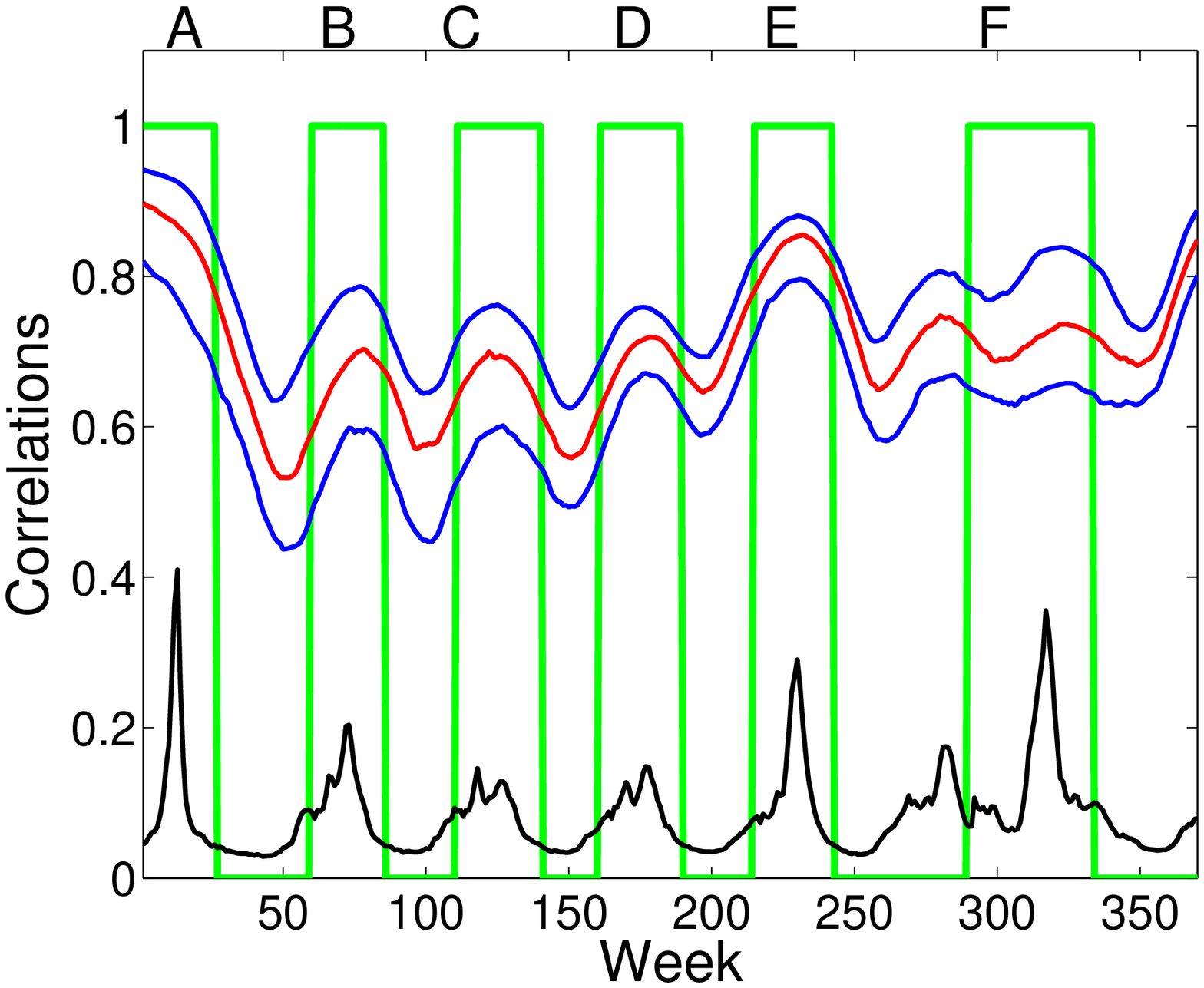} \\
		\hspace{-0.2in} (a) & \hspace{-0.2in}(b) & \hspace{-0.2in}(c)\\
	\end{tabular}
	\caption{(a) Plot of the number of isolates tested positive by the WHO and NREVSS over the period of September 29, 2003 to May 23, 2010. The isolates are divided into virus subtypes, specifically influenza A (H3N2, H1 = \{H1N2 and H1N1\}, 2009 H1N1) and influenza B. The isolates where subtyping was not performed or was not possible are also indicated. (b) Plot of posterior means of the nonparametric mean function $\mu_j(x)$ for each of the 183 states and regions in the Google Trends Flu dataset.  The thick yellow line indicates the Google Flu Trends estimate of the United States influenza rates.  (c) For New York, the 25th, 50th, and 75th quantiles of correlation with the 182 other states and regions based on the posterior mean $\hat{\Sigma}(x)$ of the covariance function.  The black line is a scaled version of the United States influenza rates, as in (b), shown for easy comparison.  The green line shown in plots (a)-(c) indicates the time periods determined to flu events.  Specifically, Event A corresponds to the 2003-2004 flu season (flu shot shortage), Event B the 2004-2005 season, Event C the 2005-2006 season (avian flu scare), Event D the 2006-2007 season, Event E the 2007-2008 season (severe), and Event F the 2009-2010 season (2009 H1N1 or ``swine flu'').}\label{fig:flu_traces} \postcap \vspace{0.1in}
\end{figure}
\subsection{Google Flu Trends Dataset}
To aid in a more rapid response to influenza activity, a team of researchers at Google devised a model based on Google user search queries that is predictive of CDC ILI rates~\citep{Ginsberg:08}---that is, the probability that a random physician visit is related to an influenza-like illness.  The \emph{Google Flu Trends} methodology was devised as follows.  From the hundreds of billions of individual searches from 2003-2008, time series of state-based weekly query rates were created for the 50 million most common search terms.  The predictive performance of a regression on the logit-transformed query rates was examined for each of the 50 million candidates and a ranked list was produced that rewarded terms predictive of rates exhibiting similar regional variations to that of the CDC data.  A massive variable selection procedure was then performed to find the optimal combination of query words (based on best fit against out-of-sample ILI data), resulting in a final set of 45 \emph{ILI-related queries}.  Using the 45 ILI-related queries as the explanatory variable, a region-independent univariate linear model was fit to the weekly CDC ILI rates from 2003-2007.  This model is used for making estimates in any region based on the ILI-related query rates from that region.  The results were validated against the CDC data both on training and test data, with the Google reported US and regional rates closely tracking the actual reported rates.  

A key advantage of the Google data (available at http://www.google.org/flutrends/) is that the ILI rate predictions are available 1 to 2 weeks before the CDC weekly reports are published.  Additionally, a user's IP address is typically connected with a specific geographic area and can thus provide information at a finer scale than the 10-regional and US aggregate reporting provided by the CDC.  Finally, the Google reports are not subject to revisions.  One important note is that the Google Flu Trends methodology aims to hone in on searches and rates of such searches that are indicative of influenza activity.  A methodology based directly on raw search queries might instead track general interest in influenza, waxing and waning quickly with various media events.  

We analyze the Google Flu Trends data from the week of September 28, 2003 through the week of October 24, 2010, providing 370 observation vectors $y_i$.  Each observation vector is 183-dimensional with elements consisting of Google estimated ILI rates at the US national level, the 50 states, 10 U.S. Department of Health \& Human Services surveillance regions, and 122 cities.  It is important to note, however, that there is substantial missing data with entire blocks of observations unavailable (as opposed to certain weeks sporadically being omitted).  At the beginning of the examined time frame only 114 of the 183 regions were reporting.  By the end of Year 1, there were 130 regions.  These numbers increased to 173, 178, 180, and 183 by the end of Years 2, 3, 4, and 5, respectively.  The high-dimensionality and missing data structure make the Google Flu Trends dataset challenging to analyze in its entirety with existing heteroscedastic models.  As part of an exploratory data analysis, in Figure~\ref{fig:flu_mapsEst} we plot sample estimates of the geographic correlation structure between the states during an event period for four representative states.  Specifically, we first subtract a moving average estimate of the mean (window size 10) and then aggregate the data over Events B-F, omitting Event A due to the quantity of missing data.  Because of the dimensionality of the data (183 dimensions) and the fact that there are only 157 event observations, we simply consider the state-level observations (plus District of Columbia), reducing the dimensionality to 51.  The limited data also impedes our ability to perform time-specific sample estimates of geographic correlations.
\begin{figure}[t!]
	\centering
		\hspace{-0.2in}
	 \begin{tabular}{cc}
		\includegraphics[width = 1.75in]{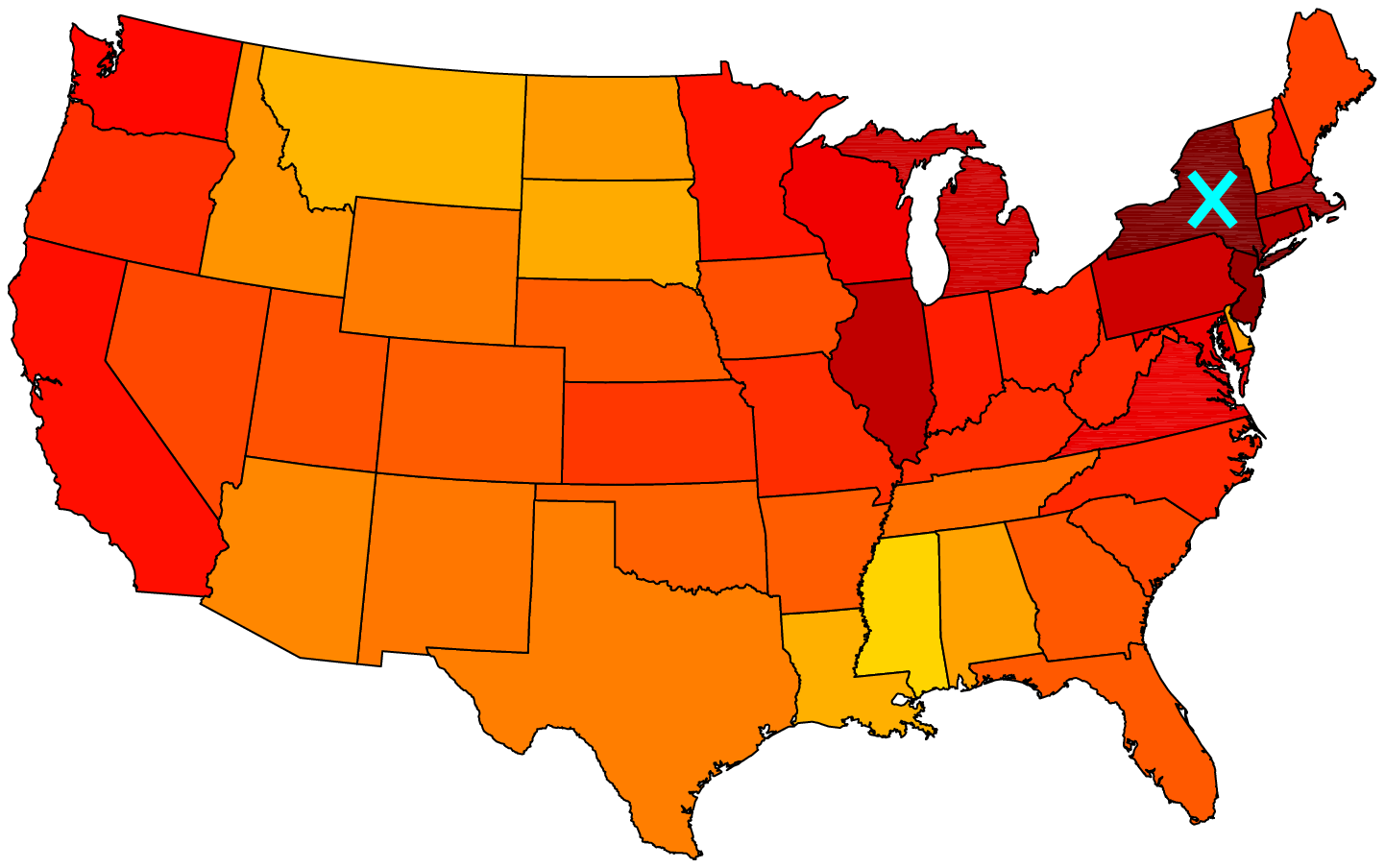} & \hspace{-0.2in}
		\includegraphics[width = 1.75in]{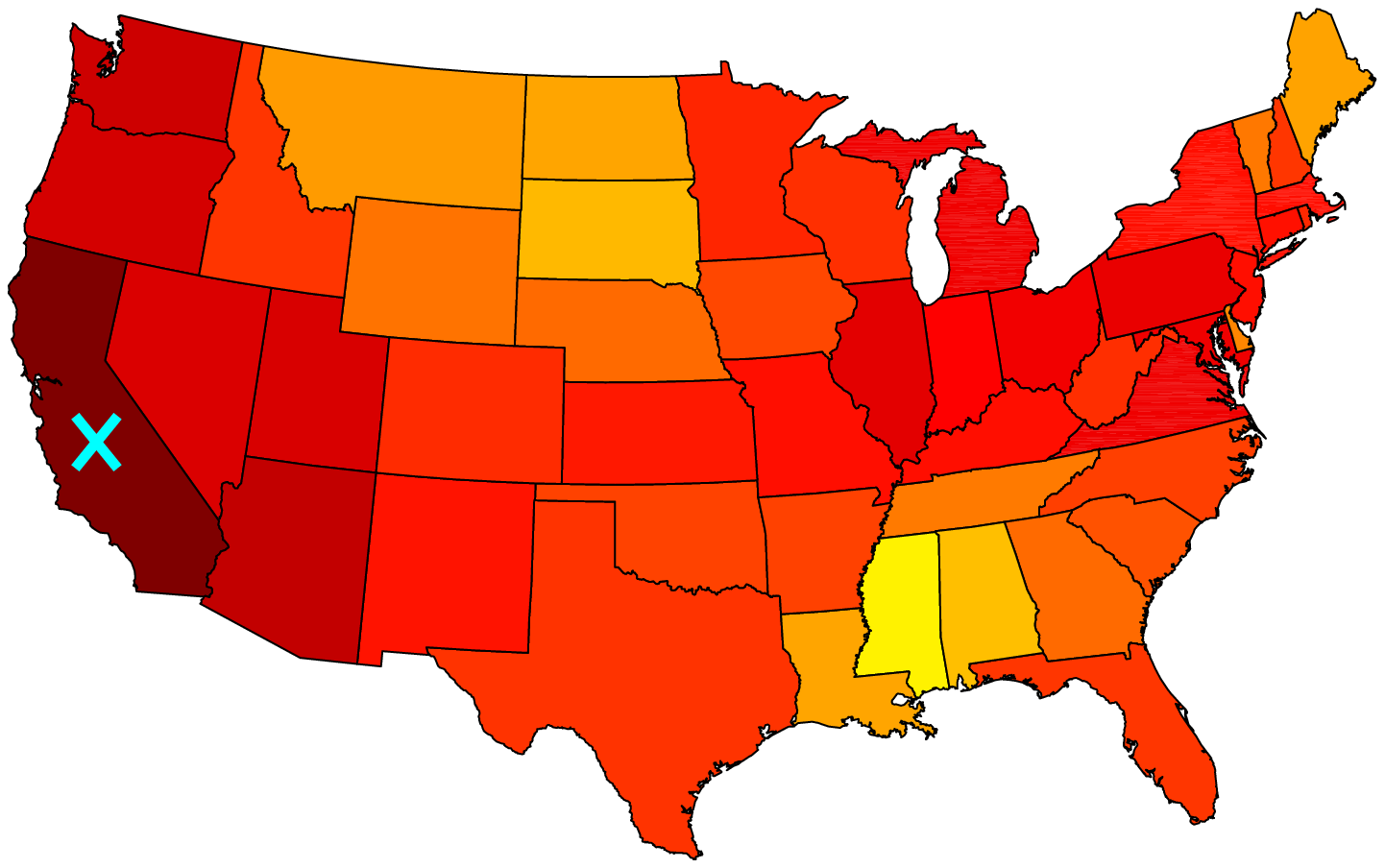}\\
		\includegraphics[width = 1.75in]{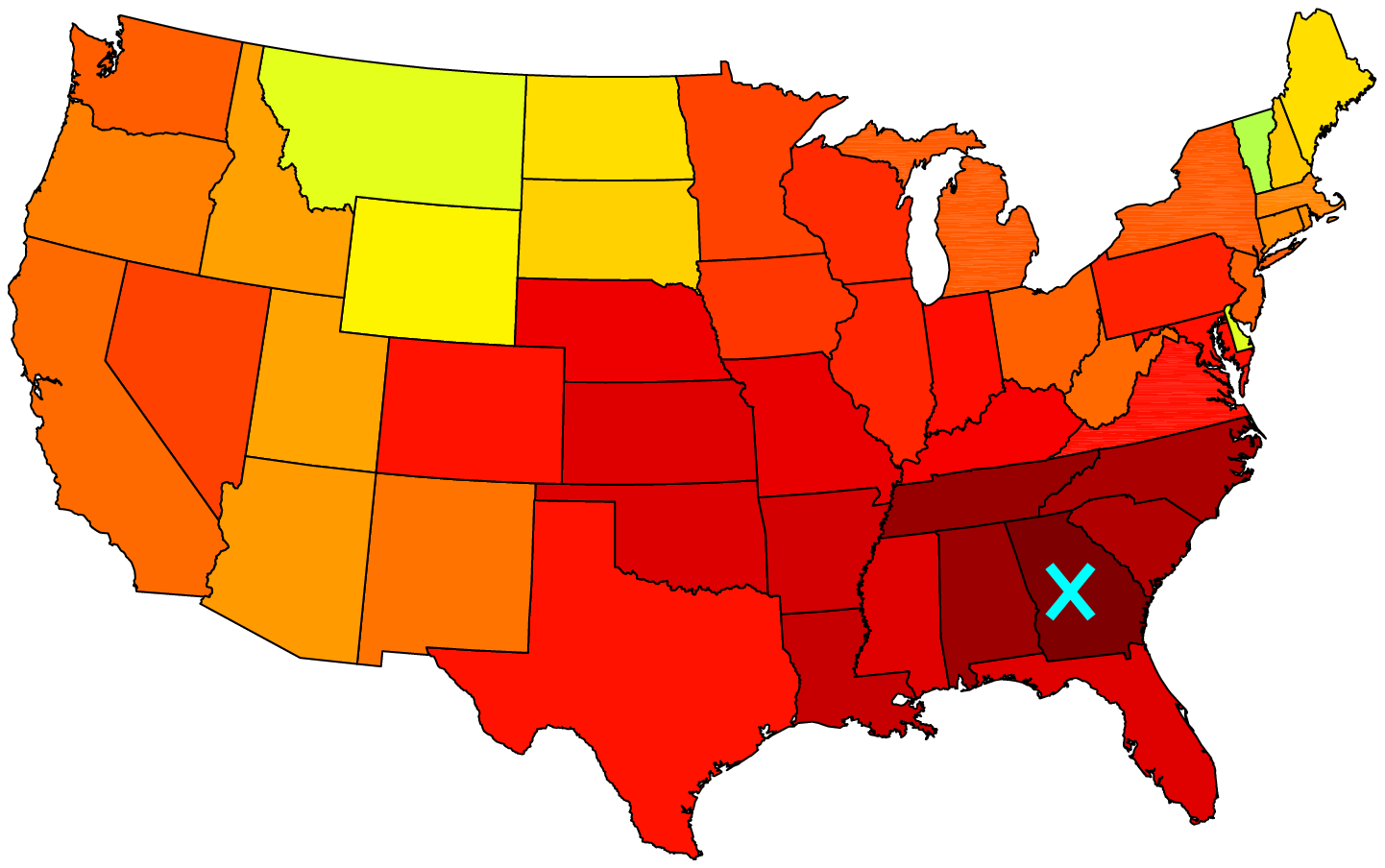}  & \hspace{-0.2in}
	    \includegraphics[width = 1.75in]{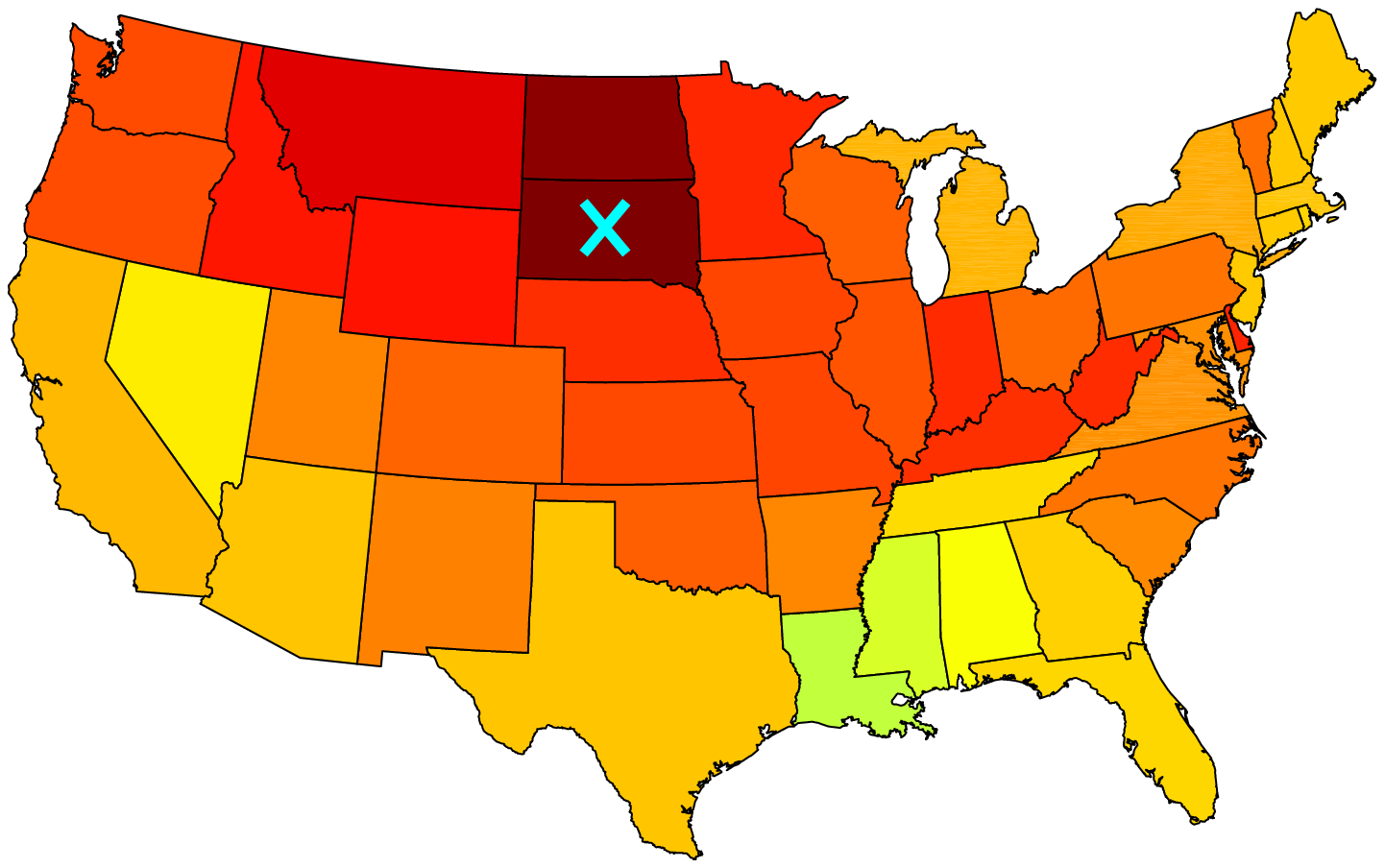}
	\end{tabular}
	\caption{For each of four geographically distinct states (New York, California, Georgia, and South Dakota), plots of correlations between the state and all other states based on the sample covariance estimate from aggregating the state-level data during the event periods B-F after subtracting a moving average estimate of the mean.  Event A was omitted due to an insufficient number of states reporting.  Note that South Dakota is missing 58 of the 157 Event B-F observations.}\label{fig:flu_mapsEst} \postcap \vspace{0.1in}
\end{figure}

\subsection{Heteroscedastic Modeling of Google Flu Trends}
Our proposed heteroscedastic model allows one to capture both spatial and temporal changes in correlation structure, providing an important additional tool in predicting influenza rates.  We specifically consider $y_i \sim \mathcal{N}(\mu(x_i),\Sigma(x_i))$ with the nonparametric function $\mu(x_i) = \Theta \xi(x_i) \psi(x_i)$ defining the mean of the ILI rates in each of the 183 regions.  For a given week $x_i$, the spatial correlation structure is captured by the covariance $\Sigma(x_i) = \Theta \xi(x_i)\xi(x_i)'\Theta' + \Sigma_0$.  Temporal changes are implicitly modeled through the proposed covariance regression framework that allows for continuous variations in $\Sigma(x_i)$.  \cite{Dukic:09} also examine portions of the Google Flu Trends data, but with the goal of on-line tracking of influenza rates on either a national, state, or regional level.  Specifically, they employ a state-space model with particle learning.  Our goal differs considerably.  We aim to jointly analyze the full 183-dimensional data, as opposed to univariate modeling.  Through such joint modeling, we can uncover important spatial dependencies lost when analyzing components of the data individually.  Such spatial information can be key in predicting influenza rates based on partial observations from select regions or in retrospectively imputing missing data.

There are a few crucial points to note.  The first is that no geographic information is provided to our model.  Instead, the spatial structure is uncovered simply from analyzing the raw 183-dimensional time series and patterns therein.  Second, because of the substantial quantity of missing data, imputing the missing values as in Section~\ref{sec:pred} is less ideal than simply updating our posterior based solely on the data that is available.  The latter is how we chose to analyze the Google Flu Trends data---our ability to do so without introducing any approximations is a key advantage of our proposed methodology.  
\begin{figure}[t!]
	\centering
	\begin{tabular}{cc}
		\hspace{-0.25in}\begin{tabular}{c}
			\rotatebox{90}{\textbf{New York}} \vspace{0.4in}\\
			\rotatebox{90}{\textbf{California}}\vspace{0.5in}\\
			\rotatebox{90}{\textbf{Georgia}} \vspace{0.4in}\\
			\rotatebox{90}{\textbf{South Dakota}} 
		\end{tabular}
	\begin{tabular}{ccc}
		\hspace{-0.2in}
		\includegraphics[width = 1.75in]{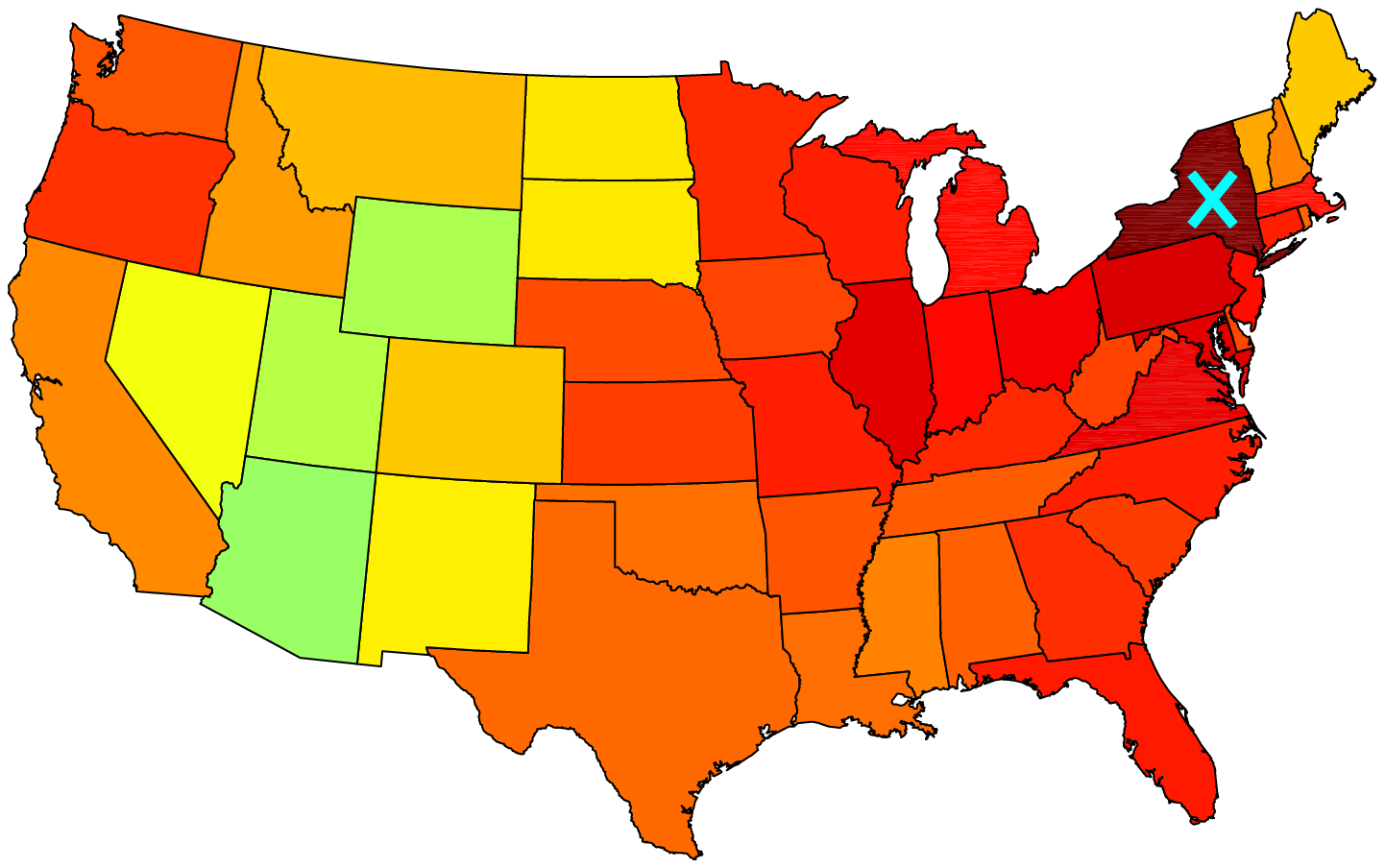} & \hspace{-0.2in}
		\includegraphics[width = 1.75in]{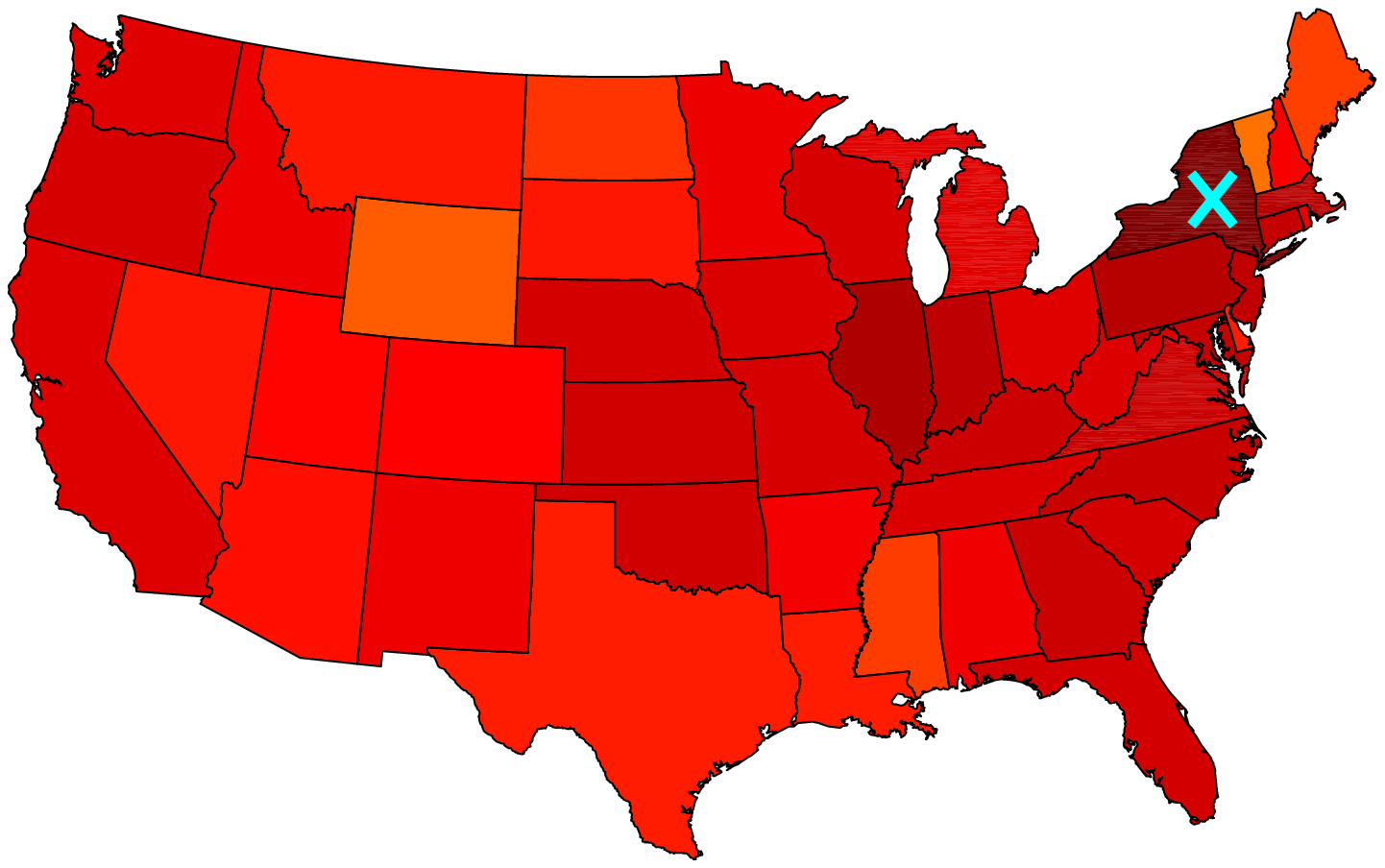} & \hspace{-0.2in}
		\includegraphics[width = 1.75in]{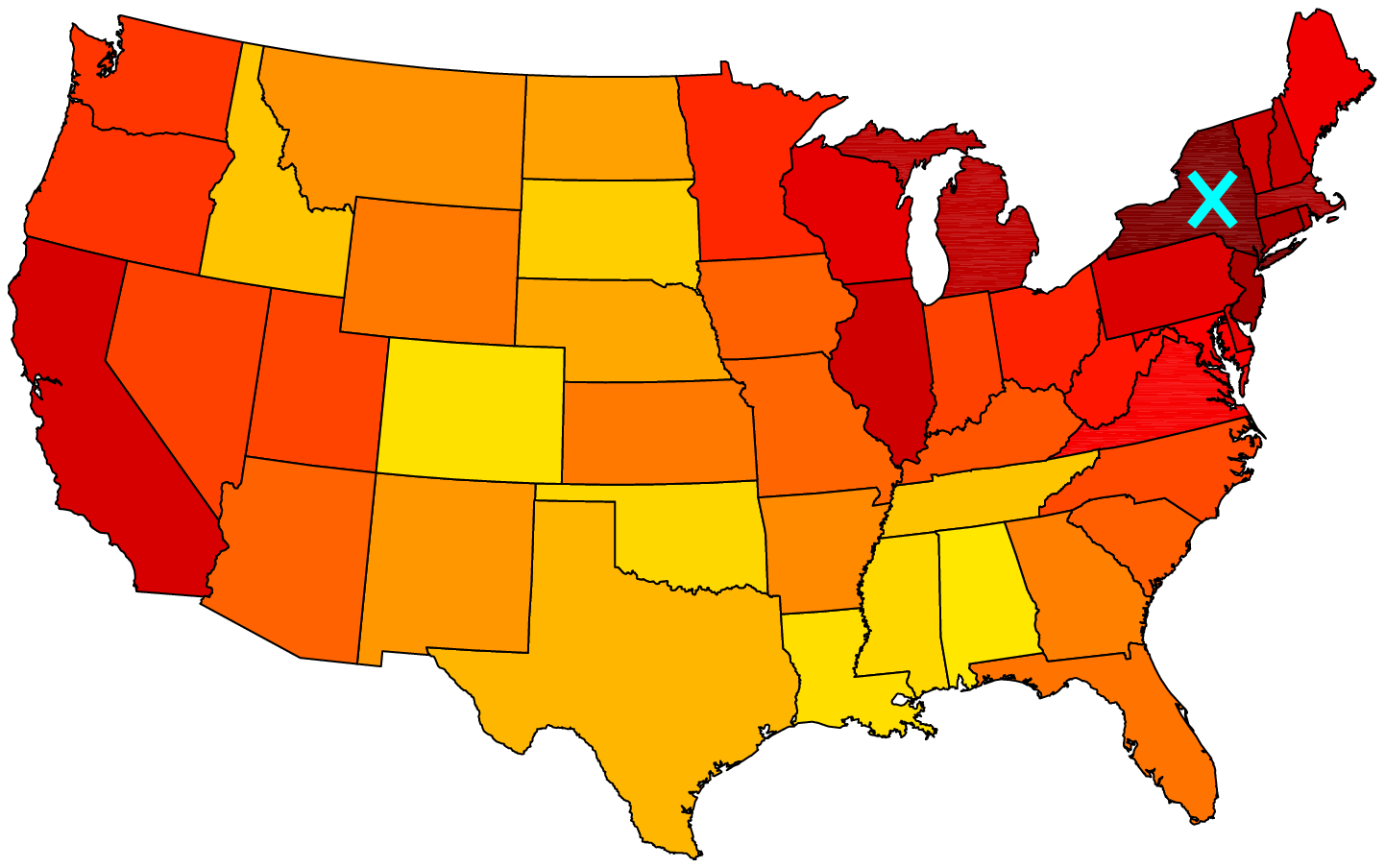} \\
		\includegraphics[width = 1.75in]{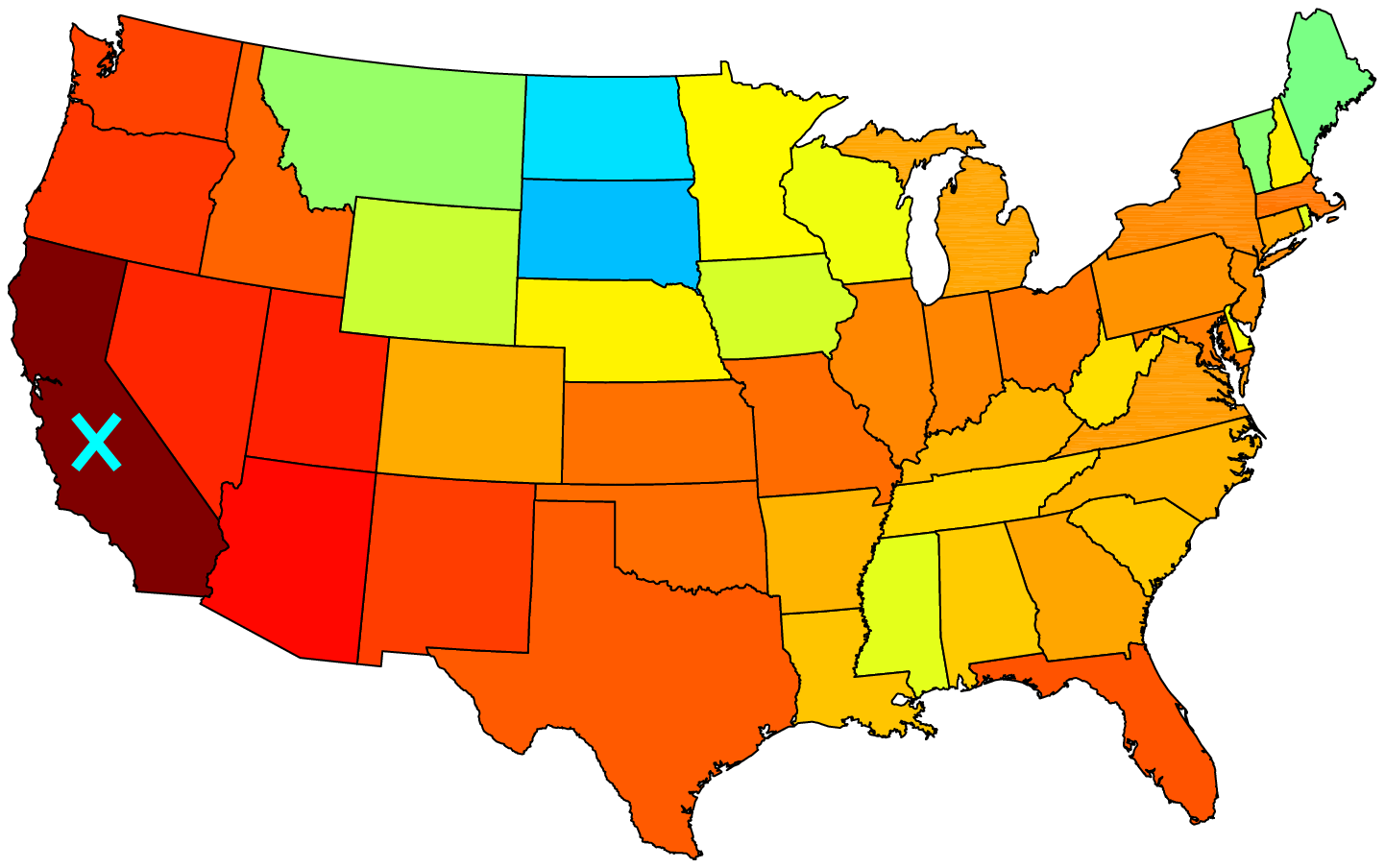} & \hspace{-0.2in}
		\includegraphics[width = 1.75in]{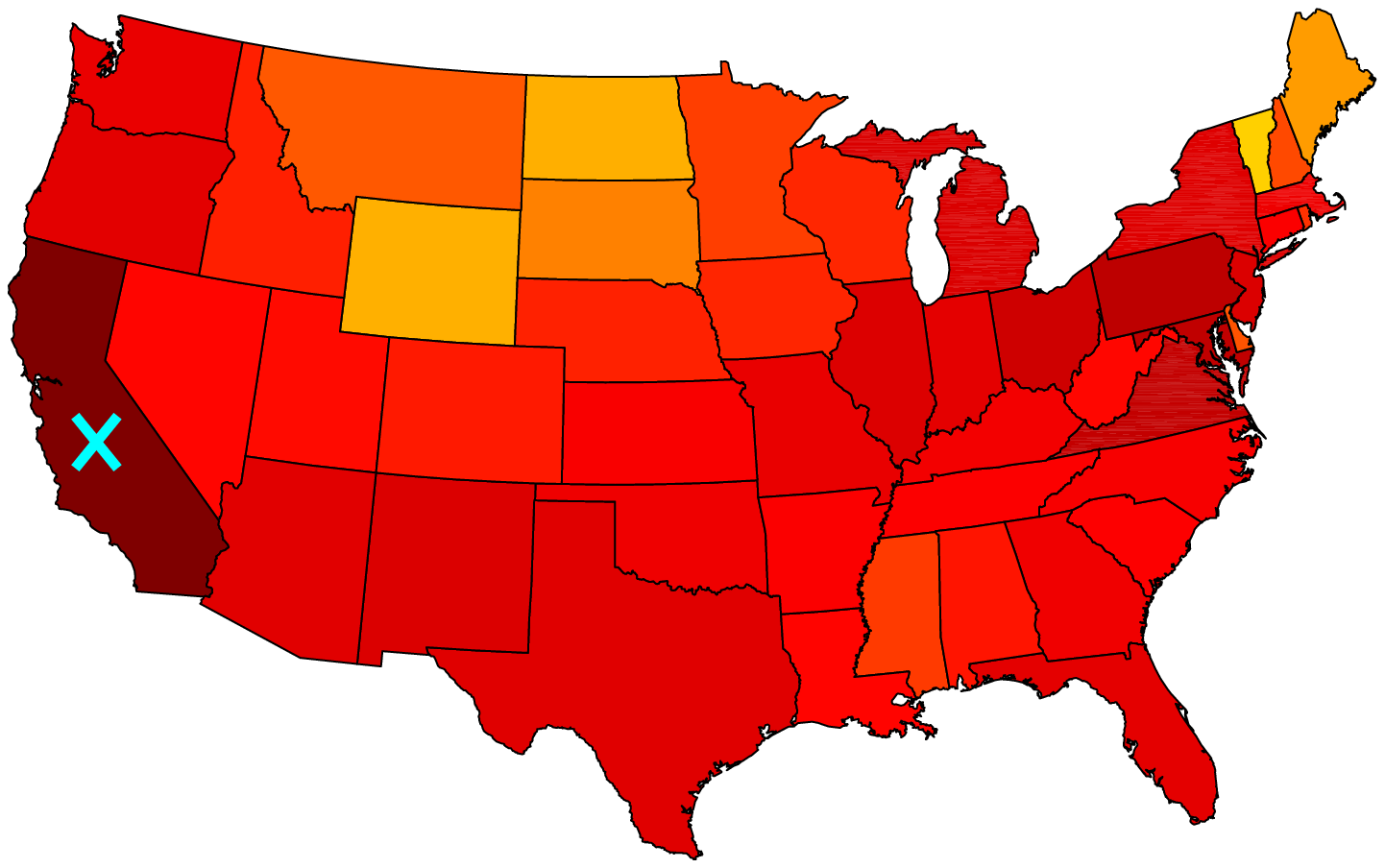} & \hspace{-0.2in}
		\includegraphics[width = 1.75in]{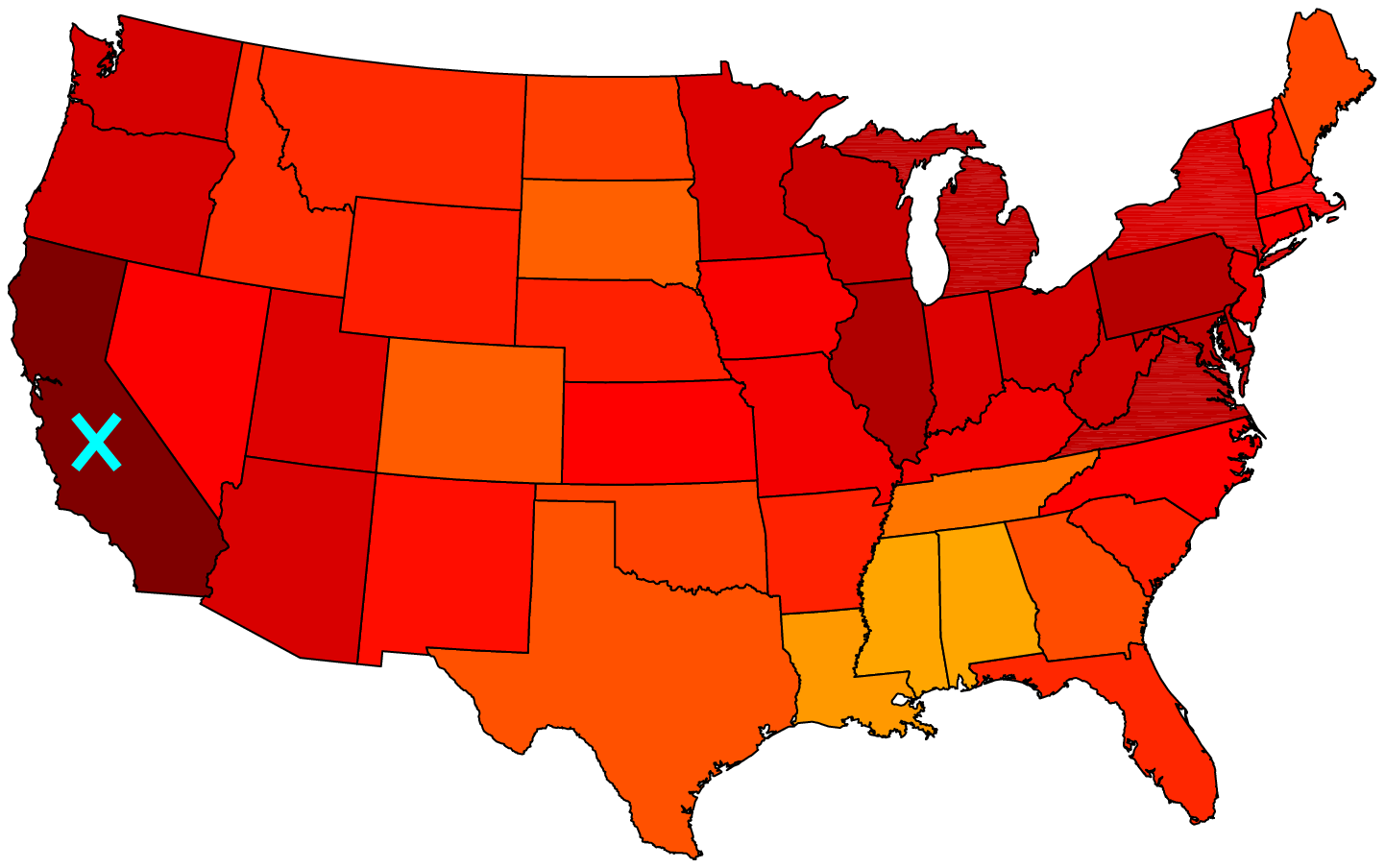} \\
		\includegraphics[width = 1.75in]{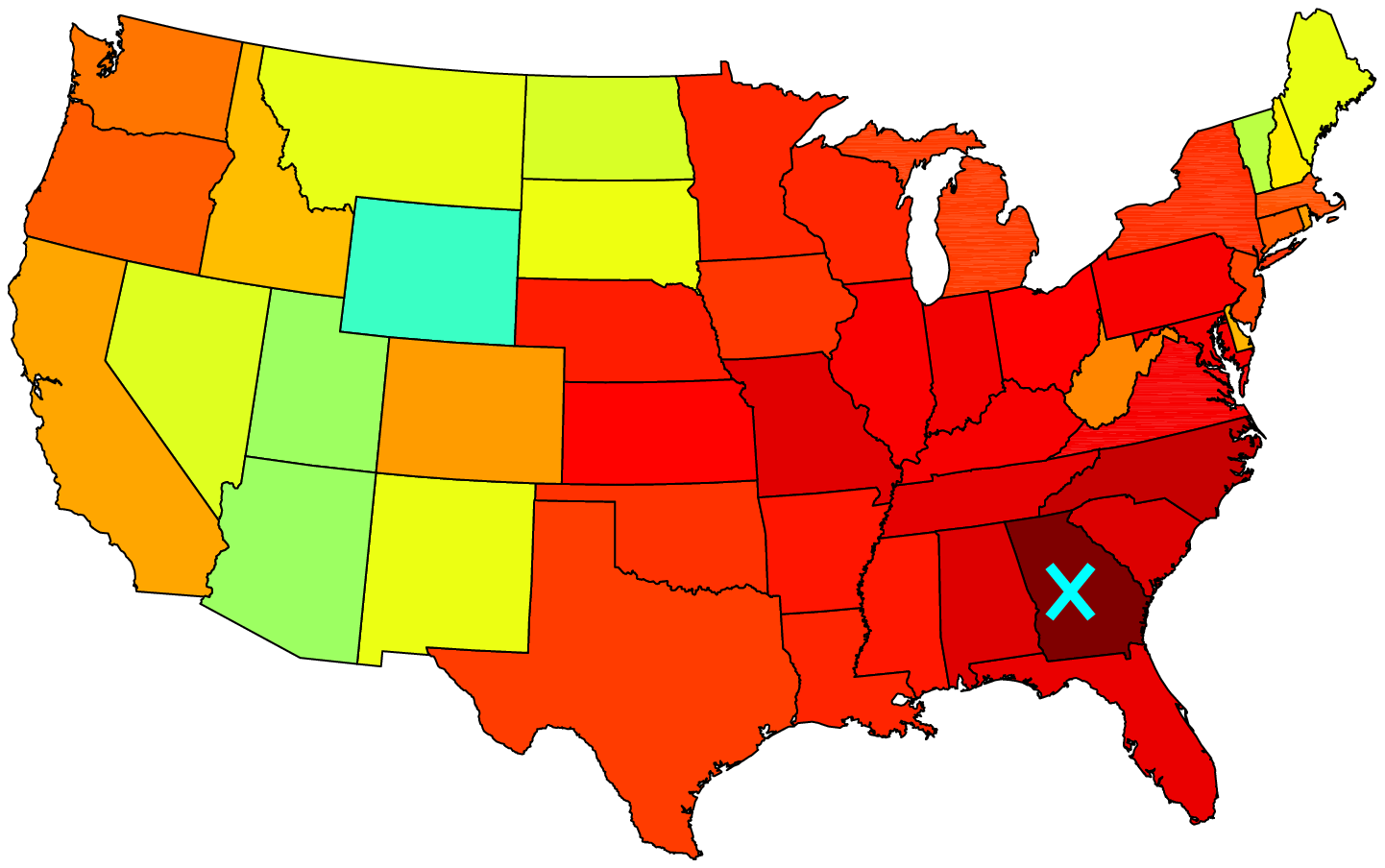} & \hspace{-0.2in}
		\includegraphics[width = 1.75in]{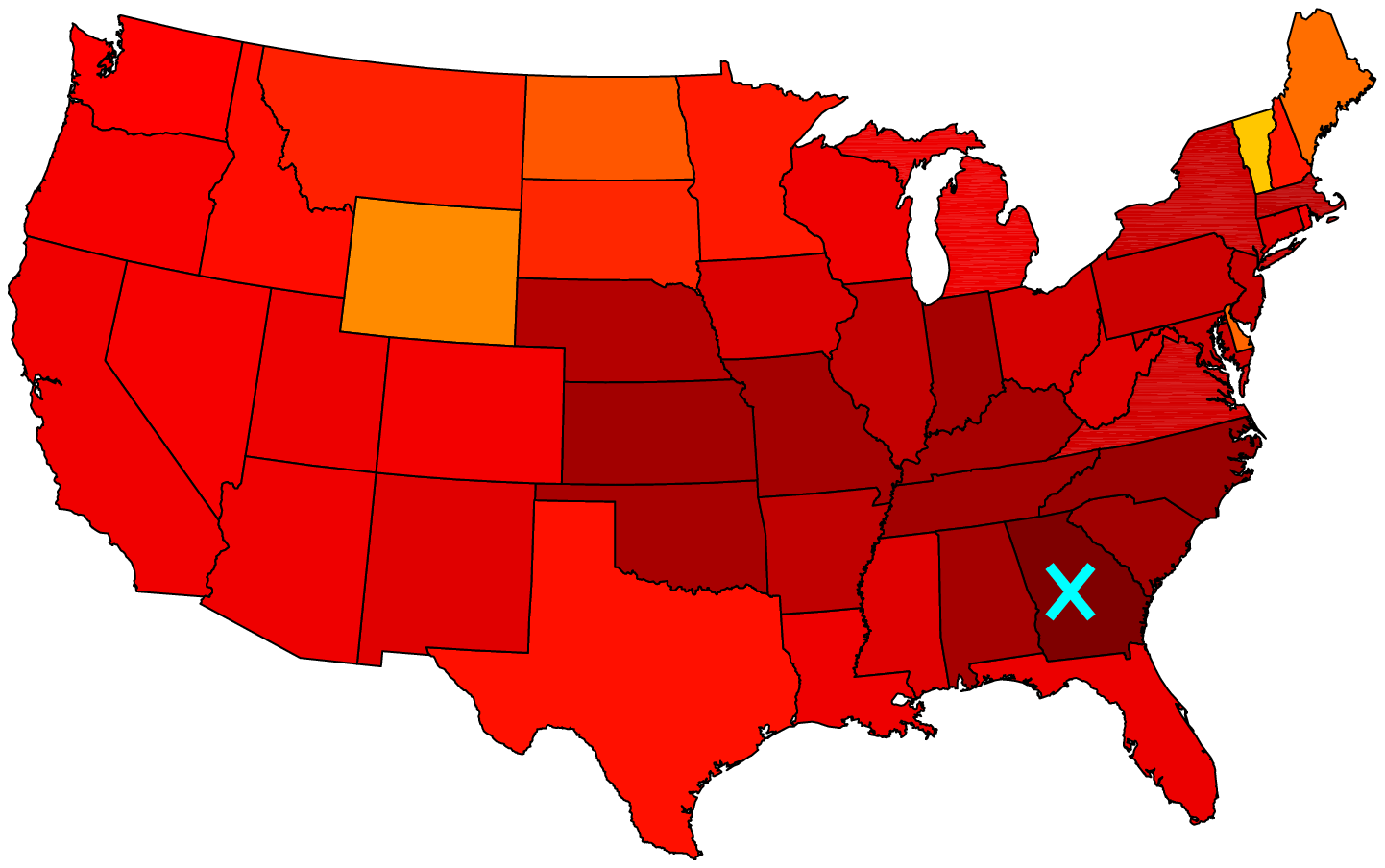} & \hspace{-0.2in}
		\includegraphics[width = 1.75in]{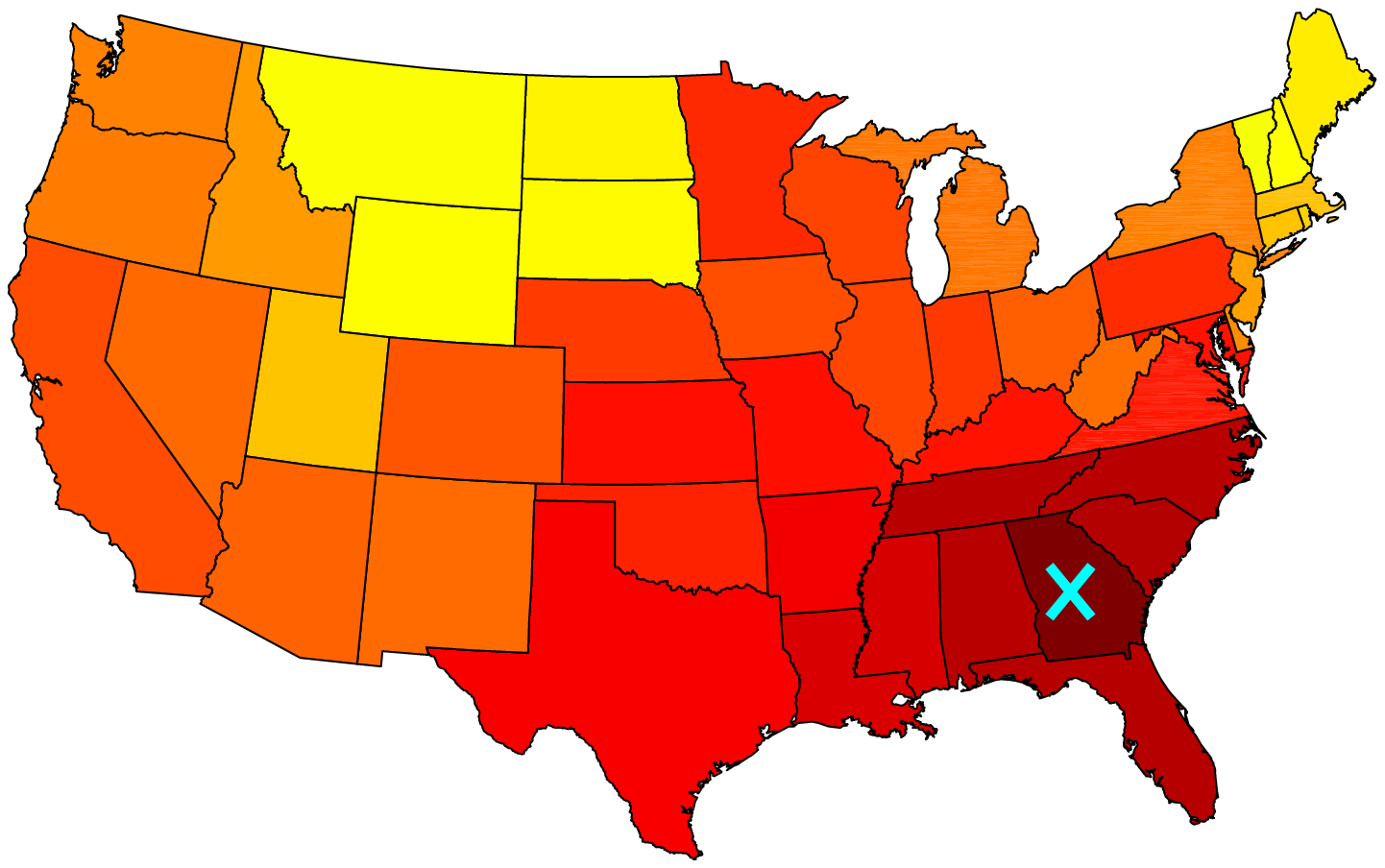} \\
		\includegraphics[width = 1.75in]{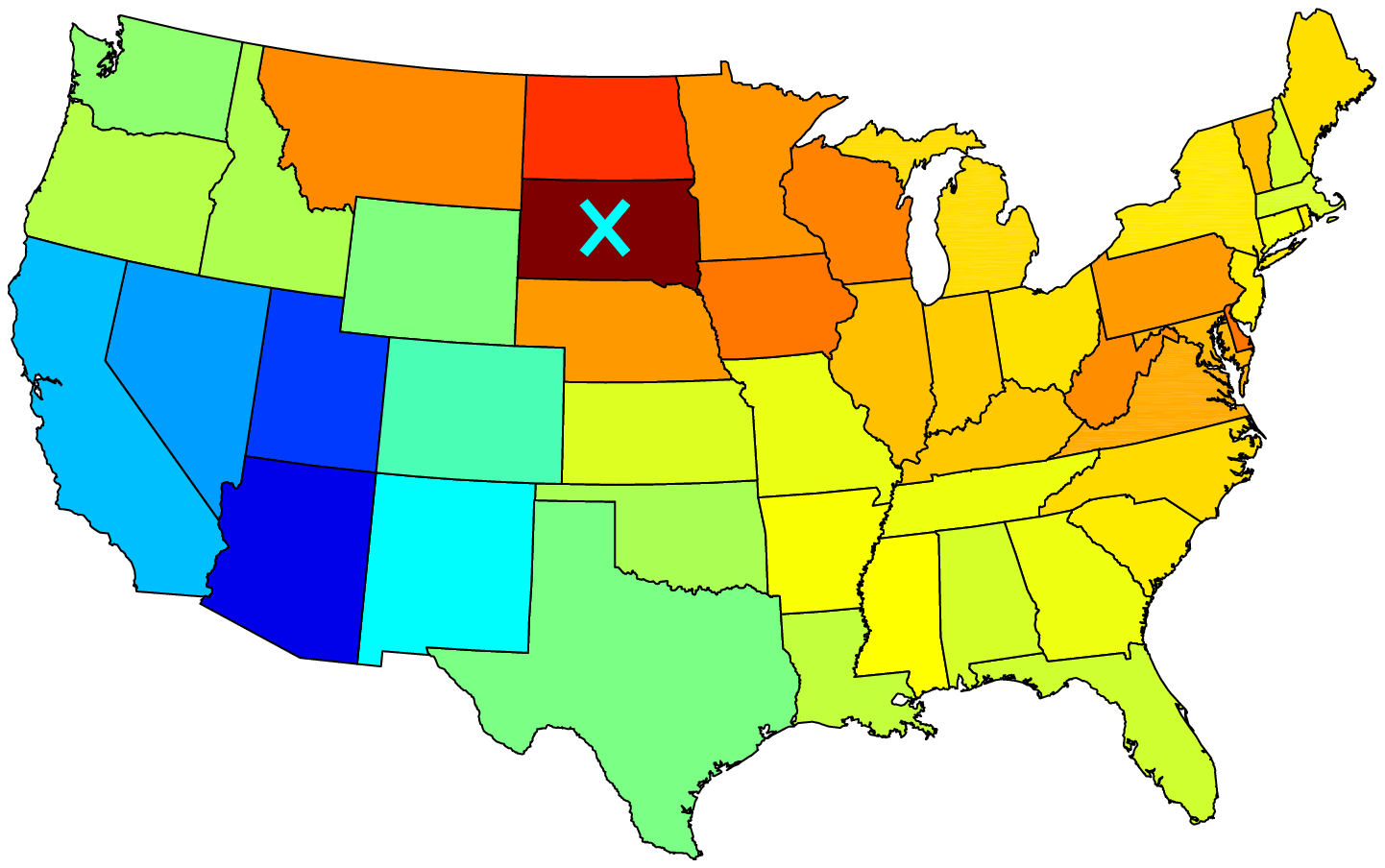} & \hspace{-0.2in}
		\includegraphics[width = 1.75in]{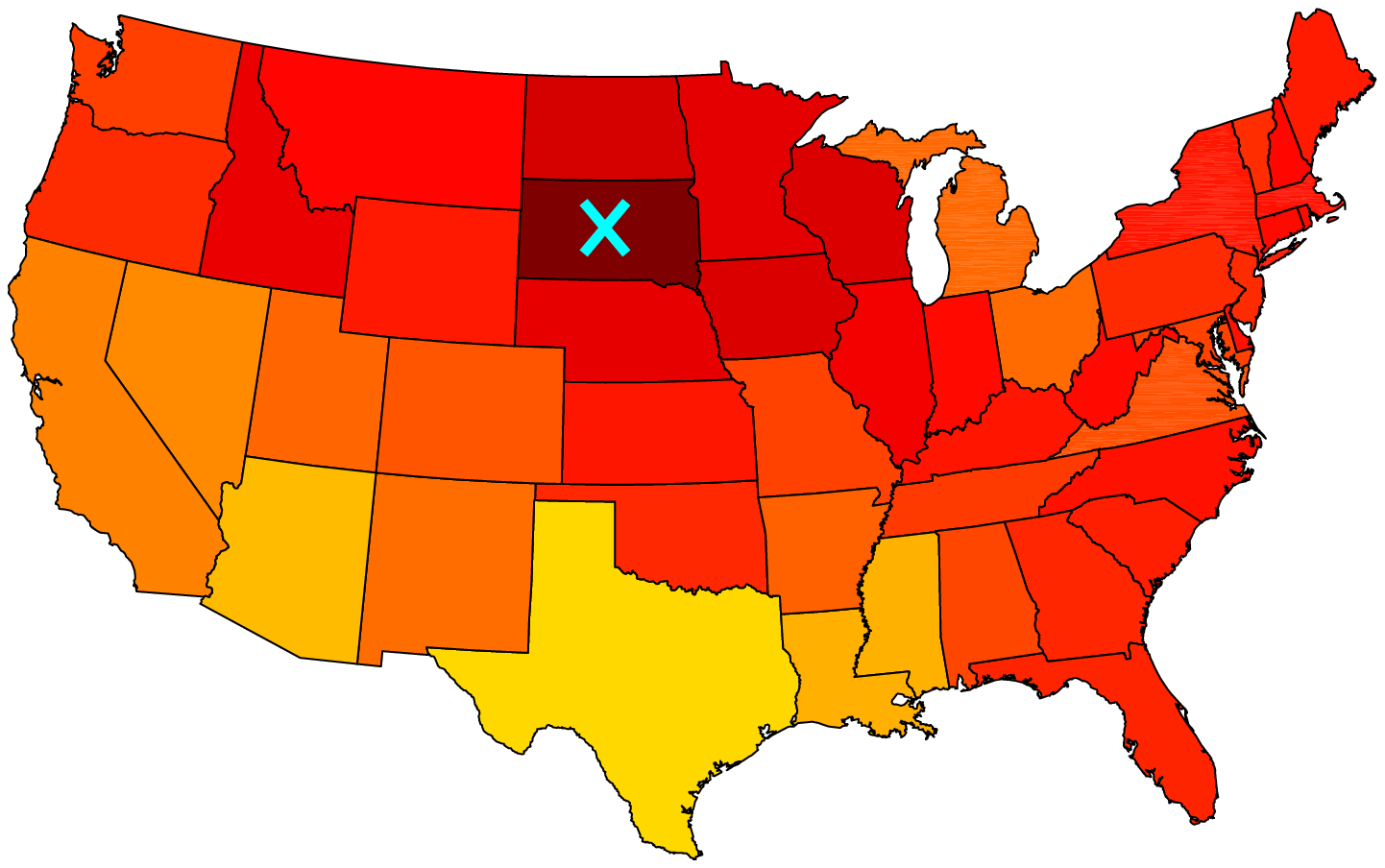} & \hspace{-0.2in}
		\includegraphics[width = 1.75in]{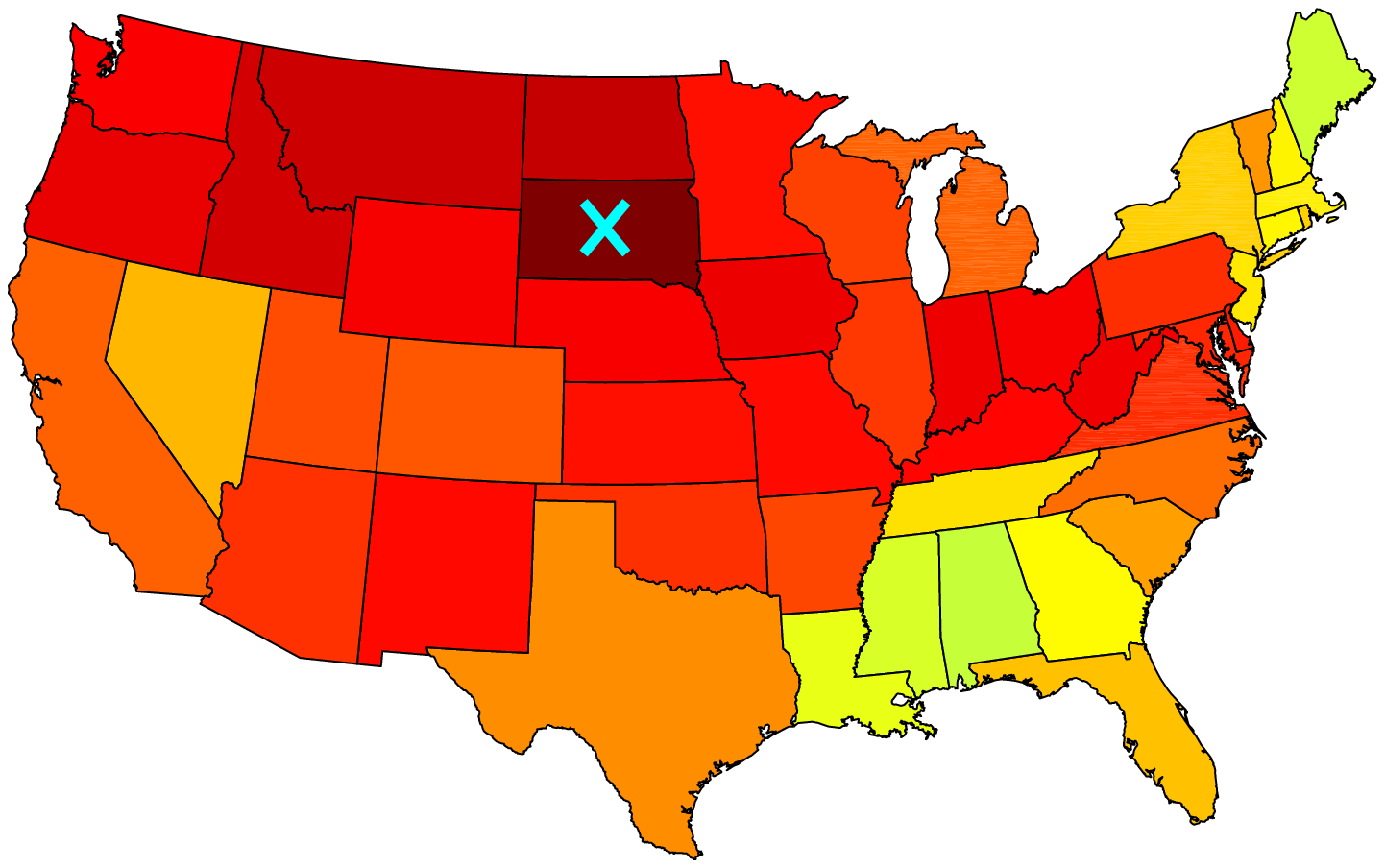} \\
		\hspace{-0.2in} \textbf{February 2006} & \hspace{-0.2in}\textbf{February 2008} & \hspace{-0.2in}\textbf{November 2009}\\
	\end{tabular}
\end{tabular}
	\caption{For each of four geographically distinct states (New York, California, Georgia, and South Dakota) and each of three key dates (February 2006 of Event C, February 2008 of Event E, and November 2009 of Event F), plots of correlations between the state and all other states based on the posterior mean $\hat{\Sigma}(x)$ of the covariance function.  The plots clearly indicate spatial structure captured by $\Sigma(x)$, and that these spatial dependencies change over time.  Note that no geographic information was included in our model.  Compare to the maps of Figure~\ref{fig:flu_mapsEst}.}\label{fig:flu_maps} \postcap \vspace{0.1in}
\end{figure}

The results presented in Figures~\ref{fig:flu_traces} and~\ref{fig:flu_maps} clearly demonstrate that we are able to capture both spatial and temporal changes in correlations in the Google Flu Trends data, even in the presence of substantial missing information.  We preprocessed the data by scaling the entire dataset by one over the largest variance of any of the 183 time series.  The hyperparameters were set as in the simulation study of Section~\ref{sec:sim}, except with larger truncation levels $L^*=10$ and $k^*=20$ and with the Gaussian process length-scale hyperparameter set to $\kappa=100$ to account for the time scale (weeks) and the rate at which ILI incidences change.  Once again, by examining posterior samples of $\Theta$, we found that the chosen truncation level was sufficiently large.  We ran 10 chains each for 10,000 MCMC iterations, discarded the first 5,000 for burn-in, and then thinned the chains by examining every 10 samples.  Each chain was initialized with parameters sampled from the prior.  To assess convergence, we performed the modified Gelman-Rubin diagnostic of~\cite{GelmanRubin} on the MCMC samples of the variance terms $\Sigma_{jj}(x_i)$ for $j$ corresponding to the state indices of New York, California, Georgia, and South Dakota, and $x_i$ corresponding to the midpoint of each of the 12 event and non-event time windows\footnote{Our tests used the code provided at http://www.lce.hut.fi/research/mm/mcmcdiag/.}. These elements are spatially and geographically disparate, with South Dakota corresponding to an element with substantial missing data.  Of the 48 resulting variables examined (4 states and 12 time points), 40 had potential scale reduction factors $R^{1/2}<1.2$ and most $R^{1/2}<1.1$.  The variables with larger $R^{1/2}$ (all less than 1.4) corresponded to two non-event time periods.  We also performed hyperparameter sensitivity, doubling the length-scale parameter to $\kappa=200$ (implying less temporal correlation) and using a larger truncation level of $L^*=k^*=20$ with less stringent shrinkage hyperparameters $a_1=a_2=2$ (instead of $a_1=a_2=10$).  The results were very similar to those presented in this section, with all conclusions remaining the same.

In Figure~\ref{fig:flu_traces}(b), we plot the posterior mean of the 183 components of $\mu(x)$, showing trends that follow the Google estimated US national ILI rate.  For New York, in Figure~\ref{fig:flu_traces}(c) we plot the 25th, 50th, and 75th quantiles of correlation with the 182 other states and regions based on the posterior mean $\hat{\Sigma}(x)$ of the covariance function.  From this plot, we immediately notice that regions become more correlated during flu seasons, as we would expect.  The specific geographic structure of these correlations is displayed in Figure~\ref{fig:flu_maps}, where we see key changes with the specific flu event.  In the more mild 2005-2006 season, we see much more local correlation structure than the more severe 2007-2008 season (which still maintains stronger regional than distant correlations.)  The November 2009 H1N1 event displays overall regional correlation structure and values similar to the 2007-2008 season, but with key geographic areas that are less correlated.  Note that some geographically distant states, such as New York and California, are often highly correlated as we might expect based on their demographic similarities and high rates of travel between them.  Interestingly, the strong local spatial correlation structure for South Dakota in February 2006 has been inferred before any data are available for that state.  Actually, no data are available for South Dakota from September 2003 to November 2006.  Despite this missing data, the inferred correlation structures over these years are fairly consistent with those of neighboring states and change in manners similar to the flu-to-non-flu changes inferred after data for South Dakota are available.  

Comparing the maps of Figure~\ref{fig:flu_maps} to those of the sample-based estimates in Figure~\ref{fig:flu_mapsEst}, we see much of the same correlation structure, which at a high level validates our findings.  Since the sample-based estimates aggregate data over Events B-F (containing those displayed in Figure~\ref{fig:flu_maps}), they tend to represent a time-average of the event-specific correlation structure we uncovered.  Note that due to the dimensionality of the dataset, the sample-based estimates are based solely on state-level measurements and thus are unable to harness the richness (and crucial redundancy) provided by the other regional reporting agencies.  Furthermore, since there are a limited number of per-event observations, the naive approach of forming sample covariances based on local bins of data is infeasible.  The high-dimensionality and missing data structure of this dataset also limit our ability to compare to alternative methods such as those cited in Section~\ref{sec:intro}---none yield results directly comparable to the full analysis we have provided here.  Instead, they are either limited to examination of the small subset of data for which all observations are present and/or a lower-dimensional selection (or projection) of observations.  On the other hand, our proposed algorithm can readily utilize all information available to model the heteroscedasticity present here.  (Compare to the common GARCH models, which cannot handle missing data and are limited to typically no more than 5 dimensions.)  In terms of computational complexity, each of our chains of 10,000 Gibbs iterations based on a naive implementation in MATLAB (R2010b) took approximately 12 hours on a machine with four Intel Xeon X5550 Quad-Core 2.67GHz processors and 48 GB of RAM.  

\section{Discussion}
\label{sec:discussion}

In this paper, we have presented a Bayesian nonparametric approach to covariance regression which allows an unknown $p \times p$ dimensional covariance matrix $\Sigma(x)$ to vary flexibly over $x \in \mathcal{X}$, where $\mathcal{X}$ is some arbitrary (potentially multivariate) predictor space.  Foundational to our formulation is quadratic mixing over a collection of dictionary elements, assumed herein to be Gaussian process random functions, defined over $\mathcal{X}$.  By inducing a prior on $\Sigma_{\mathcal{X}} = \{\Sigma(x), x \in \mathcal{X}\}$ through a shrinkage prior on a predictor-dependent latent factor model, we are able to scale to the large $p$ domain.  Our proposed methodology also yields computationally tractable algorithms for posterior inference via fully conjugate Gibbs updates---this is crucial in our being able to analyze high-dimensional multivariate datasets.  We demonstrated the utility of our Bayesian covariance regression framework through both simulated studies and analysis of the Google Trends flu dataset, the latter having nearly 200 dimensions.

There are many possible interesting and relatively direct extensions of the proposed covariance regression framework.  The most immediate are those that fall into the categories of (i) addressing the limitations of our current assumption of Gaussian marginals, and (ii) scaling to datasets with large numbers of observations.  

In longitudinal studies or spatio-temporal datasets, one is faced with repeated collections of observations over the predictor space.  These collections are clearly not independent.  To cope with such data, one could envision embedding the current framework within a hierarchical model (e.g., to model random effects on a mean), or otherwise use the proposed methodology as a building block in more complicated models.  Additionally, it would be trivial to extend our framework to accommodate multivariate categorical responses, or joint categorical and continuous responses, by employing the latent variable probit model of~\cite{AlbertChib:93}.  This would lead, for example, to a more flexible class of multivariate probit models in which the correlation between variables can change with time and other factors.  For computation, the use of parameter expansion allows us to simply modify our current MCMC algorithm to include a data augmentation step for imputing the underlying continuous variables.  Imposing the constraints on the covariance could be deferred to a post-processing stage.  Another interesting direction for future research is to consider embedding our covariance regression model in a Gaussian copula model.  One possible means of accomplishing this is through a variant of the approach proposed by~\cite{Hoff:07}, which avoids having to completely specify the marginal distributions.  

As discussed in Section~\ref{sec:compIssues}, our sampler relies on $L* \times k*$ draws from an $n$-dimensional Gaussian (i.e., posterior draws of our Gaussian process random dictionary functions).  For very large $n$, this becomes infeasible in practice since computations are, in general, $O(n^3)$.  Standard tools for scaling up Gaussian process computation to large datasets, such as covariance tapering~\citep{Kaufman:08,Du:09} and the predictive process~\citep{Banerjee:08}, can be applied directly in our context.  Additionally, one might consider using the integrated nested Laplace approximations of~\cite{Rue:09} for computations.  The size of the dataset (both in terms of $n$ and $p$) also dramatically affects our ability to sample the Gaussian process length-scale hyperparameter $\kappa$ since our proposed method relies on samples from an $np$-dimensional Gaussian.  See Section~\ref{sec:hyp} for details and possible methods of addressing this issue.  If it is feasible to perform inference over the length-scale parameter, one can consider implicitly including a test for homoscedasticity by considering $\kappa$ taking values in the extended reals (i.e., $\Re \bigcup \{\infty\}$) and thus allowing our formulation to collapse on the simpler model if $\kappa=\infty$.

Finally, we note that there are scenarios in which the functional data itself is covariance-valued, such as in diffusion tensor imaging.  In this case, each voxel in an image consists of a $3 \times 3$ covariance matrix that has potential spatio-temporal dependencies.  Specifically, take $\Sigma_{ij}(t)$ to represent the observed covariance matrix for subject $i$ at pixel $j$ and time $t$.  Here, one could imagine replacing the Gaussian process dictionary elements with splines and embedding this model within a hierarchical structure to allow variability among subjects while borrowing information.

As we readily see, the presented Bayesian nonparametric covariance regression framework easily lends itself to many interesting directions for future research with the potential for dramatic impact in many applied domains.

\section*{Acknowledgements}
The authors would like to thank Surya Tokdar for helpful discussions on the proof of prior support for the proposed covariance regression formulation.

\appendix

\section{Proofs of Theorems and Lemmas}


\begin{proof}[Proof of Theorem~\ref{thm:priorsupport}]
Since $\mathcal{X}$ is compact, for every $\epsilon_0>0$ there exists an open covering of $\epsilon_0$-balls $B_{\epsilon_0}(x_0) = \{x : ||x-x_0||_2 < \epsilon_0\}$ with a finite subcover such that $\bigcup_{x_0 \in \mathcal{X}_0} B_{\epsilon_0}(x_0) \supset \mathcal{X}$, where $|\mathcal{X}_0|=n$.  Then,
\begin{align}
	\PSigma \left( \sup_{x\in\mathcal{X}} ||\Sigma(x)-\Sigma^*(x)||_2 < \epsilon \right) = \PSigma \left( \max_{x_0 \in \mathcal{X}_0}\sup_{x\in B_{\epsilon_0}(x_0)} ||\Sigma(x)-\Sigma^*(x)||_2 < \epsilon \right).
\end{align}
Define $Z(x_0) = \sup_{x\in B_{\epsilon_0}(x_0)} ||\Sigma(x)-\Sigma^*(x)||_2$. Since
\begin{align}
	\PSigma \left( \max_{x_0 \in \mathcal{X_0}} Z(x_0) < \epsilon \right) > 0 \iff \PSigma \left( Z(x_0) < \epsilon \right) > 0, \, \forall x_0 \in \mathcal{X}_0,
\end{align}
we only need to look at each $\epsilon_0$-ball independently, which we do as follows.
\begin{align}
	\PSigma \left(\sup_{x\in B_{\epsilon_0}(x_0)} ||\Sigma(x)-\Sigma^*(x)||_2 < \epsilon \right)\nonumber\\
	&\hspace{-2.25in}\geq
	  \PSigma \left( ||\Sigma(x_0)-\Sigma^*(x_0)||_2 + \hspace{-0.1in}\sup_{x\in B_{\epsilon_0}(x_0)} ||\Sigma^*(x_0)-\Sigma^*(x)||_2 + \hspace{-0.1in}\sup_{x\in B_{\epsilon_0}(x_0)} ||\Sigma(x_0)-\Sigma(x)||_2< \epsilon \right)\nonumber\\
	&\hspace{-2.25in}\geq \PSigma \left( ||\Sigma(x_0)-\Sigma^*(x_0)||_2 < \epsilon/3\right)\nonumber\\
	&\hspace{-2in} \cdot\PSigma\left(\sup_{x\in B_{\epsilon_0}(x_0)} ||\Sigma^*(x_0)-\Sigma^*(x)||_2 < \epsilon/3\right)\cdot\PSigma\left(\sup_{x\in B_{\epsilon_0}(x_0)} ||\Sigma(x_0)-\Sigma(x)||_2< \epsilon/3 \right),
	\label{eqn:3delta}
\end{align}
where the first inequality comes from repeated uses of the triangle inequality, and the second inequality follows from the fact that each of these terms is an independent event.  We evaluate each of these terms in turn.  The first follows directly from the assumed continuity of $\Sigma^*(\cdot)$.  The second will follow from a statement of (almost sure) continuity of $\Sigma(\cdot)$ that arises from the (almost sure) continuity of the $\xi_{\ell k}(\cdot)\sim \mbox{GP}(0,c)$ and the shrinkage prior on $\theta_{\ell k}$ (i.e., $\theta_{\ell k}\rightarrow 0$ almost surely as $\ell \rightarrow \infty$, and does so ``fast enough''.)  Finally, the third will follow from the support of the conditionally Wishart prior on $\Sigma(x_0)$ at every fixed $x_0 \in \mathcal{X}$.

Based on the continuity of $\Sigma^*(\cdot)$, for all $\epsilon/3 > 0$ there exists an $\epsilon_{0,1}>0$ such that
\begin{align}
	||\Sigma^*(x_0)-\Sigma^*(x)||_2 < \epsilon/3, \quad \forall ||x-x_0||_2 < \epsilon_{0,1}.
\end{align}
Therefore, $\PSigma\left(\sup_{x\in B_{\epsilon_{0,1}}(x_0)} ||\Sigma^*(x_0)-\Sigma^*(x)||_2 < \epsilon/3\right)=1$.

Based on Theorem~\ref{thm:continuity}, each element of $\Lambda(\cdot)\triangleq \Theta\xi(\cdot)$ is almost surely continuous on $\mathcal{X}$ assuming $k$ finite.  Letting $g_{jk}(x) = [\Lambda(x)]_{jk}$, 
\begin{align}
[\Lambda(x)\Lambda(x)']_{ij} = \sum_{m=1}^k g_{im}(x)g_{jm}(x), \quad \forall x\in \mathcal{X}.
\label{eqn:GammaGamma}
\end{align}
Eq.~\eqref{eqn:GammaGamma} represents a finite sum over pairwise products of almost surely continuous functions, and thus results in a matrix $\Lambda(x)\Lambda(x)'$ with elements that are almost surely continuous on $\mathcal{X}$. Therefore, $\Sigma(x)=\Lambda(x)\Lambda(x)' + \Sigma_0 = \Theta\xi(x)\xi(x)'\Theta' + \Sigma_0$ is almost surely continuous on $\mathcal{X}$. We can then conclude that for all $\epsilon/3 > 0$ there exists an $\epsilon_{0,2} > 0$ such that
\begin{align}
	\PSigma\left(\sup_{x\in B_{\epsilon_{0,2}}(x_0)} ||\Sigma(x_0)-\Sigma(x)||_2 < \epsilon/3 \right) = 1.
\end{align}

To examine the third term, we first note that
\begin{multline}
	\PSigma \left( ||\Sigma(x_0)-\Sigma^*(x_0)||_2 < \epsilon/3\right)\\ = \PSigma \left( ||\Theta\xi(x_0)\xi(x_0)'\Theta' + \Sigma_0-\Theta^*\xi^*(x_0)\xi^*(x_0)'\Theta^{*'} + \Sigma^*_0||_2 < \epsilon/3\right),
\end{multline}
where $\{\xi^*(x_0),\Theta^*,\Sigma^*_0\}$ is any element of $\mathcal{X}_{\xi} \otimes \mathcal{X}_\Theta \otimes \mathcal{X}_{\Sigma_0}$ such that $\Sigma^*(x_0) = \Theta^*\xi^*(x_0)\xi^*(x_0)'\Theta^{*'} + \Sigma^*_0$.  Such a factorization exists by Lemma~\ref{lemma:factorization}. We can then bound this prior probability by
\begin{align}
	\PSigma \left( ||\Sigma(x_0)-\Sigma^*(x_0)||_2 < \epsilon/3\right) &\nonumber\\
	&\hspace{-2in}\geq \PSigma \left( ||\Theta\xi(x_0)\xi(x_0)'\Theta' -\Theta^*\xi^*(x_0)\xi^*(x_0)'\Theta^{*'}||_2 < \epsilon/6\right)\Pi_{\Sigma_0}\left( ||\Sigma_0-\Sigma^*_0||_2 < \epsilon/6\right)\nonumber\\
	&\hspace{-2in}\geq \PSigma \left( ||\Theta\xi(x_0)\xi(x_0)'\Theta' -\Theta^*\xi^*(x_0)\xi^*(x_0)'\Theta^{*'}||_2 < \epsilon/6\right)\Pi_{\Sigma_0}\left( ||\Sigma_0-\Sigma^*_0||_{\infty} < \epsilon/(6\sqrt{\pi})\right),
\end{align}
where the first inequality follows from the triangle inequality, and the second from the fact that for all $A \in \Re^{p\times p}$, $||A||_2 \leq \sqrt{p}||A||_{\infty}$, with the sup-norm defined as $||A||_{\infty} = \max_{1\leq i\leq p}\sum_{i=1}^p |a_{ij}|$.  Since $\Sigma_0 = diag(\sigma_1^2,\dots,\sigma_p^2)$ with $\sigma_i^2 \stackrel{i.i.d.}{\sim} \mbox{Ga}(a_\sigma,b_\sigma)$, the support of the gamma prior implies that
\begin{align}
	\Pi_{\Sigma_0}\left( ||\Sigma_0-\Sigma^*_0||_{\infty} < \epsilon/(6\sqrt{\pi})\right) = \Pi_{\Sigma_0}\left(\max_{1\leq i \leq p} |\sigma_i^2-\sigma_i^{*2}| < \epsilon/(6\sqrt{\pi}) \right) > 0.
\end{align}
Recalling that $[\xi(x_0)]_{\ell k} = \xi_{\ell k}(x_0)$ with $\xi_{\ell k}(x_0) \stackrel{i.i.d.}{\sim} \mathcal{N}(0,1)$ and taking $\Theta \in \Re^{p\times L}$ with $\mbox{rank}(\Theta)=p$,
%
\begin{align}
	\Theta\xi(x_0)\xi(x_0)'\Theta' \mid \Theta \sim \mbox{W}(k,\Theta\Theta').
\end{align}
When $k \geq p$, $\Theta\xi(x_0)\xi(x_0)'\Theta'$ is invertible (i.e., full rank) with probability 1.

By Assumption~\ref{assumption:rankTheta}, there is positive probability under $\Pi_{\Theta}$ on the set of $\Theta$ such that $\mbox{rank}(\Theta)=p$.  Since $\Theta^*\xi^*(x_0)\xi^*(x_0)'\Theta^{*'}$ is an arbitrary symmetric positive semidefinite matrix in $\Re^{p\times p}$, and based on the support of the Wishart distribution,
\begin{align}
	\PSigma\left(||\Theta\xi(x_0)\xi(x_0)'\Theta' -\Theta^*\xi^*(x_0)\xi^*(x_0)'\Theta^{*'}||_2 < \epsilon/6\right) > 0.
\end{align}
We thus conclude that $\PSigma \left( ||\Sigma(x_0)-\Sigma^*(x_0)||_2 < \epsilon/3\right) > 0$.

For every $\Sigma^*(\cdot)$ and $\epsilon> 0$, let $\epsilon_0 = \min(\epsilon_{0,1}, \epsilon_{0,2})$ with $\epsilon_{0,1}$ and $\epsilon_{0,2}$ defined as above. Then, combining the positivity results of each of the three terms in Eq.~\eqref{eqn:3delta} completes the proof.
\end{proof}

\begin{proof}[Proof of Theorem~\ref{thm:continuity}]
	We can represent each element of $\Lambda(\cdot)$ as follows:
	\begin{align}
		[\Lambda(\cdot)]_{jk} &= \lim_{L \rightarrow \infty}\begin{bmatrix}
			\begin{bmatrix}
				\theta_{11} & \theta_{12} & \dots & \theta_{1L}\\
				\theta_{21} & \theta_{22} & \dots & \theta_{2L}\\
				\vdots & \vdots & \ddots & \vdots \\
				\theta_{p1} & \theta_{p2} & \dots & \theta_{pL}
			\end{bmatrix}
			\begin{bmatrix}
				\xi_{11}(\cdot) & \xi_{12}(\cdot) & \dots & \xi_{1k}(\cdot)\\
				\xi_{21}(\cdot) & \xi_{22}(\cdot) & \dots & \xi_{2k}(\cdot)\\
				\vdots & \vdots & \ddots & \vdots \\
				\xi_{L1}(\cdot) & \xi_{L2}(\cdot) & \dots & \xi_{Lk}(\cdot)
			\end{bmatrix}
		\end{bmatrix}_{jk}
		= \sum_{\ell = 1}^\infty \theta_{j\ell}\xi_{\ell k}(\cdot).
	\end{align}	
	If $\xi_{\ell k}(x)$ is continuous for all $\ell, k$ and $s_n(x) = \sum_{\ell = 1}^\infty \theta_{j\ell}\xi_{\ell k}(x)$ uniformly converges almost surely to some $g_{jk}(x)$, then $g_{jk}(x)$ is almost surely continuous. That is, if for all $\epsilon > 0$ there exists an $N$ such that for all $n \geq N$
	\begin{align}
		\mbox{Pr}\left(\sup_{x \in \mathcal{X}} |g_{jk}(x) - s_n(x)| < \epsilon \right) = 1,
	\end{align}
	then $s_n(x)$ converges uniformly almost surely to $g_{jk}(x)$ and we can conclude that $g_{jk}(x)$ is continuous based on the definition of $s_n(x)$. To show almost sure uniform convergence, it is sufficient to show that there exists an $M_n$ with $\sum_{n=1}^\infty M_n$ almost surely convergent and
	\begin{align}
		\sup_{x\in \mathcal{X}} |\theta_{jn}\xi_{nk}(x)| \leq M_n.
	\end{align}
	Let $c_{nk} = \sup_{x\in \mathcal{X}}|\xi_{nk}(x)|$. Then,
	\begin{align}
		\sup_{x\in \mathcal{X}} |\theta_{jn}\xi_{nk}(x)| \leq |\theta_{jn}|c_{nk}.
	\end{align}
	Since $\xi_{nk}(\cdot) \stackrel{i.i.d.}{\sim} \mbox{GP}(0,c)$ and $\mathcal{X}$ is compact, $c_{nk}<\infty$ and $E[c_{nk}] = \bar{c}$ with $\bar{c}$ finite.  Defining $M_n =  |\theta_{jn}|c_{nk}$,
	\begin{align}
		E_{\Theta,c}\left[\sum_{n=1}^\infty M_n\right] &= E_{\Theta}\left[E_{c\mid \Theta}\left[\sum_{n=1}^\infty |\theta_{jn}|c_{nk} \mid \Theta\right]\right]
		= E_{\Theta}\left[\sum_{n=1}^\infty |\theta_{jn}|\bar{c}\right] = \bar{c}\sum_{n=1}^\infty E_{\Theta}\left[|\theta_{jn}|\right],
	\end{align}
	where the last equality follows from Fubini's theorem.  Based on Assumption~\ref{assumption:absSum}, we conclude that $E[\sum_{n=1}^\infty M_n]<\infty$ which implies that $\sum_{n=1}^\infty M_n$ converges almost surely.
\end{proof}

\begin{proof}[Proof of Lemma~\ref{lemma:absSum}]
	Recall that $\theta_{j\ell} \sim \mathcal{N}(0,\phi_{j\ell}^{-1}\tau_{\ell}^{-1})$ with $\phi_{j\ell} \sim \mbox{Ga}(3/2,3/2)$ and $\tau_{\ell} = \prod_{h=1}^{\ell}\delta_h$ for $\delta_1 \sim \mbox{Ga}(a_1,1), \, \delta_h \sim \mbox{Ga}(a_2,1)$.  Using the fact that if $x \sim \mathcal{N}(0,\sigma)$ then $E[|x|] = \sigma\sqrt{2/\pi}$ and if $y \sim \mbox{Ga}(a,b)$ then $1/y \sim \mbox{Inv-Ga}(a,1/b)$ with $E[1/y] = 1/(b\cdot(a-1))$, we derive that
	\begin{align}
		\sum_{\ell=1}^\infty E_{\theta}[|\theta_{j\ell}|] &= \sum_{\ell=1}^\infty E_{\phi,\tau}[E_{\theta \mid \phi,\tau}[|\theta_{j\ell}|\mid \phi_{j\ell},\tau_{\ell} ]]
		=  \sqrt{\frac{2}{\pi}}\sum_{\ell=1}^\infty E_{\phi,\tau}[\phi_{j\ell}^{-1}\tau_{\ell}^{-1}]\nonumber\\
		&=  \sqrt{\frac{2}{\pi}}\sum_{\ell=1}^\infty E_{\phi}[\phi_{j\ell}^{-1}]E_{\tau}[\tau_{\ell}^{-1}] 
		= \frac{4}{3}\sqrt{\frac{2}{\pi}}\sum_{\ell=1}^\infty E_{\delta}\left[\prod_{h=1}^{\ell} \frac{1}{\delta_h}\right] 
		= \frac{1}{a_1-1}\frac{4}{3}\sqrt{\frac{2}{\pi}}\sum_{\ell=1}^\infty \frac{1}{a_2-1}^{\ell-1}.
	\end{align}
	When $a_2 > 2$, we conclude that $\sum_{\ell} E[|\theta_{j\ell}|] < \infty$.
\end{proof}

\begin{proof}[Proof of Lemma~\ref{lemma:firstMoment}]
Recall that $\Sigma(x) = \Theta\xi(x)\xi(x)'\Theta' + \Sigma_0$ with $\Sigma_0 = \mbox{diag}(\sigma_1^2,\dots,\sigma_p^2)$.  The elements of the respective matrices are independently distributed as $\theta_{i\ell} \sim \mathcal{N}(0,\phi_{i\ell}^{-1}\tau_{\ell}^{-1})$, $\xi_{\ell k}(\cdot) \sim \mbox{GP}(0,c)$, and $\sigma_{i}^{-2} \sim \mbox{Gamma}(a_{\sigma},b_{\sigma})$.  Let $\mu_{\sigma}$ and $\sigma_{\sigma}^2$ represent the mean and variance of the implied inverse gamma prior on $\sigma_i^2$, respectively.  In all of the following, we first condition on $\Theta$ and then use iterated expectations to find the marginal moments.

The expected covariance matrix at any predictor location $x$ is simply derived as
\begin{align*}
	E[\Sigma(x)] &= E[E[\Sigma(x) \mid \Theta]] = E[E[\Theta\xi(x)\xi(x)'\Theta' \mid \Theta]] + \mu_{\sigma}I_p\\
				 &= kE[\Theta\Theta'] + \mu_{\sigma}I_p\\
				 &= \mbox{diag}(k\sum_{\ell}\phi_{1\ell}^{-1}\tau_{\ell}^{-1} + \mu_{\sigma},\dots,k\sum_{\ell}\phi_{p\ell}^{-1}\tau_{\ell}^{-1} + \mu_{\sigma}).
\end{align*}
Here, we have used the fact that conditioned on $\Theta$, $\Theta\xi(x)\xi(x)'\Theta'$ is Wishart distributed with mean $k\Theta\Theta'$ and 
\begin{align*}
	E[\Theta\Theta']_{ij} &= \sum_{\ell}\sum_{\ell'} E[\theta_{i\ell}\theta_{j\ell'}] = \sum_{\ell} E[\theta_{i\ell}^2] \delta_{ij}\\
						  &= \sum_{\ell}\mbox{var}(\theta_{i\ell})\delta_{ij} = \left(\sum_{\ell}\phi_{i\ell}^{-1}\tau_{\ell}^{-1}\right)\delta_{ij}.
\end{align*}
\end{proof}

\begin{proof}[Proof of Lemma~\ref{lemma:secondMoment}]
One can use the conditionally Wishart distribution of $\Theta\xi(x)\xi(x)'\Theta'$ to derive $\mbox{cov}(\Sigma_{ij}(x),\Sigma_{uv}(x))$.  Specifically, let $S = \Theta\xi(x)\xi(x)'\Theta'$.  Then $S = \sum_{n=1}^k z^{(n)}z^{(n)'}$ with $z^{(n)}\mid \Theta \sim \mathcal{N}(0,\Theta\Theta')$ independently for each $n$.  Then, using standard Gaussian second and fourth moment results,
\begin{align*}
	\mbox{cov}(\Sigma_{ij}(x),\Sigma_{uv}(x) \mid \Theta) &= \mbox{cov}(S_{ij},S_{uv} \mid \Theta) + \sigma_{\sigma}^2\delta_{ijuv}\\
				&= \sum_{n=1}^k E[z_i^{(n)}z_j^{(n)}z_u^{(n)}z_v^{(n)} \mid \Theta] - E[z_i^{(n)}z_j^{(n)}\mid \Theta]E[z_u^{(n)}z_v^{(n)}\mid \Theta] + \sigma_{\sigma}^2\delta_{ijuv}\\
				&=k((\Theta\Theta')_{iu}(\Theta\Theta')_{jv}+(\Theta\Theta')_{iv}(\Theta\Theta')_{ju})+\sigma_{\sigma}^2\delta_{ijuv}.
\end{align*}
Here, $\delta_{ijuv}=1$ if $i=j=u=v$ and is $0$ otherwise.  Taking the expectation with respect to $\Theta$ yields $\mbox{cov}(\Sigma_{ij}(x),\Sigma_{uv}(x))$.  However, instead of looking at one slice of the predictor space, what we are really interested in is how the correlation between elements of the covariance matrix changes with predictors.  Thus, we work directly with the latent Gaussian processes to derive $\mbox{cov}(\Sigma_{ij}(x),\Sigma_{uv}(x'))$.  Let
\begin{align}
	g_{in}(x) = \sum_{\ell}\theta_{i\ell}\xi_{\ell n}(x),
	\label{eqn:g}
\end{align}
implying that $g_{in}(x)$ is independent of all $g_{im}(x')$ for any $m\neq n$ and all $x' \in \mathcal{X}$. Since each $\xi_{\ell n}(\cdot)$ is distributed according to a zero mean Gaussian process, $g_{in}(x)$ is zero mean. Using this definition, we condition on $\Theta$ (which is dropped in the derivations for notational simplicity) and write
\begin{align*}
	\mbox{cov}(\Sigma_{ij}(x),\Sigma_{uv}(x') \mid \Theta) &= \sum_{n=1}^k \mbox{cov}(g_{in}(x)g_{jn}(x),g_{un}(x'),g_{vn}(x')) + \sigma_{\sigma}^2\delta_{ijuv}\\
		&\hspace{-1in}= \sum_{n=1}^k E[g_{in}(x)g_{jn}(x)g_{un}(x'),g_{vn}(x')] - E[g_{in}(x)g_{jn}(x)]E[g_{un}(x'),g_{vn}(x')] + \sigma_{\sigma}^2\delta_{ijuv}\\
\end{align*}
We replace each $g_{kn}(x)$ by the form in Equation~\eqref{eqn:g}, summing over different dummy indices for each.  Using the fact that $\xi_{\ell n}(x)$ is independent of $\xi_{\ell'n}(x')$ for any $\ell\neq \ell'$ and that each $\xi_{\ell n}(x)$ is zero mean, all cross terms in the resulting products cancel if a $\xi_{\ell n}(x)$ arising from one $g_{kn}(x)$ does not share an index $\ell$ with at least one other $\xi_{\ell n}(x)$ arising from another $g_{pn}(x)$.  Thus,
\begin{align*}
	\mbox{cov}(\Sigma_{ij}(x),\Sigma_{uv}(x') \mid \Theta) &= \sum_{n=1}^k \sum_{\ell}\theta_{i\ell}\theta_{j\ell}\theta_{u\ell}\theta_{v\ell} E[\xi_{\ell n}^2(x)\xi_{\ell n}^2(x')] \\
	&\hspace{0.25in}+ \sum_{\ell}\theta_{i\ell}\theta_{u\ell}E[\xi_{\ell n}(x)\xi_{\ell n}(x')] \sum_{\ell'\neq \ell}\theta_{j\ell'}\theta_{v\ell'}E[\xi_{\ell' n}(x)\xi_{\ell' n}(x')]\\
	&\hspace{0.25in}+ \sum_{\ell}\theta_{i\ell}\theta_{j\ell}E[\xi_{\ell n}^2(x)] \sum_{\ell'\neq \ell}\theta_{u\ell'}\theta_{v\ell'}E[\xi_{\ell' n}^2(x')]\\
	&\hspace{0.25in}-  \sum_{\ell}\theta_{i\ell}\theta_{j\ell}E[\xi_{\ell n}^2(x)] \sum_{\ell'}\theta_{u\ell'}\theta_{v\ell'}E[\xi_{\ell' n}^2(x')]
	 + \sigma_{\sigma}^2\delta_{ijuv}
\end{align*}
The Gaussian process moments are given by
\begin{align*}
	E[\xi_{\ell n}^2(x)] &= 1\\
	E[\xi_{\ell n}(x)\xi_{\ell n}(x')] &= E[E[\xi_{\ell n}(x)\mid \xi_{\ell n}(x')]\xi_{\ell n}(x')]=c(x,x')E[\xi_{\ell n}^2(x')]=c(x,x')\\
	E[\xi_{\ell n}^2(x)\xi_{\ell n}^2(x')] &=E[E[\xi_{\ell n}^2(x)\mid \xi_{\ell n}(x')]\xi_{\ell n}^2(x')] \\
				&= E[\{(E[\xi_{\ell n}(x)\mid \xi_{\ell n}(x')])^2 + \mbox{var}(\xi_{\ell n}(x)\mid \xi_{\ell n}(x'))\}\xi_{\ell n}^2(x')] \\
				&= c^2(x,x')E[\xi_{\ell n}^4(x')] + (1-c^2(x,x'))E[\xi_{\ell n}^2(x')] = 2c^2(x,x')+1,
\end{align*}
from which we derive that
\begin{align*}
	\mbox{cov}(\Sigma_{ij}(x),\Sigma_{uv}(x') \mid \Theta) &= k\bigg\{ (2c^2(x,x') + 1)\sum_{\ell}\theta_{i\ell}\theta_{j\ell}\theta_{u\ell}\theta_{v\ell}  + c^2(x,x')\sum_{\ell}\theta_{i\ell}\theta_{u\ell}\sum_{\ell'\neq \ell}\theta_{j\ell'}\theta_{v\ell'}\\
	&\hspace{0.35in}+ \sum_{\ell}\theta_{i\ell}\theta_{j\ell}\sum_{\ell'\neq \ell}\theta_{u\ell'}\theta_{v\ell'}
	-  \sum_{\ell}\theta_{i\ell}\theta_{j\ell}\sum_{\ell'}\theta_{u\ell'}\theta_{v\ell'}\bigg\}
	 + \sigma_{\sigma}^2\delta_{ijuv}\\
	&= kc^2(x,x') \left\{\sum_{\ell}\theta_{i\ell}\theta_{j\ell}\theta_{u\ell}\theta_{v\ell} + \sum_{\ell}\theta_{i\ell}\theta_{u\ell}\sum_{\ell'}\theta_{j\ell'}\theta_{v\ell'}\right\}
	 + \sigma_{\sigma}^2\delta_{ijuv}.
\end{align*}

An iterated expectation with respect to $\Theta$ yields the following results.  When $i\neq u$ or $j \neq v$, the independence between $\theta_{i\ell}$ (or $\theta_{j \ell}$) and the set of other $\theta_{k\ell}$ implies that $\mbox{cov}(\Sigma_{ij}(x),\Sigma_{uv}(x'))=0$.  When $i=u$ and $j=v$, but $i\neq j$,
\begin{align*}
	\mbox{cov}(\Sigma_{ij}(x),\Sigma_{ij}(x')) &= kc^2(x,x') \left\{\sum_{\ell}E[\theta^2_{i\ell}]E[\theta^2_{j\ell}] + \sum_{\ell}E[\theta^2_{i\ell}]\sum_{\ell'}E[\theta^2_{j\ell'}]\right\}\\
	&= kc^2(x,x') \left\{\sum_{\ell}\phi_{i\ell}^{-1}\phi_{j\ell}^{-1}\tau_{\ell}^{-2} + \sum_{\ell}\phi_{i\ell}^{-1}\tau_{\ell}^{-1}\sum_{\ell'}\phi_{j\ell'}^{-1}\tau_{\ell'}^{-1}\right\}.
\end{align*}
Finally, when $i=j=u=v$,
\begin{align*}
	\mbox{cov}(\Sigma_{ii}(x),\Sigma_{ii}(x')) &= kc^2(x,x') \left\{2\sum_{\ell}E[\theta^4_{i\ell}] + \sum_{\ell}E[\theta^2_{i\ell}]\sum_{\ell'\neq \ell}E[\theta^2_{i\ell'}]\right\} + \sigma_{\sigma}^2\\
	&= kc^2(x,x') \left\{6\sum_{\ell}\phi_{i\ell}^{-2}\tau_{\ell}^{-2} + \sum_{\ell}\phi_{i\ell}^{-1}\tau_{\ell}^{-1}\sum_{\ell'\neq \ell}\phi_{i\ell'}^{-1}\tau_{\ell'}^{-1}\right\} + \sigma_{\sigma}^2\\
	&= kc^2(x,x') \left\{5\sum_{\ell}\phi_{i\ell}^{-2}\tau_{\ell}^{-2} + (\sum_{\ell}\phi_{i\ell}^{-1}\tau_{\ell}^{-1})^2\right\} + \sigma_{\sigma}^2.
\end{align*}
\end{proof}

\section{Derivation of Gibbs Sampler}

In this Appendix, we derive the conditional distribution for sampling the Gaussian process dictionary elements.  Combining Eq.~\eqref{eq:base} and Eq.~\eqref{eq:Lamx}, we have that 
\begin{align}
	y_i = \Theta\begin{bmatrix} 
		\xi_{11}(x_i) & \xi_{12}(x_i) & \dots & \xi_{1k}(x_i)\\
		\xi_{21}(x_i) & \xi_{22}(x_i) & \dots & \xi_{2k}(x_i)\\
		\vdots & \vdots & \ddots & \vdots\\
		\xi_{L1}(x_i) & \xi_{L2}(x_i) & \dots & \xi_{Lk}(x_i)
	\end{bmatrix}\eta_i + \epsilon_i = 
	\Theta\begin{bmatrix} \sum_{m=1}^k \xi_{1m}(x_i)\eta_{im} \\ \vdots \\ \sum_{m=1}^k \xi_{Lm}(x_i)\eta_{Lm} \end{bmatrix} + \epsilon_i
	\label{eqn:y2}
\end{align}
implying that
\begin{align}
	y_{ij} = \sum_{\ell=1}^L\sum_{m=1}^k \theta_{j\ell}\eta_{im}\xi_{\ell m}(x_i) + \epsilon_{ij}.
\end{align}
Conditioning on $\xi(\cdot)^{-\ell m}$, we rewrite Eq.~\eqref{eqn:y2} as
\begin{align}
	y_i = \eta_{im}\begin{bmatrix} \theta_{1\ell} \\ \vdots \\ \theta_{p\ell} \end{bmatrix}\xi_{\ell m}(x_i) + \tilde{\epsilon}_i, \quad
	\tilde{\epsilon}_i \sim \mathcal{N}\left(\mu_{\ell m}(x_i) \triangleq \begin{bmatrix} 
	\sum_{(r,s)\neq(\ell,m)}\theta_{1r}\eta_{is}\xi_{rs}(x_i)\\ 
	\vdots\\ 
	\sum_{(r,s)\neq(\ell,m)} \theta_{pr}\xi_{rs}(x_i) \end{bmatrix}, 
	\Sigma_0 \right).
\end{align}
Let $\theta_{\cdot \ell} = \begin{bmatrix} \theta_{1\ell} & \dots & \theta_{p\ell} \end{bmatrix}'$. Then,
\begin{align}
	\begin{bmatrix} y_1 \\ \vdots \\ y_n \end{bmatrix} = 
		\begin{bmatrix} \eta_{1m}\theta_{\cdot \ell} & 0 & \dots & 0\\ 
			0 & \eta_{2m}\theta_{\cdot \ell} & \dots & 0 \\
			\vdots & \vdots & \ddots & \vdots\\
			0 & 0 & \dots & \eta_{nm}\theta_{\cdot \ell} 
		\end{bmatrix}
		\begin{bmatrix} \xi_{\ell m}(x_1) \\ \xi_{\ell m}(x_2) \\ \vdots \\ \xi_{\ell m}(x_n) \end{bmatrix} + 
		\begin{bmatrix} \tilde{\epsilon}_1 \\ \tilde{\epsilon}_2 \\ \vdots \\ \tilde{\epsilon}_n \end{bmatrix}
\end{align}
Defining $A_{\ell m} = \mbox{diag}(\eta_{1m}\theta_{\cdot \ell}, \dots, \eta_{nm}\theta_{\cdot \ell})$, our Gaussian process prior on the dictionary elements $\xi_{\ell m}(\cdot)$ implies the following conditional posterior
\begin{align}
	\begin{bmatrix} \xi_{\ell m}(x_1) \\ \xi_{\ell m}(x_2) \\ \vdots \\ \xi_{\ell m}(x_n) \end{bmatrix} \mid \{y_i\},\Theta,\eta,\xi(\cdot),\Sigma_0 &\sim 
		\mathcal{N}\left(\tilde{\Sigma}A_{\ell m}'\begin{bmatrix}\Sigma_0^{-1} & & \\ & \ddots & \\ & & \Sigma_0^{-1} \end{bmatrix}
		\begin{bmatrix} \tilde{y}_1 \\ \vdots \\ \tilde{y}_n \end{bmatrix},\tilde{\Sigma} \right)\nonumber\\
		&= \mathcal{N}\left(\tilde{\Sigma}\begin{bmatrix} \eta_{1m}\sum_{j=1}^p \theta_{j\ell}\sigma_j^{-2}\tilde{y}_{1j} \\ \vdots \\ \eta_{nm}\sum_{j=1}^p \theta_{j\ell}\sigma_j^{-2}\tilde{y}_{nj} \end{bmatrix},\tilde{\Sigma} \right),
\end{align} 
where $\tilde{y}_i = y_i - \mu_{\ell m}(x_i)$ and, taking $K$ to be the matrix of correlations $K_{ij} = c(x_i,x_j)$ defined by the Gaussian process parameter $\kappa$,
\begin{align}
	\tilde{\Sigma}^{-1} &= K^{-1} + A_{\ell m}'\begin{bmatrix}\Sigma_0^{-1} & & \\ & \ddots & \\ & & \Sigma_0^{-1} \end{bmatrix} A_{\ell m} 
	= K^{-1} + \mbox{diag}\left(\eta_{1m}^2\sum_{j=1}^p \theta_{j\ell}^2\sigma_j^{-2},\dots,\eta_{nm}^2\sum_{j=1}^p\theta_{j\ell}^2\sigma_j^{-2}\right).	
\end{align}

\bibliographystyle{plainnat}
\bibliography{../../Bibliography/Bibliography}

\end{document}